%% file: LFCS2020_JV_Main.tex
\documentclass[a4paper, fleqn]{article}

\usepackage{authblk}
\usepackage{epigraph}
\usepackage{lettrine}
\usepackage{type1ec}
\usepackage{latexsym}
\usepackage{amsmath}
\usepackage{amssymb}
\usepackage{amsthm}
\usepackage{setspace}
\usepackage{afterpage}
\usepackage{newlfont}
\usepackage{array}
\usepackage{tabularx}
\usepackage{dcolumn}
\usepackage{graphicx}
\usepackage{titlesec}
\usepackage{booktabs}
\usepackage[bf]{caption}
\usepackage{tikz}
\usepackage{bussproofs}
\usepackage{proof}
\usepackage{rotating}
\usepackage{multirow}
\usepackage{mathrsfs}
\usepackage{synttree}
\usepackage{forest}
\usepackage{stmaryrd}
\usepackage[show]{ed}
\usepackage{hyperref}

\usepackage[lined,linesnumbered,algoruled]{algorithm2e}

\theoremstyle{definition}
\newtheorem{theorem}{Theorem}[section]
\theoremstyle{definition}
\newtheorem{lemma}[theorem]{Lemma}
\theoremstyle{definition}

\theoremstyle{definition}
\newtheorem{proposition}[theorem]{Proposition}
\theoremstyle{definition}
\newtheorem{corollary}[theorem]{Corollary}
\theoremstyle{definition}
\newtheorem{definition}{Definition}[section]
\theoremstyle{definition}
\newtheorem{example}{Example}[section]
\theoremstyle{definition}

\newcommand{\hyp}{\mid}

\newcommand{\mc}{\mathcal}
\newcommand{\mf}{\mathsf}

\newcommand{\vart}{\mathsf{t}}

\newcommand{\ufor}{\Vd^{\forall}}
\newcommand{\efor}{\Vd^{\exists}}

\newcommand{\nbr}{\vartriangleright}

\newcommand{\p}[1]{{#1}}
\newcommand{\n}[1]{\ov{#1}}

\newcommand{\Wcompl}[1]{\W\setminus{#1}}

\newcommand{\leftrule}[1]{$\mathsf{L#1}$}
\newcommand{\rightrule}[1]{$\mathsf{R#1}$}

\newcommand{\lufor}{\leftrule{\ufor}}
\newcommand{\rufor}{\rightrule{\ufor}}
\newcommand{\lefor}{\leftrule{\efor}}
\newcommand{\refor}{\rightrule{\efor}}

\newcommand{\init}{$\mf{init}$}
\newcommand{\lbot}{\leftrule{\bot}}
\newcommand{\rtop}{\rightrule{\top}}

\newcommand{\lbox}{\leftrule{\Box}}
\newcommand{\rbox}{\rightrule{\Box}}

\newcommand{\lneg}{\leftrule{\neg}}
\newcommand{\rneg}{\rightrule{\neg}}
\newcommand{\lto}{\leftrule{\to}}
\newcommand{\rto}{\rightrule{\to}}
\newcommand{\lland}{\leftrule{\land}}
\newcommand{\rland}{\rightrule{\land}}
\newcommand{\llor}{\leftrule{\lor}}
\newcommand{\rlor}{\rightrule{\lor}}

\newcommand{\labelrule}[1]{\textsf{#1}}
\newcommand{\ruleM}{\labelrule M}
\newcommand{\ruleC}{\labelrule C}
\newcommand{\ruleNtau}{\labelrule N}
\newcommand{\ruleNovtau}{$\ov\tau^\0$}
\newcommand{\ruledec}{\labelrule{dec}}
\newcommand{\ruleovdec}{$\overline{\labelrule{dec}}$}

\newcommand{\hyplwk}{\hyprule{Lwk}}
\newcommand{\hyprwk}{\hyprule{Rwk}}
\newcommand{\hyplctr}{\hyprule{Lctr}}
\newcommand{\hyprctr}{\hyprule{Rctr}}
\newcommand{\hypcut}{\hyprule{cut}}
\newcommand{\hyprule}[1]{\textsf{#1}}

\newcommand{\ltrset}{\llbracket}
\newcommand{\rtrset}{\rrbracket}
\newcommand{\ltr}{\llbracket}
\newcommand{\rtr}{\rrbracket}

\newcommand{\ov}{\overline}

\newcommand{\G}{\Gamma}
\newcommand{\D}{\Delta}

\newcommand{\N}{\mathcal N}

\newcommand{\W}{\mathcal W}
\newcommand{\V}{\mathcal V}
\newcommand{\C}{\mathcal C}

\newcommand{\M}{\mathcal M}

\newcommand{\lan}{\mathcal L}

\newcommand{\atm}{Atm}

\newcommand{\s}{\mathcal S}
\newcommand{\logic}{\logicnamestyle L}

\newcommand{\pow}{\mathcal P}

\newcommand{\Mb}{\M_\bi}
\newcommand{\Mn}{\M_\st}
\newcommand{\Nb}{\N_\bi}
\newcommand{\Nn}{\N_\st}
\newcommand{\Mr}{\M_r}
\newcommand{\Wimp}{\W^i}
\newcommand{\Wnorm}{\W\setminus\W^i}

\newcommand{\0}{\emptyset}

\newcommand{\logicnamestyle}[1]{\textbf{#1}}
\newcommand{\Estar}{\logicnamestyle{E}^*}

\newcommand{\E}{\logicnamestyle{E}}
\newcommand{\EN}{\logicnamestyle{EN}}
\newcommand{\EC}{\logicnamestyle{EC}}
\newcommand{\ECN}{\logicnamestyle{ECN}}
\newcommand{\EM}{\logicnamestyle{M}}
\newcommand{\EMN}{\logicnamestyle{MN}}
\newcommand{\EMC}{\logicnamestyle{MC}}
\newcommand{\EMCN}{\logicnamestyle{MCN}}
\newcommand{\EP}{\logicnamestyle{EP}}
\newcommand{\ENP}{\logicnamestyle{ENP}}
\newcommand{\ECP}{\logicnamestyle{ECP}}
\newcommand{\ECNP}{\logicnamestyle{ECNP}}
\newcommand{\EMP}{\logicnamestyle{MP}}
\newcommand{\EMNP}{\logicnamestyle{MNP}}
\newcommand{\EMCP}{\logicnamestyle{MCP}}

\newcommand{\ED}{\logicnamestyle{ED}}
\newcommand{\END}{\logicnamestyle{END}}
\newcommand{\ECD}{\logicnamestyle{ECD}}
\newcommand{\ECND}{\logicnamestyle{ECND}}
\newcommand{\EMD}{\logicnamestyle{MD}}
\newcommand{\EMND}{\logicnamestyle{MND}}
\newcommand{\EMCD}{\logicnamestyle{MCD}}

\newcommand{\EDoneplus}{\logicnamestyle{ED}_1^+}
\newcommand{\EDtwoplus}{\logicnamestyle{ED}_2^+}

\newcommand{\EMDtwoplus}{\logicnamestyle{MD}_2^+}

\newcommand{\EDnplus}{\logicnamestyle{ED}_n^+}

\newcommand{\ECDnplus}{\logicnamestyle{ECD}_n^+}

\newcommand{\EMDnplus}{\logicnamestyle{MD}_n^+}
\newcommand{\EMNDnplus}{\logicnamestyle{MND}_n^+}
\newcommand{\EMCDnplus}{\logicnamestyle{MCD}_n^+}

\newcommand{\EPD}{\logicnamestyle{EPD}}

\newcommand{\EMstar}{\logicnamestyle{M}^*}
\newcommand{\ECstar}{\logicnamestyle{EC}^*}
\newcommand{\ENstar}{\logicnamestyle{EN}^*}
\newcommand{\ETstar}{\logicnamestyle{ET}^*}
\newcommand{\EDstar}{\logicnamestyle{ED}^*}
\newcommand{\EPstar}{\logicnamestyle{EP}^*}

\newcommand{\EMCstar}{\logicnamestyle{MC}^*}
\newcommand{\EMNstar}{\logicnamestyle{MN}^*}
\newcommand{\EDnplusstar}{\hEstar{D_n^+}}
\newcommand{\EMTstar}{\logicnamestyle{MT}^*}
\newcommand{\EMDstar}{\logicnamestyle{MD}^*}
\newcommand{\EMPstar}{\logicnamestyle{MP}^*}
\newcommand{\EMDnplusstar}{\hMstar{D_n^+}}

\newcommand{\K}{\logicnamestyle{K}}

\newcommand{\KD}{\logicnamestyle{KD}}
\newcommand{\KP}{\logicnamestyle{KP}}
\newcommand{\KDnplus}{\logicnamestyle{KD}_n^+}

\newcommand{\CPL}{\logicnamestyle{CPL}}

\newcommand{\MCstar}{\logicnamestyle{MC}^*}

\newcommand{\EX}{\logicnamestyle{EX}}
\newcommand{\MX}{\logicnamestyle{MX}}

\newcommand{\Lstar}{\logic^*}

\newcommand{\hEstar}[1]{\mathbf{E{#1}^*}}
\newcommand{\hMstar}[1]{\mathbf{M{#1}^*}}

\newcommand{\hilbertaxiomstyle}[1]{${#1}$}

\newcommand{\RE}{\hilbertaxiomstyle{RE}}
\newcommand{\RM}{\hilbertaxiomstyle{RM}}
\newcommand{\RN}{\hilbertaxiomstyle{RN}}

\newcommand{\MP}{\hilbertaxiomstyle{MP}}

\newcommand{\axM}{\hilbertaxiomstyle{M}}
\newcommand{\axC}{\hilbertaxiomstyle{C}}
\newcommand{\axN}{\hilbertaxiomstyle{N}}
\newcommand{\axK}{\hilbertaxiomstyle{K}}
\newcommand{\axT}{\hilbertaxiomstyle{T}}
\newcommand{\axD}{\hilbertaxiomstyle{D}}
\newcommand{\axP}{\hilbertaxiomstyle{P}}
\newcommand{\axfour}{\hilbertaxiomstyle{4}}

\newcommand{\axX}{\hilbertaxiomstyle{X}}

\newcommand{\RDnplus}{\hilbertaxiomstyle{RD}$_n^+$}

\newcommand{\ax}{\AxiomC}
\newcommand{\uinf}{\UnaryInfC}
\newcommand{\binf}{\BinaryInfC}
\newcommand{\llab}{\LeftLabel}
\newcommand{\rlab}{\RightLabel}
\newcommand{\disp}{\DisplayProof}

\newcommand{\bin}{bi-neighbourhood}

\newcommand{\st}{{st}}
\newcommand{\bi}{{bi}}

\newcommand{\CstEstar}{\C^{\st}_{\Estar}}

\newcommand{\Cstlogic}{\C^{\st}_{\logic}}
\newcommand{\Cblogic}{\C^{\bi}_{\logic}}

\newcommand{\semcond}[1]{{#1}}
\newcommand{\cM}{\semcond M}

\newcommand{\cN}{\semcond N}

\newcommand{\cC}{\semcond C}
\newcommand{\cT}{\semcond T}
\newcommand{\cP}{\semcond P}
\newcommand{\cD}{\semcond D}
\newcommand{\cRDnplus}{\semcond RD$_n^+$}

\newcommand{\cX}{\semcond X}

\newcommand{\set}[1]{\mathsf{set}(#1)}

\newcommand{\R}{\mathcal R}

\newcommand{\w}{wg}
\newcommand{\seq}{\Rightarrow}

\newcommand{\Nat}{\mathbb{N}}

\newcommand{\AND}{\bigwedge}

\newcommand{\vd}{\vdash}
\newcommand{\Vd}{\Vdash}
\newcommand{\eg}{e.g.}
\newcommand{\ie}{i.e.}
\newcommand{\ih}{inductive hypothesis}

\newcommand{\hp}{height-preserving}

\newcommand{\varR}{\mathscr R}

\newcommand{\angA}{\langle A \rangle}
\newcommand{\angB}{\langle B \rangle}

\newcommand{\angSigma}{\langle \Sigma \rangle}

\newcommand{\str}[1]{\langle #1 \rangle}
\newcommand{\fint}{i}
\newcommand{\hG}{G}
\newcommand{\hH}{H}

\newcommand{\rboxm}{$\mathsf{R\Box m}$}

\newcommand{\rulen}{$\mf N$}

\newcommand{\rulec}{$\mf C$}
\newcommand{\rulep}{$\mf P$}
\newcommand{\rulet}{$\mf T$}
\newcommand{\ruled}{$\mf D_2$}
\newcommand{\ruledaux}{$\mf D_1$}
\newcommand{\ruledm}{$\mf D_{\mathsf M}$}
\newcommand{\rulednplus}{$\mf D_n^+$}
\newcommand{\ruledmplus}{$\mf D_m^+$}
\newcommand{\rulediplus}{$\mf D_i^+$}

\newcommand{\ruledoneplus}{$\mf D_1^+$}
\newcommand{\ruledtwoplus}{$\mf D_2^+$}
\newcommand{\hypruleT}{\hyprule T}
\newcommand{\hypruleP}{\hyprule P}

\newcommand{\hypcalc}[1]{\mathbf{H.#1}}

\newcommand{\HL}{\hypcalc L}
\newcommand{\HE}{\hypcalc E}

\newcommand{\HEC}{\hypcalc{EC}}
\newcommand{\HMC}{\hypcalc{MC}}

\newcommand{\HMCN}{\hypcalc{MCN}}
\newcommand{\HMCNT}{\hypcalc{MCNT}}

\newcommand{\HEstar}{\hypcalc E^*}
\newcommand{\HEXstar}{\hypcalc{EX}^*}

\newcommand{\HMCstar}{\hypcalc{MC}^*}

\newcommand{\wk}{$\mathsf{wk}$}
  
\newcommand{\lwk}{$\mathsf{Lwk}$}
\newcommand{\rwk}{$\mathsf{Rwk}$}
\newcommand{\lctr}{$\mathsf{Lctr}$}

\newcommand{\angctr}{\lctr}

\newcommand{\exctr}{\hyprule{Ectr}}    
\newcommand{\exwk}{\hyprule{Ewk}}
\newcommand{\subrule}{$\mathsf{sub}$}
\newcommand{\subrulem}{$\mathsf{sub_M}$}

\newcommand{\lanLS}{\mc L_{\mathsf{LS}}}
\newcommand{\hl}[3]{\left[#1\right]^{#2}_{#3}}
\newcommand{\oneEs}{\mathbf{LSE^*}}
\newcommand{\Lbox}{\mathsf{L}\Box}
\newcommand{\Rbox}{\mathsf{R}\Box}

\newcommand{\MRbox}{\Rbox\mathsf{m}}
\newcommand{\iset}[2]{\{#1\}_{#2}}
\newcommand{\Lufor}{\mathsf{L}\ufor}
\newcommand{\Rufor}{\mathsf{R}\ufor}
\newcommand{\Lefor}{\mathsf{L}\efor}
\newcommand{\Refor}{\mathsf{R}\efor}
\newcommand{\dec}{\mathsf{dec}}
\newcommand{\ovdec}{\ov{\dec}}
\newcommand{\RLM}{\mathsf{M}}
\newcommand{\RLC}{\mathsf{C}}
\newcommand{\RLN}{\mathsf{N}}

\newcommand{\Lwedge}{\mathsf{L}\wedge}

\newcommand{\Init}{\mf{init}}
\newcommand {\hide}[1]{}

\newcommand{\size}[1]{\mathsf{size}(#1)}

\title{Hypersequent calculi for non-normal modal and deontic
	logics: Countermodels and optimal complexity\thanks {This work has been partially supported by the ANR project TICAMORE ANR-16-
CE91-0002-01. Lellmann has been supported by WWTF project MA16-28. Pimentel has been partially supported by CNPq.}}
\author[1]{Tiziano Dalmonte}
\author[2]{Bj\"{o}rn Lellmann}
\author[1]{Nicola Olivetti}
\author[3]{Elaine Pimentel}

\affil[1]{Aix Marseille Univ, Universit\'{e} de Toulon, CNRS, LIS,
Marseille, France\\
\{tiziano.dalmonte,nicola.olivetti\}@lis-lab.fr}
\affil[2]{
bj.lellmann@gmail.com}
\affil[3]{Departamento de Matem\'{a}tica, UFRN, Natal, Brazil\\
elaine.pimentel@gmail.com}

\begin{document}
\maketitle

\begin{abstract}
  We present some hypersequent calculi for all systems of the
  classical cube and their extensions with axioms \axT, \axP, \axD,
  and, for every $n \geq 1$, rule \RDnplus.  The calculi are internal
  as they only employ the language of the logic, plus additional
  structural connectives.  We show that the calculi are complete with
  respect to the corresponding axiomatisation by a syntactic proof of
  cut elimination.  Then we define a terminating root-first proof
  search strategy based on the hypersequent calculi and show that it
  is optimal for \textsf{coNP}-complete logics.  Moreover, we obtain
  that from every saturated leaf of a failed proof it is possible to
  define a countermodel of the root hypersequent in the \bin\
  semantics, and for regular logics also in the relational
  semantics. We finish the paper by giving a translation between
  hypersequent {\em rule applications} and {\em derivations} in a
  labelled system for the classical cube.

\end{abstract}

\input{introduction-JV}
\input{axiomatisation-NNMLs-JV}
\input{standard-semantics-JV}
\input{hypersequent-calculi-JV}
\input{complexity}
\input{countermodel}
\input{translations}
\input{conclusion-hypersequent-calculi-JV}

\bibliographystyle{abbrv}
\bibliography{references}
\end{document}

%% file: introduction-JV.tex

\section{Introduction}
Non-normal modal logics--NNMLs for short--have a long history, going
back to the seminal works by Kripke, Montague, Segeberg, Scott, and
Chellas (see \cite{Chellas:1980} for an introduction). They are
``non-normal'' as they do not contain all axioms of minimal normal
modal logic \K. NNMLs find an interest in several areas, from
epistemic to deontic reasoning. They also play a r\^{o}le in
multi-agent reasoning and strategic reasoning in games.  For instance
in epistemic reasoning they offer a simple (although partial) solution
to the problem of logical omniscience (see \cite{Vardi:1986}); in
deontic logic, they allow avoiding well-known paradoxes (such as
Ross's Paradox) and to represent conflicting obligations
(see~\cite{Goble:2013}); multi-agent logics with non-normal modalities
have been proposed to capture agency and ability: $\Box A$ is read as
the agent can bring about $A$ (see~\cite{Elgesem:1997}); a related
interpretation is the game-theoretical interpretation of $\Box A$ as
``the agent has a winning strategy to bring about $A$'' (indeed,
non-normal monotonic logic \EM{} can be seen as a 2-agent case of
coalition logic with determinacy~\cite{DBLP:journals/logcom/Pauly02}).
Finally, NNMLs are needed also when $\Box A$ is interpreted as ``$A$
is true in most of the cases''~\cite{Askounis}.
 
In this work, we consider the {\em classical cube} of NMMLs, given by
the extensions of minimal modal logic \E, containing only the {\em
  congruence rule} \RE, with axioms \axC, \axM\ and \axN.  We also
consider extensions with axioms/rules T, D, P, and D$_n^+$, where T is
the {\em reflexivity} axiom in classical normal modal logic, and the
others axioms are significant in deontic logic. More precisely,
reading $\Box A$ as ``it is obligatory that $A$'', D is the
characteristic axiom of deontic logic
$\neg (\Box A \land \Box \neg A)$, expressing that something and its
negation cannot at the same time be obligatory; and P is the axiom
$\neg \Box \bot$, expressing that something impossible cannot be
obligatory. Although the axioms P and D are equivalent in normal modal
logic, this does not hold in the non-normal setting.  The system EMNP
is considered as a meaningful {\em minimal system of Deontic
  Logic}~\cite{Goble:2013,Orlandelli:2014}.  Finally, although the
rules D$_n^+$ have never been considered ``officially'' in the
literature, (but see \cite{Goble:2013} and \cite{Hansson3}), they
properly generalise the axiom D for systems without C, expressing that
there cannot be $n$ incompatible obligations: if
$\neg (A_1 \land \ldots \land A_n)$ then
$\neg(\Box A_1 \land \ldots \land \Box A_n)$.

NNMLs have a well-understood semantics defined in terms of
neighbourhood models \cite{Chellas:1980,Pacuit:2017}: in these models
each world $w$ has an associated set of neighbourhoods $\N(w)$, each
one of them being a set of worlds/states.  If we accept the
traditional interpretation of a ``proposition'' as a set of worlds (=
its truth set), we can think of each neighbourhood in $\N(w)$ as the
proposition: a formula $\Box A$ is true in a world $w$ if ``the
proposition'' $A$, \ie\ the truth-set of $A$, belongs to $\N(w)$.  The
classical cube and all mentioned extensions can be modelled by
imposing additional closure properties of the set of neighbourhoods.

In this work we adopt a variant of the neighbourhood semantics defined
in terms of bi-neighbourhood models~\cite{Dalmonte:2018}: in these
structures each world has associated a set of \emph{pairs} of
neighbourhoods. The intuition is that the two components of a pair
provide positive and negative support for a modal formula. This
variant is significant and more natural for ``non-monotonic'' logics
(\ie\ not containing axiom \axM). The reason is that, instead of
specifying exactly the truth sets in $\N(w)$, the pairs of
neighbourhoods specify \emph{lower} and \emph{upper} bounds of truth
sets, so that the same pair may be a ``witness'' for several
propositions.  For this reason, the generation of
\emph{countermodels}, one of the goals of the present work, is easier
in the bi-neighbourhood semantics than in the standard one.
Bi-neighbourhood models can be transformed into standard ones and
vice-versa.

The proof-theory of NNMLs is not quite as developed as their
semantics, apart from early works on regular modal logics
like~\cite{Fitting:1983book}. In particular, it is curious to note
that, although some proof-systems for NNMLs have been proposed in the
past, countermodel generation has been rarely addressed and
computational properties of proof systems are seldom analysed.
Indeed, the
works~\cite{Lavendhomme:2000,Gilbert,Negri:2017,Dalmonte:2018} propose
countermodel extraction in the neighbourhood semantics, but all of
them require either a complicated procedure or an extended language
with labels.  The recent~\cite{DBLP:conf/tableaux/Lellmann19} presents
a nested sequent calculus for a logic combining normal and monotone
non-normal modal logic that supports countermodel extraction. However
the nested sequent structure is not suitable for logics lacking
monotonicity.  In contrast, cut-free sequent/linear nested calculi for
the classical cube and its extensions with standard axioms of normal
modal logics (the non-normal counterpart of logics from \K\ to
$\mathbf{S5}$), including deontic axioms D and P, are studied in
\cite{Indr:2005,Indr:2011,Lellmann:2019,Orlandelli:2019}.  In
particular, \cite{Orlandelli:2019} focuses on cut-free sequent calculi
on calculi for deontic logic, partially covering the family of systems
defined in this paper.  However, neither semantic completeness, nor
countermodel extraction, nor complexity are studied in the mentioned
papers.

In this work, we intend to fill this gap by proposing \emph{modular}
calculi for the classical cube and the mentioned deontic extensions
that provide {\em direct countermodel extraction} and are of {\em
  optimal complexity}.  Our calculi are semantically based on
bi-neighbourhood models, and have two syntactic features: they
manipulate hypersequents and sequents may contain blocks of $\Box$-ed
formulas in the antecedent.  A hypersequent~\cite{Avron} is just a
multiset of sequents and can be understood as a (meta-logical)
disjunction of sequents. Sequents within hypersequents can be read as
formulas of the logic and, for this reason, our calculi are ``almost''
internal.  Blocks of formulas are interpreted as conjunctions of
negative $\Box$-ed formulas.  Intuitively, each block represents a
neighbourhood satisfying one or more $\Box$-ed formulas, and this
allows for the formulation of modular calculi. It is worth noticing
that the calculi have also good proof-theoretical properties, as they
support a syntactic proof of cut admissibility.

We make clear that, for the purpose of having sound and complete
calculi for NNMLs, neither hypersequents, nor blocks are necessary, as
for instance the sequent calculi in \cite{Lavendhomme:2000,
  Orlandelli:2019, Indr:2005, Indr:2011} show.  But as we shall see,
the hypersequent framework is very adequate to extract countermodels
from a single failed proof, ensuring at the same time good
computational and structural properties.  As a matter of fact, even in
the bi-neighbourhood semantics, non-normal modal logics, in particular
without monotonicity, ultimately need to consider truth sets of
formulas. Hence, in order to obtain calculi suitable for a reasonably
straightforward countermodel construction, we need to be able to
represent essentially all worlds of a possible model in the data
structure used by the calculus.  While this could also be accomplished
by other types of calculi, for obtaining \emph{small} countermodels in
non-monotonic logics it is crucial that every world (represented by a
component of the hypersequent) has access to {\em all} other worlds
which have been constructed so far.  This very strongly suggests a
{\em flat} structure, as given by hypersequents, in contrast for
instance, with the tree-like structure of nested sequents.

A further advantage of using hypersequents is that all rules become
invertible, thus there is no need for backtracking in proof
search. For the same reason, the hypersequent calculi provide
\emph{direct} countermodel extraction: from \emph{one} failed proof we
can \emph{directly} extract a countermodel in the bi-neighbourhood
semantics of the sequent/formula at the root of the derivation.  A
particular case is the one of regular logics, i.e., logics containing
both M and C (whence normal modal logic K as well). These systems
admit a relational semantics. We show how to extract a relational
countermodel from a failed proof-search in the calculi for these
logics as well.

We also consider the problem of obtaining optimal decision
procedures. The known complexity bounds are not the same for logics
without and with axiom \axC{}. Namely (see~\cite{Vardi:1986}), the
former are \textsf{coNP}, whereas the latter are \textsf{PSPACE}, a
fact that also follows by a general result on non-iterative modal
logics~\cite{Schroeder:2008}. For logics with \axC{} (belonging to our
cube), a \textsf{PSPACE} decision procedure can be obtained by
standard proof-search in sequent calculi, like those ones
in~\cite{Lavendhomme:2000}. Therefore we concentrate on the more
significant case of logics \emph{without} \axC{}: it turns out that
for these logics our calculi provide an \emph{optimal} \textsf{coNP}
decision procedure. For logics including \axC{}, we can still obtain
an optimal \textsf{PSPACE} decision procedure by adopting an {\em
  unkleene'd} version of our calculi, which ``sacrifices'' the
invertibility of rules.

We finish this work by presenting a formal translation between
hypersequent calculi, restricted to the classical cube, and the
labelled calculi $\oneEs$, presented in~\cite{Dalmonte:2018}.  As
mentioned above, our calculi have an internal flavour, since sequents
have an interpretation within the logics (although hypersequents do
not). Labelled systems, on the other hand, are intrinsically external
due to the use of symbols that are not in the base logical language,
in the form of labels.  Establishing translations between sequent
based systems and labelled systems is often a hard
task~\cite{simpson94phd,DBLP:conf/tableaux/CiabattoniMS13,DBLP:conf/aiml/GirlandoON18}. Indeed,
hyper/nested sequents typically carry the semantical information
within their structure, while labels explicitly mark semantical
behaviours to formulas.  The results presented in this work shows that
our hypersequent calculi provide a compact encoding of derivations in
the labelled framework.

All in all, we believe that the structure of our calculi, namely  
hypersequents with blocks,  is adequate for NNMLs from a semantical, computational and a proof-theoretical point of view since it:
 (i) has a semantic interpretation;
(ii) allows direct countermodel generation; (iii) supports optimal
complexity decision procedures;
(iv) has good proof-theoretical properties; and (v) has a natural translation to labelled systems present in the literature.

This article is a thoroughly revised and significantly extended
version of the conference
paper~\cite{DBLP:conf/lfcs/DalmonteLOP20}. Some of the most
significant extensions with respect to that work are the modular
extension of all the results to twice the amount of axioms, the
extension to the relational semantics for regular logics, and the
investigation of the formal relation with the calculi presented
in~\cite{Dalmonte:2018}, in the form of mutual simulation.

The plan of the paper is as follows: In Section 2 we introduce the
logical systems considered in this work. In Section 3, we present both
standard neighbourhood semantics and its bi-neighbourhood variant.  In
Section 4, we introduce hypersequent calculi and we prove the main
proof-theoretical properties, including cut-admissibility, from which
also follows their syntactic completeness.  In Section 5, we show how
the calculi can provide a decision procedure for the respective logics
and we analyse their complexity.  In Section 6, we show how the
calculi can be used to extract directly countermodels from failed
proofs, which is one of the main goals of this work; additionally,
this directly yields semantic completeness.  Finally in Section 7 we
explore the relation of hypersequent calculi with previously
introduced labelled calculi for the classical cube of NNML, while
Section 8 contains some final discussion and conclusions.


%% file: axiomatisation-NNMLs-JV.tex

\section{Non-normal modal logics as axiom systems}
\label{section:classical logics axioms}

In this section we introduce, axiomatically, the class of non-normal
modal logics we consider in this work.
 
\begin{definition} {\em Non-normal modal logics} are defined over a
  propositional modal language $\lan$, based on a set
  $\atm = \{p_1, p_2, p_3, ...\}$ of countably many propositional
  variables.  The well-formed \emph{formulas} of $\lan$ are defined by
  the following grammar
  \begin{center}
    $A ::= p \mid \bot \mid \top \mid A\land A \mid A \lor A \mid A\to A \mid \Box A$.
  \end{center}
\end{definition}

\noindent
In the following, we use $A, B, C, D, E$ and $p, q, r$ as
metavariables for, respectively, arbitrary formulas and atoms of
$\lan$.  We consider the standard abbreviations
$\neg A := A \to \bot$, $A\leftrightarrow B := (A\to B)\land(B\to A)$,
and $\Diamond A := \neg\Box\neg A$.

\begin{definition}\label{def:axioms}
The family of non-normal modal logics  is generated by
\begin{enumerate}
\item  any axiomatization of  classical propositional logic (\CPL) 
  formulated in the language $\lan$, comprising the rule 
  of  \emph{modus ponens} (\MP)\label{item:axiomspc}
  \begin{center}	
    \ax{$A \rightarrow B$}
    \ax{$A$}
    \llab{\MP}
    \binf{$B$}
    \disp
  \end{center}
\item The rule RE of Figure  \ref{figure:hilbert axioms}.\label{item:re}	
\item Any or none of the other axioms or  rules of Figure \ref{figure:hilbert axioms}.
\end{enumerate}
The minimal non-normal modal logic is $\E$, defined by only
items~\ref{item:axiomspc} and~\ref{item:re} above. We denote
non-normal modal logics by $\EX$, where $\mathbf X$ stands for the
(possibly empty) additional set of axioms and rules from Figure
\ref{figure:hilbert axioms}.  We adopt the convention of replacing E
with M for systems containing axiom M, which are consequently denoted
by $\MX$. We also drop the ``R'' of rule \RDnplus. E.g., we write
$\mathbf{MD_3^+}$ for the logic given by extending $\E$ with the axiom
\axM{} and rule $RD_3^+$. In addition, given a non-normal modal logic
$\logic$, we will write $\Lstar$ to indicate an arbitrary extension of
$\logic$ with some other axioms.

As usual, we say that a formula $A$ is {\em provable} in $\logic$
(denoted by $ \vd_{\logic} A$) if it is an instance of an axiom of
$\logic$ or it is obtained from previous formulas by applying the
rules of $\logic$.
\end{definition}
\begin{figure}[t]
  \centering
  \fbox{\begin{minipage}{33em}
      \vspace{0.1cm}

      \begin{tabular}{lllllll}
        \vspace{0.2cm}
        &\ax{$A \leftrightarrow B$}
          \llab{\RE}
          \uinf{$\Box A \leftrightarrow \Box B$}
          \disp
        &&
           \multicolumn{4}{l}{
           \ax{$\neg(A_1 \land ... \land A_n)$}
           \llab{\RDnplus}
           \rlab{$(n \geq 1)$}
           \uinf{$\neg(\Box A_1 \land ... \land \Box A_n)$}
           \disp} \\

        \vspace{0.2cm}
        &\axM \quad $\Box (A \land B)\to \Box A$ &&
                                                    \axC \quad $\Box A \land \Box B \to \Box (A \land B)$ &&
                                                                                                             \axN \quad $\Box \top$ \\

        &\axT \quad $\Box A \to A$ &&
                                      \axD \quad $\neg(\Box A \land \Box \neg A)$ &&
                                                                                     \axP \quad $\neg\Box \bot$ \\
      \end{tabular}
    \end{minipage}}
  \caption{Modal axioms and rules.}\label{figure:hilbert axioms}
\end{figure}

Logic $\E$ is the weakest system of the so-called \emph{classical
  cube} \cite{Chellas:1980,Lellmann:2019}, generated by any
combination of axioms \axM, \axC, and \axN, as shown in Figure
\ref{figure:classical cube}.

As it is well known, the axioms \axM{} and \axN{} are respectively equivalent to  the rules
of \emph{monotonicity} \RM\ and \emph{necessitation} \RN
$$RM \ \frac{A\to B}{\Box A \to \Box B} \quad\quad RN \ \frac{A}{\Box A}$$
Moreover, axiom \axK:
$\Box( A\rightarrow B)\rightarrow \Box A\rightarrow \Box B$ is
derivable in $\EMC$.  As a consequence, the top system $\EMCN$
coincides with the minimal normal modal logic $\K$.

In the following, we call \emph{monotonic}\index{monotonic logic} any
system containing axiom \axM\ (and \emph{non-monotonic} otherwise), we
call \emph{regular}\index{regular logic} any system containing both
axioms \axM{} and \axC, and we call \emph{normal}\index{normal logic}
any system containing \axM, \axC, and \axN{}.

\begin{figure}
\qquad \hfill
\resizebox{4cm}{3cm}{
\begin{tikzpicture}
	\node (a) at  (0,0)  {\E};
    \node (b) at (0, 2.1) {\EM};
    \node  (c) at (-1.5, 0.7) {\EC};
    \node (d) at (2.3, 0.7) {\EN};
    \node (e) at (-1.5, 2.8) {\EMC};
    \node (f) at (2.3, 2.8) {\EMN};
    \node (g) at (0.8, 1.4) {\ECN};
    \node (h) at (0.8, 3.5) {\EMCN{} (\K)};

\draw (a) -- (b);
\draw (a) -- (c);
\draw (a) -- (d);
\draw (b) -- (e);
\draw (b) -- (f);
\draw (c) -- (e);
\draw [dashed] (c) -- (g);
\draw (d) -- (f);
\draw [dashed] (d) -- (g);
\draw (e) -- (h);
\draw (f) -- (h);
\draw [dashed] (g) -- (h);
\end{tikzpicture}
}
\hfill \qquad
\caption{\label{figure:classical cube} The classical cube.}
\end{figure}
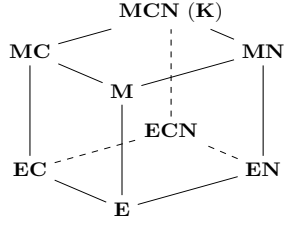
In the present work, we will in particular consider extensions of the
basic classical cube with any combination of the axioms \axT, \axD,
\axP, and for every $n \in \Nat$, $n \geq 1$, the rule \RDnplus\
(Figure~\ref{figure:classical cube}).

The logics given by the rules \RDnplus\ have a peculiar interest in
deontic logic since the latter exclude the possiblity of having $n$
obligations that cannot be realised all together.  Let us consider the
following example, essentially from Hansson~\cite[p. 41]{Hansson3}:
(1) I have to keep my mobile switched on (as I'm waiting for an urgent
call), (2) I have to attend my child schoolplay, (3) being in the
audience of a schoolplay I must keep my mobile switched
off. Representing these three claims by $mobile\_on$, $schoolplay$,
and $\neg (mobile\_on \land schoolplay)$, by using rule $RD_3^+$, it
can be concluded that the three obligations are incompatible:
$$\neg (\Box mobile\_on \land \Box schoolplay \land  \Box \neg (mobile\_on \land schoolplay))$$
This conclusion cannot be obtained in any non-normal modal logic
without $RD_3^+$ or \axC{}, even if it contains both deontic axioms 
\axD\ and \axP. 

While rules \RDnplus\ are entailed by axioms \axD\ and \axP\ in normal
modal logics, this is not the case in non-normal modal logics,
therefore they may be considered explicitly.  It is also worth noting
that the axioms \axD\ and \axP\ are equivalent in normal modal logics,
but not in non-normal ones.

The relations among the systems defined by adding \axD, \axP, and
\RDnplus\ to the systems of the classical cube are displayed in
Figure~\ref{figure:deontic systems}.  Finally, observe also that all
\axD, \axP, and \RDnplus\ are entailed by axiom \axT.

\begin{figure}[t]
\begin{center}
\resizebox{10cm}{9cm}{ 
\begin{small}
\begin{tikzpicture}
    \node (E) at  (0,0)  {\E};
    \node (EM) at (0, 4.7) {\EM};
    \node  (EC) at (-3, 1.5) {\EC};
    \node (EN) at (4.9, 1.5) {\EN};
    \node (EMC) at (-3, 6.1) {\EMC};
    \node (EMN) at (4.9, 6.1) {\EMN};
    \node (ECN) at (1.9, 3) {\ECN};
    \node (EMCN) at (1.9, 7.5) {\EMCN}; 
    \node  (EP) at (-1.1, 1.1) {$\EP \equiv \EDoneplus$};
    \node  (ED) at (2.5, 1.3) {\ED};
    \node  (EDtwo) at (1.1, 3.5) {$\EDtwoplus$};
    \node  (EDn) at (1.1, 4.6) {$\EDnplus$};
    \node  (MD) at (2.5, 6.5) {$\EMD \equiv \EMDtwoplus$};
    \node  (END) at (3.7, 4) {\END};
    \node  (ENP) at (6, 2.5) {\ENP};
    \node  (ECD) at (-4.1, 2.5) {\ECD};
    \node  (ECP) at (-4.1, 7.9) {$\ECP \equiv \ECDnplus$};
    \node  (ECND) at (-4.1, 8.9) {$\ECND \equiv \ECNP$};
    \node  (EPD) at (1.1, 2.1) {\EPD};
    \node  (MCD) at (-1.4, 9.4) {$\EMCD \equiv \EMCP \equiv \EMCDnplus$};
    \node (MP) at (-1.3, 5.7) {\EMP};
    \node  (MND) at (3.7, 8.6) {\EMND};
    \node  (MCND) at (0.7, 11.8) {$\KD \equiv \KP \equiv \KDnplus$}; 
    \node (MNP) at (6, 7.1) {\EMNP};
    \node (midECD-ECP) at (-2.88, 3.94) {};    
    \node  (MDn) at (0.2, 8.1) {$\EMDnplus$};     
    \node  (MNDn) at (3.7, 9.7) {$\EMNDnplus$}; 
\draw  (E) -- (EM);
\draw  (E) -- (EC);
\draw  (E) -- (EN);
\draw  (EM) -- (EMC);
\draw  (EM) -- (EMN);
\draw  (EC) -- (EMC);
\draw [dotted] (EC) -- (ECN);
\draw  (EN) -- (EMN);
\draw [dotted] (EN) -- (ECN);
\draw [dotted] (EMC) -- (EMCN);
\draw [dotted] (EMN) -- (EMCN);
\draw [dotted] (ECN) -- (EMCN);
   \draw [dashed] (E) -- (EP);
   \draw [dashed] (E) -- (ED);
   \draw [dashed] (EDtwo) -- (EPD);
   \draw [dashed] (EDtwo) -- (EDn);
   \draw [dashed] (ED) -- (MD);
   \draw [dashed] (ED) -- (END);
   \draw [dashed] (EN) -- (END);
   \draw [dashed] (EP) -- (ECP);
   \draw  (EC) -- (ECP);
   \draw [dashed] (ED) -- (EPD);
   \draw [dashed] (ED) -- (ECD);
   \draw [dashed] (EP) -- (EPD);
   \draw  (EC) -- (ECD);
   \draw  (EN) -- (ENP);
   \draw  (ENP) -- (MNP);
   \draw [dashed] (END) -- (ENP);
   \draw [dashed] (MD) -- (MDn);
   \draw [dashed] (MDn) -- (MCD);
   \draw  (MP) -- (MD);
   \draw  (EM) -- (MP);
   \draw [dashed] (EP) -- (MP);
   \draw [dashed] (MD) -- (EDtwo);
   \draw [dashed] (END) -- (MND);
   \draw  (EMN) -- (MND);
   \draw  (MD) -- (MND);
   \draw  (MND) -- (MNDn);
   \draw  (MNDn) -- (MCND);
   \draw [dashed] (EMCN) -- (MCND);
   \draw  (MCD) -- (MCND);
   \draw  (EMN) -- (MNP);
   \draw [dashed] (EDn) -- (ECP);
   \draw  (MP) -- (MCD);
   \draw  (EMC) -- (MCD);
   \draw  (ECP) -- (MCD);
   \draw  (MNP) -- (MND);
   \draw (ECP) -- (ECD);
   \draw (ECP) -- (ECND);
   \path[-]          (ECND)  edge   [bend left]   node {} (MCND);
   \draw [dashed] (EP) -- (ENP);
   \draw [dashed] (END) -- (ECND);
   \draw [dashed] (EDn) -- (MDn);
   \draw [dashed] (MDn) -- (MNDn);
\end{tikzpicture}
\end{small}
}
\end{center}
\caption{\label{figure:deontic systems} Diagram of deontic systems (
``Pantheon'').}
\end{figure}
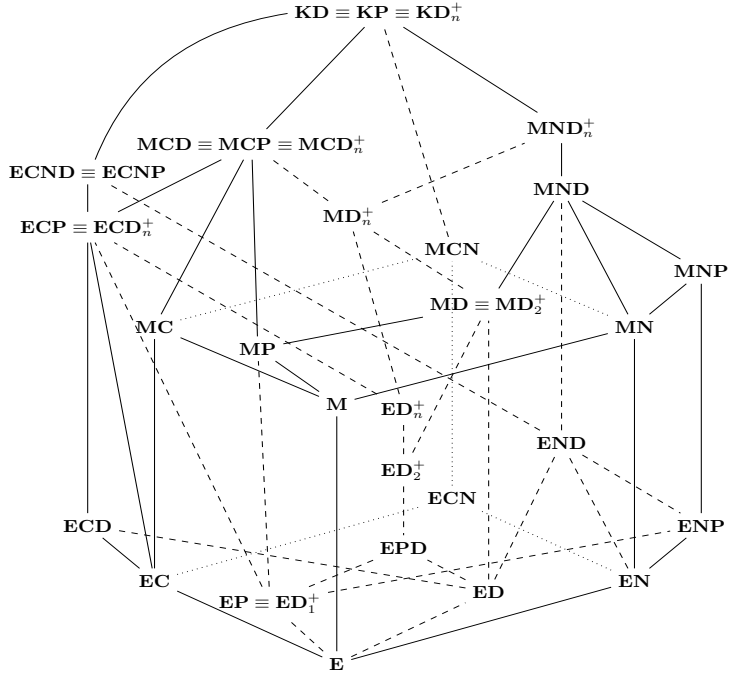


%% file: standard-semantics-JV.tex

\section{Semantics}
\label{section:standard semantics}

The standard semantic characterisation of non-normal modal logics is
given in terms of \emph{neighbourhood} models \cite{Pacuit:2017}, also
called \emph{minimal} \cite{Chellas:1980}, or \emph{Scott-Montague}.
In this work they will be called \emph{standard neighbourhood}, or
just \emph{standard} models.  Standard neighbourhood models are a
generalisation of Kripke models for normal modal logics, where the
binary relation is replaced by a so-called neighbourhood function,
which assigns to each world a set of sets of worlds.  Intuitively, the
neighbourhood function assigns to each world the propositions that are
necessary/obligatory/etc.~in it.  Standard neighbourhood models are
defined as follows.

\begin{definition}[Standard Neighbourhood Model]
\label{def:standard semantics}
A \emph{standard neighbourhood
  model} is a triple $\M = \langle \W,\N,\V \rangle$, where $\W$ is a
non-empty set, whose elements are called worlds, $\N$ is a function
$\W \longrightarrow \mathcal{PP}(\W)$, called neighbourhood function,
and $\V : \atm \longrightarrow \mathcal P(\W)$ is a valuation function
for propositional variables of $\mathcal L$.  The forcing relation
$\M, w \Vd_{\st} A$ is defined as follows, where
$\ltrset A\rtrset_{\M}$ denotes the set
$\{v \mid \M, v \Vd_{\st} A\}$, also called the \emph{truth set} of
$A$:

\begin{center}
\begin{tabular}{l l l}
$\M, w \Vd_{\st} p$ & iff & $w \in V(p)$; \\

$\M, w \not\Vd_{\st}\bot$; \\

$\M, w \Vd_{\st}\top$; \\

$\M, w \Vd_{\st} A \to B$ & iff & $\M, w \not\Vd_{\st} A$ or $\M, w \Vd_{\st} B$; \\

$\M, w \Vd_{\st} \Box A$ & iff & $\ltrset A\rtrset_{\M}\in \N(w)$.
\end{tabular}
\end{center}

\end{definition}
We adopt the standard definitions of validity.

\begin{definition}\label{def:validity}
  We say that a formula $A$ is \emph{valid} in a model
  $\M=\langle \W, \N, \V \rangle$, written $\M \models A$, if
  $\M, w \Vd A$ for all $w\in\W$.  We say that $A$ is \emph{valid} in
  a class of models $\C$, written $\C \models A$, if $\M\models A$ for
  all $\M\in\C$.
\end{definition}

The class of all standard models characterises the basic logic \E. For
the extensions of logic $\E$, we need to consider some additional
closure properties of the neighbourhood function, as specified in the
next definition.

\begin{definition}[Semantic conditions for Extensions]
  Given a standard neighbourhood model
  $\M = \langle \W,\N,\V \rangle$, We consider the following
  conditions on $\N$
  \begin{center}
    \begin{tabular}{l l l}
      (\axM) & If $\alpha\in \N(w)$ and $\alpha\subseteq\beta$, then $\beta\in\N(w)$. \\

      (\axC) & If $\alpha, \beta\in \N(w)$, then $\alpha\cap\beta \in \N(w)$. \\

      (\axN) & $\W \in \N(w)$. \\

      (\axT) & If $\alpha\in \N(w)$, then $w\in\alpha$. \\

      (\axP) & $\0 \notin \N(w)$.  \\

      (\axD) & If $\alpha\in \N(w)$, then $\Wcompl\alpha\notin \N(w)$. \\

      (\RDnplus) & If $\alpha_1, ..., \alpha_n\in \N(w)$, 
                   then $\alpha_1 \cap ... \cap \alpha_n \not= \emptyset$. \\

    \end{tabular}
  \end{center}
\end{definition}

\noindent
The properties (\axM), (\axC), (\axN) are respectively called
\emph{supplementation}, \emph{closure under intersection}, and
\emph{containing the unit}; accordingly, a standard model is
supplemented, closed under intersection, or contains the unit if it
satisfies the corresponding property.

In the following, for every logic $\logic$ we denote by $\Cstlogic$
the class of standard models for $\logic$.

\begin{theorem}\label{standard completeness}
  Logic $\Estar$ is sound and complete with respect to the
  corresponding standard neighbourhood models, that is:
  $\models_{\CstEstar} A$ if and only if $\vd_{\Estar} A$.
\end{theorem}       

A proof of this result can be found in \cite{Chellas:1980} for the
systems of the classical cube and their extensions with axioms \axT,
\axP, \axD.  The proof can be easily extended also to the rules
\RDnplus.

A special case is given by regular logics $\MCstar$ (i.e. possibly
lacking the necessitation axiom N) For these logics there exists a
relational semantics which goes back to Kripke himself: in
\cite{KripkeII} he introduces relational models with so-called
non-normal worlds, with the aim of characterising a family of Lewis'
and Lemmon's systems in which necessitation fails or is validated only
in a restricted form.  Here we consider a definition from Priest
\cite[Section~4.2]{Priest:2008}.

\begin{definition}
  \label{def:relational models classical nnmls}
  A \emph{relational model} with \emph{non-normal worlds} is a tuple
  $\M = \langle \W, \Wimp, \R, \V \rangle$, where $\W$ is a non-empty
  set of worlds, $\Wimp \subseteq \W$ is the set of non-normal worlds,
  $\R \subseteq \W \times \W$ is a binary relation, and
  $\V: \atm \longrightarrow \pow(\W)$ is a valuation function for
  propositional variables.  The forcing relation $w \Vd_r A$ is defined
  as in Definition~\ref{def:standard semantics}, except for boxed
  formulas, for which it is defined as follows:
  \begin{center}
    $w \Vd_r \Box A$ \quad iff \quad $w \notin\Wimp$ and for all
    $v \in \W$, if $w\R v$ then $v \Vd A$.
  \end{center}
\end{definition}

Observe that in impossible worlds all boxed formulas are falsified.
Validity is defined as usual: we say that a formula is valid in a
model if it is true in all worlds (no matter if they are normal or
non-normal).  \footnote{This is one of the \emph{two} definitions of
  validity considered by Priest \cite{Priest:2008}.}  It is easy to
verify that non-normal relational models validate axioms \axM\ and
\axC\ but do not validate axiom \axN.  Notice also that in case
$\Wimp$ is empty we have the standard definition of Kripke models for
normal modal logics.

A semantic characterisation of some extensions of logic $\EMC$ by
means of non-normal relational models can be given by considering the
usual frame properties of Kripke semantics.

\begin{theorem}[{\cite[pp.300]{Fitting:1983book}}]\label{th:completeness relational semantics}
  Every \emph{regular} logic $\MCstar$ is sound and complete with
  respect to the corresponding relational models.
\end{theorem}

A variant of the neighbourhood semantics, called \emph{\bin}
semantics, was introduced in \cite{Dalmonte:2018}.  Instead of a set
of neighbourhoods, worlds in bi-neighbourhood models are equipped with
a set of \emph{pairs} of neighbourhoods.  The intuition is that the
two components of a pair provide ``positive'' and ``negative'' support
for a modal formula.  As we shall see, a technical motivation to
consider bi-neighhbourhood models is that they are more suitable for
countermodel generation than standard ones, but they also have an
interest in their own.

\begin{definition}\label{bi-neighbourhood semantics}
  A \emph{bi-neighbourhood model} is a triple
  $\M= \langle \W, \N, \V\rangle$, where $\W$ is a non-empty set, $\V$
  is a valuation function and $\N$ is a function assigning to each
  world $w$ a subset of $\mathcal P(\W) \times\mathcal P(\W)$.  The
  forcing relation $\M,w\Vd_{bi} A$ is defined as in Definition
  \ref{def:standard semantics} except for the modality, for which the
  clause is:
  \begin{center}
    $\M,w\Vd_{bi} \Box A$ \quad iff \quad
    there is $(\alpha,\beta)\in \N(w)$ 
    s.t.~$\alpha\subseteq\llbracket A\rrbracket_{\M}\subseteq\Wcompl\beta$. 
  \end{center}
\end{definition}

The above definition introduces the general class of \bin\ models,
thus characterising the basic logic \E. For extensions we need to
consider the further conditions contained in next definition.

\begin{definition}[Bi-neighbourhood conditions for extensions]
  Given a bi-neigh-bourhood model $\M= \langle \W, \N, \V\rangle$, we
  consider the following conditions:
  \begin{center}
    \begin{tabular}{l l}
      (\cM) & If $(\alpha, \beta)\in\N(w)$, then $\beta=\emptyset$. \\

      (\cN) & There is $\alpha\subseteq\W$ such that for all $w\in\W$, $(\alpha,\emptyset)\in\N(w)$. \\
      
      (\cC) & If $(\alpha, \beta), (\gamma, \delta)\in \N(w)$, then $(\alpha\cap\gamma, \beta\cup\delta)\in \N(w)$. \\
      
      (\cT) & If $(\alpha, \beta)\in \N(w)$, then $w\in\alpha$. \\
      
      (\cP) & If $(\alpha, \beta)\in\N(w)$, then $\alpha\not=\emptyset$. \\
      
      (\cD) & If $(\alpha, \beta), (\gamma, \delta)\in \N(w)$, then $\alpha\cap\gamma\not=\emptyset$ or $\beta\cap\delta\not=\emptyset$. \\
      
      (\cRDnplus) & If $(\alpha_1, \beta_1), ..., (\alpha_n, \beta_n)\in\N(w)$,
                    then $\alpha_1 \cap ... \cap \alpha_n \not= \0$.  \\
    \end{tabular}
  \end{center}
\end{definition}

In the following, we simply write $w \Vd A$ and
$\llbracket A\rrbracket$, omitting both the model $\M$ and the
subscript $\bi$, $\st$, or $r$ whenever they are clear from the
context.

For every condition (\cX) above, we call \emph{\cX-model} any \bin{}
model statisfying (\cX).  The class of \bin{} models for a given
non-normal modal logic $\logic$ is determined by the conditions
corresponding to the axioms of $\logic$.  We denote by $\Cblogic$ the
class of \bin{} models for $\logic$.

We now prove that non-normal modal logics are sound and complete with
respect to the corresponding classes of \bin\ models.  For the systems
of the classical cube, a direct proof based on the canonical model
construction can be found in \cite{Dalmonte:2018}.  Here we give an
indirect argument that relies on the completeness of non-normal modal
logics with respect to standard models and the mutual transformation
between standard and \bin-models.

First of all, given a standard model, an equivalent \bin\ model can be
obtained simply by taking as pairs each neighbourhood and its
complement (for the classical cube this transformation is already
introduced in \cite{Dalmonte:2018}).

\begin{proposition}\label{prop:transformation from standard to bin}
  Let $\Mn=\langle \W, \Nn, \V\rangle$ be a standard model, and
  $\Mb=\langle \W, \Nb, \V\rangle$ be the \bin\ model defined by
  taking the same $\W$ and $\V$ and, for all $w\in \W$,
  \begin{center}
    $\Nb(w) = \left\lbrace
      \begin{array}{lll} \{(\alpha, \Wcompl\alpha) \mid \alpha \in\Nn(w)\} && \textup{if $\Mn$ is not supplemented.} \\ 
        \{(\alpha, \0) \mid \alpha \in\Nn(w)\} && \textup{if $\Mn$ is supplemented.} \end{array} \right.$
  \end{center}
  \noindent
  Then for every formula $A$ of $\lan$ and every $w\in \W$,
  $\Mb, w \Vd A$ if and only if $\Mn, w \Vd A$.  Moreover, for every
  \axX $\in$ \{\axM, \axC, \axN, \axT, \axP, \axD, \RDnplus\}, if
  $\Mn$ satisfies the condition corresponding to \axX{} in the
  standard semantics, then $\Mb$ is a \bin\ \cX-model.
\end{proposition}

\begin{proof} 
  The equivalence is proved by induction on $A$.  The basic cases
  $A= p, \bot, \top$ are trivial since the evaluation $\V$ is the same
  in the two models, and the inductive cases of boolean connectives
  are straightforward by applying the induction hypothesis.  We
  consider the case $A = \Box B$.  If $\Mn$ is not supplemented we
  have: $\Mb, w \Vd \Box B$ iff
  $(\ltr B \rtr_{\bi}, \W\setminus \ltr B \rtr_{\bi}) \in \Nb(w)$ iff
  $\ltr B \rtr_{\bi} \in \Nn(w)$ iff (\ih)
  $\ltr B \rtr_{\st} \in \Nn(w)$ iff $\Mn, w \Vd \Box B$.  If $\Mn$ is
  supplemented we have: $\Mb, w \Vd \Box B$ iff there is
  $(\alpha, \0) \in \Nb(w)$ such that
  $\alpha\subseteq \ltr B \rtr_{\bi}$ iff $\alpha \in \Nn(w)$ and
  (\ih) $\alpha\subseteq \ltr B \rtr_{\st}$ iff (by supplementation)
  $\ltr B \rtr_{\st} \in \Nn(w)$ iff $\Mn, w \Vd \Box B$.
	
  Now we show that $\Mn$ satisfies the right properties.  For axiom
  \axM\ the proof is immediate by definition of $\Nb$.  For the
  following conditions we just consider the non-supplemented case,
  whereas the supplemented case is an easy simplification.
	
  (\axN) $(\W, \0)\in\Nb(w)$ because $\W \in\Nn(w)$.
	
  (\axC) If $(\alpha, \W\setminus\alpha), (\beta,\W\setminus\beta)\in\Nb(w)$,
  then $\alpha,\beta\in\Nn(w)$, that implies $\alpha \cap\beta\in\Nn(w)$.
  Thus $(\alpha\cap\beta, \W\setminus(\alpha\cap\beta)) = (\alpha\cap\beta, \W\setminus\alpha\cup\W\setminus\beta)\in\Nb(w)$.
	
  (\axT) If $(\alpha, \W\setminus\alpha)\in\Nb(w)$,
  then $\alpha\in\Nn(w)$, thus $w\in\alpha$.
  
  (\axP) If $(\alpha, \W\setminus\alpha)\in\Nb(w)$,
  then $\alpha\in\Nn(w)$, thus $\alpha\not=\0$.
  
  (\axD) If $(\alpha, \W\setminus\alpha), (\beta,\W\setminus\beta)\in\Nb(w)$,
  then $\alpha,\beta\in\Nn(w)$. Thus $\beta\not=\W\setminus\alpha$, that implies
  $\alpha\cap\beta\not=\0$ or $\W\setminus\alpha\cap\W\setminus\beta\not=\0$.
  
  (\RDnplus) If $(\alpha_1, \W\setminus\alpha_1), ..., (\alpha_n, \W\setminus\alpha_n)\in\Nb(w)$,
  then $\alpha_1, ..., \alpha_n\in\Nn(w)$, thus $\alpha_1 \cap ... \cap \alpha_n \not=\0$.
\end{proof}

For the opposite direction we propose two transformations:
a more general one, independent of the language,
and a ``finer'' one which is defined  with respect to a set of formulas.
The general transformation is new, whereas the second one is already introduced in \cite{Dalmonte:2018} for the classical cube, where the equivalence proof is also sketched.
The general transformation is as follows.

\begin{proposition}\label{prop:rough transformation}
  Let $\Mb=\langle \W, \Nb, \V\rangle$ be a bi-neighbourhood model,
  and $\Mn=\langle \W, \Nn, \V\rangle$ be the standard model defined
  by taking the same $\W$ and $\V$ and, for all $w\in \W$,
  \begin{center}
    $\Nn(w)=\{\gamma \subseteq\W \mid \textup{there is } (\alpha, \beta) \in\Nb(w)
    \textup{ such that } 
    \alpha\subseteq \gamma \subseteq \Wcompl\beta\}$.
  \end{center}
  \noindent
  Then for every formula $A$ of $\lan$ and every $w\in \W$,
  $\Mn, w \Vd A$ if and only if $\Mb, w \Vd A$.  Moreover, for every
  \axX $\in$ \{\axM, \axC, \axN, \axT, \axP, \axD, \RDnplus\}, if
  $\Mb$ is a \bin\ \cX-model, then $\Mn$ satisfies the condition
  corresponding to \axX{} in the standard semantics.
\end{proposition}

\begin{proof}
  The equivalence is proved by induction on $A$.  As before, we only
  consider the inductive step where $A\equiv \Box B$.  We have:
  $\Mn, w \Vd \Box B$ iff $\llbracket B\rrbracket_{\st} \in \Nn(w)$
  iff (i.h.) $\llbracket B\rrbracket_{\bi} \in \Nn(w)$ iff there is
  $(\alpha, \beta) \in\Nb(w)$ such that
  $\alpha\subseteq \llbracket
  B\rrbracket_{\bi}\subseteq\W\setminus\beta$ iff $\Mb, w \Vd \Box B$.
	
  Now we prove that $\Mn$ satisfies the right properties.
	
  (\axM) Let $\Mb$ be a \cM-model, and assume $\gamma\in\Nn(w)$ and
  $\gamma\subseteq\delta$.  Then there is $(\alpha, \0)\in\Nb(w)$ such
  that $\alpha\subseteq\gamma\subseteq\W\setminus\0$.  Thus
  $\alpha\subseteq\delta\subseteq\W\setminus\0$, which implies
  $\delta\in\Nn(w)$.
	
  (\axN) Let $\Mb$ be a \cN-model.  Then there is
  $(\alpha, \0)\in\Nb(w)$.  Since
  $\alpha\subseteq\W\subseteq\W\setminus\0$, by definition
  $\W\in\Nn(w)$.
	
  (\axC) Let $\Mb$ be a \cC-model, and assume
  $\gamma, \delta\in\Nn(w)$.  Then there are
  $(\alpha_1, \beta_1), (\alpha_2, \beta_2)\in\Nb(w)$ such that
  $\alpha_1\subseteq\gamma\subseteq\W\setminus\beta_1$,
  $\alpha_2\subseteq\delta\subseteq\W\setminus\beta_2$.  By condition
  (\cC), $(\alpha_1\cap\alpha_2, \beta_1\cup \beta_2)\in\Nb(w)$, where
  $\alpha_1\cap\alpha_2 \subseteq\gamma\cap\delta$, and
  $\gamma\cap\delta \subseteq\W\setminus\beta_1\cap\W\setminus\beta_2
  = \W\setminus\beta_1\cup\beta_2$.  Then
  $\gamma\cap \delta\in\Nn(w)$.
	
  (\axT) Let $\Mb$ be a \cT-model, and assume $\gamma\in\Nn(w)$.  Then
  there is $(\alpha, \beta)\in\Nb(w)$ such that
  $\alpha\subseteq\gamma\subseteq\W\setminus\beta$.  By condition
  (\cT), $w\in\alpha$, then $w\in\gamma$.
	
  (\axP) Let $\Mb$ be a \cP-model, and assume towards contradiction
  that $\0\in\Nn(w)$.  Then there is $(\alpha, \beta)\in\Nb(w)$ such
  that $\alpha\subseteq\0\subseteq\W\setminus\beta$.  Thus
  $\alpha=\0$, against condition (\cP).
	
  (\axD) Let $\Mb$ be a \cD-model, and assume towards contradiction
  that $\gamma, \W\setminus\gamma\in\Nn(w)$.  Then there are
  $(\alpha_1, \beta_1), (\alpha_2, \beta_2)\in\Nb(w)$ such that
  $\alpha_1\subseteq\gamma\subseteq\W\setminus\beta_1$,
  $\alpha_2\subseteq\W\setminus\gamma\subseteq\W\setminus\beta_2$.
  Then $\alpha_1\cap\alpha_2=\0$ and $\beta_1\cap\beta_2=\0$, against
  condition (\cD).
	
  (\RDnplus) Let $\Mb$ be a \cRDnplus-model, and assume
  $\gamma_1, ..., \gamma_n\in\Nn(w)$.  Then there are
  $(\alpha_1, \beta_1), ..., (\alpha_n, \beta_n)\in\Nb(w)$ such that
  $\alpha_i\subseteq\gamma\subseteq\W\setminus\beta_i$ for all
  $1\leq i \leq n$.  By condition (\cRDnplus),
  $\alpha_1 \cap ... \cap \alpha_n \not=\0$.  Then
  $\gamma_1 \cap ... \cap \gamma_n \not=\0$.
\end{proof}

For the non-monotonic case, we provide another transformation, defined
with respect to any set of formulas $\s$ closed under subformulas.
This transformation in general produces standard models with a
\emph{smaller} neighbourhood function.

\begin{proposition}\label{prop:finer transformation}
  Let $\Mb=\langle \W, \Nb, \V\rangle$ be a bi-neighbourhood model and
  $\s$ be a set of formulas of $\lan$ closed under subformulas.  We
  define the standard model $\Mn=\langle \W, \Nn, \V\rangle$ with the
  same $\W$ and $\V$ and by taking, for all $w\in \W$,
  \begin{center}
    $\Nn(w)=\{\ltrset A\rtrset_{\bi} \mid \Box A \in \s \textup{ and } \Mb, w \Vd \Box A\}$.
  \end{center}
  Then for every formula $A\in \s$ and every world $w\in \W$,
  $\Mn, w \Vd A$ if and only if $\Mb, w \Vd A$.  Moreover, (\cN) if
  $\Box\top\in \s$ and $\Mb$ is a \cN-model, then $\Mn$ contains the
  unit; (\cC) if $\Box A, \Box B \in \s$ implies
  $\Box(A\land B) \in \s$ and $\Mb$ is a \cC-model, then $\Mn$ is
  closed under intersection; (\cT/\cP/\cD/\cRDnplus) If $\Mb$ is a
  \cT/\cP/\cD/\cRDnplus-model, then $\Mn$ satisfies the corresponding
  condition in the standard semantics.
\end{proposition}

\begin{proof}
  The equivalence is proved by induction on $A$.  The basic cases are
  immediate.  If $A\equiv B\circ C$, where
  $\circ\in\{\land, \lor, \to\}$, the claims holds by applying the
  inductive hypothesis since $B,C\in \s$ because $\s$ is closed under
  subformulas.  If $A\equiv \Box B$, then $B\in \s$ and, by \ih,
  $\ltrset B\rtrset_{\st}=\ltrset B\rtrset_{\bi}$.  Thus
  $\Mn, w \Vd\Box B$ iff $\ltrset B\rtrset_{\st}\in \Nn(w)$ iff
  $\ltrset B\rtrset_{\bi}\in \Nn(w)$ iff there is $\Box C \in \s$ such
  that $\ltrset C\rtrset_{\bi} = \ltrset B\rtrset_{\bi}$ and
  $\Mn, w \Vd\Box C$ iff $\Mn, w \Vd\Box B$.
	
  (\cN) Let $\Mb$ be a \cN-model.  Then $\Mb,w \Vd\Box\top$. Since
  $\Box\top \in \s$, by definition
  $\ltrset\top\rtrset_{\bi} = \W \in \Nn(w)$.
	
  (\cC) Assume $\alpha, \beta \in \Nn(w)$.  Then there are
  $\Box A, \Box B\in \s$ such that $\alpha=\ltrset A\rtrset_{\bi}$,
  $\beta=\ltrset B\rtrset_{\bi}$, and $\Mb,w \Vd \Box A$,
  $\Mb,w \Vd \Box B$, that is $\Mb,w \Vd \Box A\land\Box B$.  Since
  $\Mb$ is a \cC-model we have $\Mb,w \Vd \Box(A\land B)$.  By the
  properties of $\s$, $\Box(A\land B)\in\s$.  Then by definition
  $\ltrset A\land B\rtrset_{\bi}\in \Nn(w)$, where
  $\ltrset A\land B\rtrset_{\bi}= \ltrset A\rtrset_{\bi} \cap \ltrset
  B\rtrset_{\bi} = \alpha\cap\beta$.
	
  (\cT) Assume $\alpha\in\Nn(w)$.  Then
  $\alpha=\llbracket A\rrbracket_{\bi}$ for some $A$ such that
  $\Box A \in \s$ and $\Mb,w \Vd \Box A$.  Since $\Mb$ is a \cT-model,
  $\Mb,w \Vd A$, that is $w\in\llbracket A\rrbracket_{\bi}=\alpha$.
	
  (\cP) Assume by contradiction that $\0\in\Nn(w)$.  Then there is
  $\Box A \in \s$ such that $\Mb,w \Vd \Box A$ and
  $\llbracket A\rrbracket_{\bi}=\0=\llbracket\bot\rrbracket_{\bi}$.
  Thus $\Mb,w \Vd \Box \bot$, against the soundness of axiom \axP\
  with respect to \cP-models.
	
  (\cD) Assume $\alpha, \W\setminus\alpha \in \Nn(w)$.  Then there are
  $\Box A, \Box B\in \s$ such that $\alpha=\ltrset A\rtrset_{\bi}$,
  $\W\setminus\alpha=\ltrset B\rtrset_{\bi}$, and $\Mb,w \Vd \Box A$,
  $\Mb,w \Vd \Box B$.  Then
  $\ltrset A\rtrset_{\bi} = \W\setminus\ltrset B\rtrset_{\bi} =
  \ltrset \neg B\rtrset_{\bi}$, that is $\Mb,w \Vd \Box \neg B$,
  against the soundness of axiom \axD\ with respect to \cD-models.
	
  (\cRDnplus) Assume $\alpha_1, ..., \alpha_n \in \Nn(w)$.  Then there
  are $\Box A_1, ... , \Box A_n\in \s$ such that
  $\alpha_i=\ltrset A_i\rtrset_{\bi}$ and $\Mb,w \Vd \Box A_i$ for
  every $1 \leq i \leq n$, that is
  $\Mb,w \Vd \Box A_1 \land ... \land \Box A_n$.  Then
  $\Mb \not\models \neg(\Box A_1 \land ... \land \Box A_n)$, and since
  $\Mb$ is a \cRDnplus-model,
  $\Mb \not\models \neg(A_1 \land ... \land A_n)$, that is
  $\ltrset A_1\rtrset_{\bi} \cap ... \cap \ltrset A_n\rtrset_{\bi} =
  \alpha_1\cap... \cap \alpha_n \not=\0$.
\end{proof}

For the monotonic case, an analogous result could be obtained by
considering the \emph{supplementation} of the neighbourhood function
in the above proposition.  That is we can consider
\begin{center}
  $\Nn'(w)=\{\alpha\subseteq \W \mid \textup{there is } \Box A \in \s
  \textup{ such that }  \Mb, w \Vd \Box A \textup{ and } \ltrset A\rtrset_{\bi} \subseteq \alpha\}$.
\end{center}
\noindent
However, in this case the advantage in the size of the neighbourhood
function with respect to the transformation in
Proposition~\ref{prop:rough transformation} is not as relevant as for
the non-monotonic case.

From Propositions \ref{prop:transformation from standard to bin} and
\ref{prop:rough transformation}, standard and \bin\ semantics are
equivalent, in the sense that a formula is valid in a certain class of
standard models if and only if it is valid in the corresponding class
of \bin\ models.  Since non-normal modal logics are characterised by
standard models, we obtain the following result.

\begin{theorem}[Characterisation]\label{completeness non-monotonic}
  Every non-normal modal logic $\logic$ is sound and complete with
  respect to the corresponding class of \bin{} models, that is: for
  every $A \in \lan$, $\models_{\Cblogic} A$ if and only if
  $\vd_{\logic} A$.
\end{theorem}

We conclude this section with a few observations. First of all, the
bi-neighbourhood/neighbourhood semantics is (more) significant for the
non-mono-tonic systems, that is lacking the axiom \axM: for the
monotonic ones, the truth condition for $\Box$ in \bin-models boils
down to the well-known $\exists\forall$-definition found in the
literature (see \eg~\cite{Pacuit:2017}).  Moreover, by the
transformation presented in Proposition~\ref{prop:rough
  transformation}, elements of \bin\ pairs can be seen as \emph{lower}
and \emph{upper} bounds of neighbourhoods of standard models.

These transformations have also an interest in proof-search: as we
shall see in the following, given a failed proof in our calculi, it is
possible to directly extract a countermodel in the \bin\ semantics,
that can be transformed into an equivalent standard countermodel as a
later step.  Furthermore, whereas in a standard model each
(non-equivalent) $\Box$-ed formula must be ``witnessed'' by a
different neighbourhood, the same \bin\ pair can ``witness'' several
boxed formulas.  For this reason, \bin\ countermodels extracted from
failed proofs have typically a smaller neighbourhood function than the
corresponding standard models.


%% file: hypersequent-calculi-JV.tex

\section{Hypersequent calculi}
\label{sec:def hyp calc}

In order to define our calculi, we extend the structure of sequents in
two ways.  Firstly, sequents can contain so-called \emph{blocks of
  formulas} in addition to formulas of $\lan$.  Secondly, we use {\em
  hypersequents} rather than simple sequents.

\begin{definition}
  A \emph{block} \index{block} is a structure $\str\Sigma$, where
  $\Sigma$ is a finite multiset of formulas of $\lan$.  A
  \emph{sequent} is a pair $\G\seq \D$, where $\G$ is a finite
  multiset of formulas and blocks, and $\D$ is a finite multiset of
  formulas.  A \emph{hypersequent} is a finite multiset of sequents,
  and is written
  \begin{center}
    $\G_1 \seq \D_1 \hyp \ldots \hyp \G_n \seq \D_n$.
  \end{center}
  Given a hypersequent
  $\hH = \G_1 \seq \D_1 \hyp \ldots \hyp \G_n \seq \D_n$, we call
  \emph{components} of $\hH$ the sequents
  $\G_i \seq \D_i,\;1\leq i\leq n$.
\end{definition}

Observe that blocks can occur only in the antecedent of sequents and
not in their succedent.  Both blocks and sequents, but not
hypersequents, can be interpreted as formulas of $\lan$.  The
\emph{formula interpretation} of sequents is as follows:
\begin{center}
  $\fint(A_1, \ldots, A_n, \str\Sigma_1, \ldots, \str\Sigma_m \seq B_1, \ldots, B_k) =
  A_1 \land \ldots \land A_n \land \Box\AND \Sigma_1 \land \ldots \land \Box\AND\Sigma_m  \to
  B_1 \lor \ldots \lor B_k$.
\end{center}
By contrast, there is no formula interpretation for hypersequents in
$\lan$.  The reason is that non-normal modalities are not \emph{strong
  enough} to express the structural connective ``$\hyp$'' of
hypersequents: Intuitively, every component of a hypersequent
corresponds to a world in a model, and all worlds of a model are
potentially relevant for calculating the truth set of a formula, so we
would need a global modality to express the hypersequent structure.

The semantic interpretation of sequents and hypersequents is as
follows.

\begin{definition}
  We say that a sequent $S$ is \emph{valid} in a possible-worlds model
  $\M$ (written $\M\models S$) if for every world $w$ of $\M$,
  $\M, w\Vd \fint(S)$.  We say that a hypersequent $H$ is \emph{valid}
  in $\M$ if for some component $S$ of $\hH$, $\M \models S$.
  Finally, we say that an inference rule $\varR$ is \emph{sound} with
  respect to a model $\M$ (resp.~a class of models $\C$) if in case
  all premisses of $\varR$ are valid in $\M$ (resp.~$\C$), then the
  conclusion of $\varR$ is also valid in $\M$ (resp.~$\C$).
\end{definition}

\begin{figure}
\noindent
\fbox{\begin{minipage}{\textwidth}
\begin{small}
\vspace{0.1cm}

\noindent
\textbf{Propositional rules} 

\vspace{0.3cm}
\noindent
\ax{}
\llab{\init}
\uinf{$\hG \hyp \G, p \seq p, \D$}
\disp
\hfill
\ax{}
\llab{\lbot}
\uinf{$\hG \hyp \G, \bot \seq \D$}
\disp 
\hfill
\ax{}
\llab{\rtop}
\uinf{$\hG \hyp \G \seq \top, \D$}
\disp

\vspace{0.4cm}
\noindent
\resizebox{\textwidth}{!}{
\ax{$\hG \hyp \G, A \to B \seq A, \D$}
\ax{$\hG \hyp \G, A \to B, B \seq \D$}
\llab{\lto}
\binf{$\hG \hyp \G,  A \to B \seq \D$}
\disp
\hfill
\ax{$\hG \hyp \G, A \seq B, A\to B, \D$}
\llab{\rto}
\uinf{$\hG \hyp \G \seq  A \to B, \D$}
\disp}

\vspace{0.4cm}
\noindent
\ax{$\hG \hyp \G, A \land B, A, B \seq \D$}
\llab{\lland}
\uinf{$\hG \hyp \G, A \land B \seq \D$}
\disp
\hfill
\ax{$\hG \hyp \G \seq A, A \land B, \D$}
\ax{$\hG \hyp \G \seq B, A \land B, \D$}
\llab{\rland}
\binf{$\hG \hyp \G \seq A \land B, \D$}
\disp

\vspace{0.4cm}
\noindent
\ax{$\hG \hyp \G, A \lor B, A \seq \D$}
\ax{$\hG \hyp \G, A \lor B, B \seq \D$}
\llab{\llor}
\binf{$\hG \hyp \G, A \lor B \seq \D$}
\disp
\hfill
\ax{$\hG \hyp \G \seq A, B, A \lor B, \D$}
\llab{\rlor}
\uinf{$\hG \hyp \G \seq A \lor B, \D$}
\disp

\vspace{0.5cm}
\noindent
\textbf{Modal rules for the classical cube}

\vspace{0.4cm}
\noindent
\ax{$\hG \hyp \G, \Box A, \str A \seq \D$}
\llab{\lbox}
\uinf{$\hG \hyp \G, \Box A \seq \D$}
\disp
\hfill
\ax{$\hG \hyp \G, \langle\Sigma\rangle \seq  \Box B, \D \hyp \Sigma \seq B$}
\llab{\rboxm} 
\uinf{$\hG \hyp \G, \langle\Sigma\rangle \seq \Box B, \D$}
\disp

\vspace{0.4cm}
\noindent
\ax{$\hG \hyp \G, \langle\Sigma\rangle \seq \Box B, \D \hyp \Sigma \seq B$}
\ax{$\{  \hG \hyp \G,  \langle\Sigma\rangle \seq \Box B, \D \hyp B \seq A  \}_{A\in\Sigma}$}
\llab{\rbox} 
\binf{$\hG \hyp \G,  \langle\Sigma\rangle \seq \Box B, \D$}
\disp 

\vspace{0.4cm}
\noindent
\ax{$\hG \hyp \G, \langle \top\rangle \seq \D$}
\llab{\rulen}
\uinf{$\hG \hyp \G \seq \D$}
\disp
\hfill
\ax{$\hG \hyp \G, \langle\Sigma\rangle, \langle\Pi\rangle,  \langle\Sigma,\Pi\rangle \seq \D$}
\llab{\rulec} 
\uinf{$\hG \hyp \G, \langle\Sigma\rangle, \langle\Pi\rangle \seq \D$}
\disp

\vspace{0.5cm}
\noindent
\textbf{Modal rules for extensions} 

\vspace{0.4cm}
\noindent
\ax{$\hG \hyp \G, \str\Sigma, \Sigma \seq \D$}
\llab{\rulet}
\uinf{$\hG \hyp \G, \str\Sigma \seq \D$}
\disp
\hfill
\ax{$\hG \hyp \G, \str\Sigma \seq \D \hyp \Sigma \seq$}
\llab{\rulep}
\uinf{$\hG \hyp \G, \str\Sigma \seq \D$}
\disp

\vspace{0.4cm}
\noindent
\ax{$\hG \hyp \G, \str{\Sigma_1}, \ldots, \str{\Sigma_n} \seq \D \hyp \Sigma_1, \ldots, \Sigma_n \seq$}
\llab{\rulednplus}
\uinf{$\hG \hyp \G, \str{\Sigma_1}, \ldots, \str{\Sigma_n} \seq \D$}
\disp 

\vspace{0.4cm}
\noindent
\ax{$\hG \hyp \G, \str\Sigma, \str\Pi \seq \D \hyp \Sigma, \Pi \seq$}
\ax{$\{ \hG \hyp \G,  \str\Sigma, \str\Pi \seq \D \hyp \ \seq A, B\}_{A\in\Sigma, B\in\Pi}$}
\llab{\ruled}
\binf{$\hG \hyp \G, \str\Sigma, \str\Pi \seq \D$}
\disp 

\vspace{0.4cm}
\noindent
\ax{$\hG \hyp \G, \str\Sigma \seq \D \hyp \Sigma \seq$}
\ax{$\{ \hG \hyp \G,  \str\Sigma \seq \D \hyp \ \seq A\}_{A\in\Sigma}$}
\llab{\ruledaux}
\binf{$\hG \hyp \G, \str\Sigma \seq \D$}
\disp 

\end{small}
\end{minipage}}
\caption{\label{fig:hypersequent rules} Rules of the hypersequent calculi $\HEstar$.}  
\end{figure}

For every logic $\Estar$, the corresponding hypersequent calculus
$\HEstar$ is defined by a subset of the rules in
Figure~\ref{fig:hypersequent rules}, as summarised in
Table~\ref{tab:hyp calculi}.

\begin{table}[t]
\centering
\begin{tabular}{l l l}
$\hypcalc\E$ := \{propositional rules, \lbox, \rbox\}.
&
$\hypcalc\EM$ := \{propositional rules, \lbox, \rboxm\}.
\\

$\hypcalc\ENstar$ := $\hypcalc\Estar$ $\cup$ \{\rulen\}.
&
$\hypcalc\EMNstar$ := $\hypcalc\EMstar$ $\cup$ \{\rulen\}.
\\

$\hypcalc\ECstar$ := $\hypcalc\Estar$ $\cup$ \{\rulec\}.
&
$\hypcalc\EMCstar$ := $\hypcalc\EMstar$ $\cup$ \{\rulec\}.
\\

$\hypcalc\ETstar$ := $\hypcalc\Estar$ $\cup$ \{\rulet\}.
&
$\hypcalc\EMTstar$  := $\hypcalc\EMstar$ $\cup$ \{\rulet\}.
\\

$\hypcalc\EPstar$ := $\hypcalc\Estar$ $\cup$ \{\rulep\}.
&
$\hypcalc\EMPstar$ := $\hypcalc\EMstar$ $\cup$ \{\rulep\}.
\\

$\hypcalc\EDstar$ := $\hypcalc\Estar$ $\cup$ \{\ruledaux, \ruled\}.
&
$\hypcalc\EMDstar$  := $\hypcalc\EMstar$ $\cup$ 
\{\ruledoneplus, \ruledtwoplus\}.
\\

$\hypcalc\EDnplusstar$ := $\hypcalc\Estar$ $\cup$
\{\rulediplus $\mid$ $1 \leq i \leq n$\}.
&
$\hypcalc\EMDnplusstar$  := $\hypcalc\EMstar$ $\cup$ 
\{\rulediplus $\mid$ $1 \leq i \leq n$\}.
\\
\end{tabular}
\caption{\label{tab:hyp calculi} Hypersequent calculi $\HEstar$.}
\end{table}

The rules are given in their cumulative, or \emph{kleene'd}, versions,
\ie\ the principal formulas or blocks are copied to the premiss(es).
The propositional rules are just the hypersequent versions of kleene'd
rules of sequent calculi.

As mentioned in the introduction, the hypersequent structure \emph{is
  not needed } to obtain a sound and complete calculus for the logics
under investigation. Moreover, it can be checked that whenever a
hypersequent $\hH = \G_1 \seq \D_1 \hyp \ldots \hyp \G_n \seq \D_n$ is
derivable, then there is some component $\G_i \seq \D_i$ which is
derivable.

The choice of both the hypersequent structure and also of cumulative
rules is motivated by the possibility of \emph{directly} obtaining
countermodels of non-valid formulas.  In particular, the hypersequent
structure allows us to make all rules invertible.  In this respect,
observe that backward applications of rules \rbox, \rboxm, \rulep,
\ruledaux, \ruled, and \rulednplus\ create new components, but the
principal component in the conclusion is kept in the premiss in order
keep the possibility of potential alternative rule applications.

Similarly to propositional connectives, boxed formulas are handled by
separate left and right rules.  Observe that rule \rbox\ has multiple
premisses, but the number of the premisses is fixed by the cardinality
of the principal block $\str\Sigma$.  The rule \rboxm\ is a right rule
for $\Box$ which replaces \rbox\ in the definition of monotonic
calculi.  Apart from the distinction between monotonic and
non-monotonic calculi, the calculi are modular; in particular,
extensions of $\HE$ and $\hypcalc\EM$ do not require to modify the
basic rules for $\Box$, being defined by simply adding the rules
corresponding to the additional axioms.

Every axiom has a corresponding rule, with the only exceptions of
axiom \axD\ and rules \RDnplus: axiom \axD\ needs both \ruledaux\ and
\ruled\ in the non-monotonic case, whereas it needs \ruledoneplus\ and
\ruledtwoplus\ in the monotonic case.  Moreover, the rules \RDnplus\
need \ruledmplus\ for every $1 \leq m \leq n$.  These requirements
makes contraction admissible (see
Proposition~\ref{prop:adm-struct-rules} and
Section~\ref{sec:struct-prop}), and we could forego them by instead
adopting explicit contraction rules.  Finally, as we shall see in the
countermodel extraction (see Section~\ref{sec:countermodels hyp}), the
rule \ruledaux\ is the syntactic counterpart of the property
$(\0,\0)\notin\N(w)$, which is satified by every \bin\ \cD-model.

Blocks have a central role in all modal rules.  Modal rules
essentially state how to handle blocks.  Notice that the only rule
which expands blocks is \rulec, thus in absence of this rule the
blocks occurring in a proof for a single formula contain only one
formula.  The possibility of collecting formulas by means of blocks
allows us to avoid rules with $n$ principal boxed formulas, as are
common in standard sequent calculi (compare~\cite{Lellmann:2019}).  As
we shall see, blocks also allow for an easy computation of the \bin\
function for the definition of countermodels.

\begin{figure}
\begin{small}

\noindent
\begin{tabular}{l c}
\vspace{1.0cm}
(\RE) &
\ax{$A \seq B$}
\llab{\exwk}
\uinf{$\Box A, \str A \seq \Box B \hyp A \seq B$}
\ax{$B \seq A$}
\rlab{\exwk}
\uinf{$\Box A, \str A \seq \Box B \hyp B \seq A$}
\rlab{\rbox}
\binf{$\Box A, \str A \seq \Box B$}
\rlab{\lbox}
\uinf{$\Box A \seq \Box B$}
\disp \\

\vspace{1.0cm}
(\axM) &
\ax{$\Box (A\land B), \str{A \land B} \seq \Box A \hyp A\land B, A, B \seq A$}
\rlab{\lland}
\uinf{$\Box (A\land B), \str{A \land B} \seq \Box A \hyp A\land B \seq A$}
\rlab{\rboxm}
\uinf{$\Box (A\land B), \str{A \land B} \seq \Box A$}
\rlab{\lbox}
\uinf{$\Box (A\land B) \seq \Box A$}
\disp \\

\vspace{1.0cm}
(\axN) &
\ax{$\str\top \seq \Box \top \hyp \top\seq\top$}
\ax{$\str\top \seq \Box \top \hyp \top\seq\top$}
\rlab{\rbox}
\binf{$\str\top \seq \Box \top$}
\rlab{\rulen}
\uinf{$\seq \Box \top$}
\disp \\

(\axC) &
\ax{$\ldots,  \str{A,B} \seq \Box (A \land B) \hyp A, B \seq A\land  B$}
\ax{$\ldots \hyp A\land B \seq A$}
\ax{$\ldots \hyp A \land B \seq B$}
\rlab{\rbox}
\TrinaryInfC{$\Box A \land \Box B, \Box A, \Box B, \str A, \str B, \str{A,B} \seq \Box (A \land B)$}
\rlab{\rulec}
\uinf{$\Box A \land \Box B, \Box A, \Box B, \str A, \str B \seq \Box (A\land B)$}
\rlab{\lbox}
\uinf{$\Box A \land \Box B, \Box A, \Box B, \str A \seq \Box (A\land B)$}
\rlab{\lbox}
\uinf{$\Box A \land \Box B, \Box A, \Box B \seq \Box (A\land B)$}
\rlab{\lland}
\uinf{$\Box A \land \Box B \seq \Box (A\land B)$}
\disp
\end{tabular}

\vspace{1.0cm}

\begin{tabular}{l}
(\axT)\qquad
\ax{$\Box A, \str A, A \seq A$}
\rlab{\rulet}
\uinf{$\Box A, \str A \seq A$}
\rlab{\lbox}
\uinf{$\Box A \seq A$}
\disp
\qquad\qquad\qquad\qquad
(\axP)\qquad
\ax{$\Box \bot, \str\bot \seq \ \hyp \bot \seq $}
\rlab{\rulep}
\uinf{$\Box \bot, \str\bot \seq$}
\rlab{\lbox}
\uinf{$\Box \bot \seq$}
\disp
\end{tabular}

\vspace{1.0cm}

\resizebox{\textwidth}{!}{
\begin{tabular}{l}
\begin{large}(\axD)\end{large}
\ax{$\Box A \land \Box \neg A, \Box A, \Box \neg A, \str A, \str{\neg A} \seq \ \hyp A \seq A$}
\llab{\lneg}
\uinf{$\Box A \land \Box \neg A, \Box A, \Box \neg A, \str A, \str{\neg A} \seq \ \hyp A, \neg A \seq$}
\ax{$\Box A \land \Box \neg A, \Box A, \Box \neg A, \str A, \str{\neg A} \seq \ \hyp A \seq A$}
\rlab{\rneg}
\uinf{$\Box A \land \Box \neg A, \Box A, \Box \neg A, \str A, \str{\neg A} \seq \ \hyp \  \seq A, \neg A$}
\rlab{\ruled}
\binf{$\Box A \land \Box \neg A, \Box A, \Box \neg A, \str A, \str{\neg A} \seq$}
\rlab{\lbox}
\uinf{$\Box A \land \Box \neg A, \Box A, \Box \neg A, \str A \seq$}
\rlab{\lbox}
\uinf{$\Box A \land \Box \neg A, \Box A, \Box \neg A \seq$}
\rlab{\lland}
\uinf{$\Box A \land \Box \neg A \seq$}
\disp
\end{tabular}}

\vspace{1.0cm}

\begin{tabular}{l}
(\RDnplus)\qquad\qquad
\ax{$A_1, \ldots, A_n \seq$}
\rlab{\exwk}
\uinf{$\Box A_1 \land \ldots \land \Box A_n, \Box A_1, \ldots, \Box A_n, \str{A_1}, \ldots, \str{A_n} \seq \ \hyp A_1, \ldots, A_n \seq$}
\rlab{\rulednplus}
\uinf{$\Box A_1 \land \ldots \land \Box A_n, \Box A_1, \ldots, \Box A_n, \str{A_1}, \ldots, \str{A_n} \seq$}
\rlab{\lland\ $\times$ $n$}
\uinf{$\Box A_1 \land \ldots \land \Box A_n, \Box A_1, \ldots, \Box A_n \seq$}
\rlab{\lland\ $\times$ $n$}
\uinf{$\Box A_1 \land \ldots \land \Box A_n \seq$}
\disp
\end{tabular}
\end{small}
\caption{\label{fig:hypersequent derivations} Derivations of modal axioms and rules.}  
\end{figure}

Derivations of modal axioms and rules are displayed in
Figure~\ref{fig:hypersequent derivations}. Note that the simulations
of the rules make use of the external weakening rule \exwk, which is
shown to be admissible in Prop.~\ref{prop:adm-struct-rules}.  In the
derivations we further implicitly make use of the following lemma,
which states that initial hypersequents can be generalised to
arbitrary formulas.

\begin{proposition}\label{prop:extended initial sequents hyp}
  $\hG \hyp \G, A \seq A, \D$ is derivable in $\HEstar$ for every
  $A, \G, \D, \hG$.
\end{proposition}

\begin{proof}
  By structural induction on $A$.  If $A = p, \bot, \top$, then
  $\hG \hyp \G, A \seq A, \D$ is an initial hypersequent, whence it is
  derivable.  If $A = B \land C$ we consider the following derivation

  \begin{center}
    \ax{$\hG \hyp \G, B \land C, B, C \seq B, B \land C, \D$}
    \ax{$\hG \hyp \G, B \land C, B, C \seq C, B \land C, \D$}
    \rlab{\rland}
    \binf{$\hG \hyp \G, B \land C, B, C \seq B \land C, \D$}
    \rlab{\lland}
    \uinf{$\hG \hyp \G, B \land C \seq B \land C, \D$}
    \disp
  \end{center}

  \noindent
  where the premisses are derivable by \ih.  The cases $A = B \lor C$
  or $A = B \land C$ are analogous.  If $A = \Box B$ we consider the
  following derivation

  \begin{center}
    \ax{$\hG \hyp \G,\Box B, \str B \seq \Box B, \D \hyp B \seq B$}
    \ax{$\hG \hyp \G,\Box B, \str B \seq \Box B, \D \hyp B \seq B$}
    \rlab{\rbox}
    \binf{$\hG \hyp \G,\Box B, \str B \seq \Box B, \D$}
    \rlab{\lbox}
    \uinf{$\hG \hyp \G,\Box B \seq \Box B, \D$}
    \disp
  \end{center}

  \noindent
  where the premisses are derivable by \ih.
\end{proof}

The hypersequent calculi are sound with respect to the corresponding
\bin\ models.

\begin{theorem}[Soundness]\label{th:soundness hyp} 
  If $\hH$ is derivable in $\HEXstar$, then it is valid in all \cX-models.
\end{theorem}

\begin{proof}
  The initial hypersequents are clearly valid.  We show that all rules
  are sound with respect to the corresponding \bin\ models.  Since the
  proof is standard for propositional rules, we just consider the
  modal rules.

  (\lbox) Assume $\M \models \hG \hyp \G, \Box A, \str A \seq \D$.
  Then $\M\models\hG$, or $\M\models \G, \Box A, \str A \seq \D$.  In
  the first case we are done.  In the second case,
  $\M \models \fint(\G, \Box A, \str A \seq \D) = \fint(\G, \Box A,
  \Box A \seq \D)$, which is equivalent to
  $\fint(\G, \Box A \seq \D)$.

  (\rbox) Assume
  $\M \models \hG \hyp \G, \str \Sigma \seq \Box B, \D \hyp \Sigma
  \seq B$ and
  $\M \models \hG \hyp \G, \str\Sigma \seq \Box B, \D \hyp B \seq A$
  for all $A\in\Sigma$.  Then (i) $\M \models \hG$, or (ii)
  $\M \models \G, \str\Sigma \seq \Box B, \D$, or (iii)
  $\M \models \Sigma \seq B$ and $\M \models B \seq A$ for all
  $A\in\Sigma$.  If (i) or (ii) we are done.  If (iii), then
  $\M \models \AND\Sigma \to B$ and $\M \models B \to A$ for all
  $A\in\Sigma$, that is $\M \models \AND\Sigma \leftrightarrow B$.
  Since \RE\ is valid,
  $\M \models \Box\AND\Sigma \to \Box B = \fint(\str\Sigma \seq \Box
  B)$.  Thus $\M \models \fint(\G, \str \Sigma \seq \Box B, \D)$.

  (\rboxm) Analogous to \rbox, by considering that in \cM-models
  $\M\models \AND \Sigma \to B$ implies
  $\M\models \Box \AND \Sigma \to \Box B$.

  (\rulen) Suppose $\M$ is a \cN-model and assume
  $\M \models\hG \hyp \G, \str\top \seq \D$.  Then $\M\models\hG$, or
  $\M\models \G, \str\top \seq \D$.  In the first case we are done.
  In the second case, $\M\models \fint(\G, \str\top \seq \D)$, which
  is equivalent to $\Box \top \to \fint(\G \seq \D)$.  Since
  $\Box \top$ is valid in $\M$, $\M\models \G \seq \D$.

  (\rulec) Suppose $\M$ is a \cC-model and assume
  $\M \models \hG \hyp \G, \str\Sigma, \str\Pi, \str{\Sigma,\Pi} \seq
  \D$.  Then $\M \models \hG$ or
  $\M \models \G, \str\Sigma, \str\Pi, \str{\Sigma,\Pi} \seq \D$.  In
  the first case we are done.  In the second case,
  $\M \models \fint(\G, \str\Sigma, \str\Pi, \str{\Sigma,\Pi} \seq \D)
  = \fint(\G, \Box\AND\Sigma, \Box\AND\Pi,
  \Box(\AND\Sigma\land\AND\Pi) \seq \D)$.  This is equivalent to
  $\Box\AND\Sigma \land \Box\AND\Pi \land \Box(\AND\Sigma\land\AND\Pi)
  \to \fint(\G \seq \D)$, and since axiom \axC\ is valid in $\M$, this
  is equivalent to
  $\Box\AND\Sigma \land \Box\AND\Pi \to \fint(\G \seq \D)$.  Thus
  $\M \models \fint(\G, \str\Sigma, \str\Pi \seq \D)$.

  (\rulet) Suppose $\M$ is a \cT-model and assume
  $\M \models \hG \hyp \G, \str\Sigma, \Sigma \seq \D$.  Then
  $\M \models \hG$ or $\M \models \G, \str\Sigma, \Sigma \seq \D$.  In
  the first case we are done.  In the second case,
  $\M \models \fint(\G, \str\Sigma, \Sigma \seq \D) = \Box\AND\Sigma
  \land \AND\Sigma \to \fint(\G \seq \D)$.  Since axiom \axT\ is valid
  in $\M$, this is equivalent to
  $\Box\AND\Sigma \to \fint(\G \seq \D)$.  Then
  $\M \models \G, \str\Sigma \seq \D$.

  (\rulep) Suppose $\M$ is a \cP-model and assume
  $\M \models \hG \hyp \G, \str\Sigma \seq \D \hyp \Sigma \seq$.  Then
  (i) $\M\models \hG$, or (ii) $\M\models \G, \str\Sigma \seq \D$, or
  (iii) $\M \models \Sigma \seq$.  If (i) or (ii) we are done.  If
  (iii), then $\M \models \AND\Sigma \to \bot$.  and by the validity
  of axiom \axP,
  $\M \models \Box\AND\Sigma \to \bot = \fint(\str\Sigma \seq)$.  Then
  $\M\models \G, \str\Sigma \seq \D$.

  (\rulednplus) Suppose $\M$ is a \cRDnplus-model and assume
  $\M \models \hG \hyp \G, \str{\Sigma_1}, \ldots, \str{\Sigma_n} \seq
  \D \hyp \Sigma_1, \ldots, \Sigma_n \seq$.  Then (i) $\M\models \hG$,
  or (ii)
  $\M\models \G, \str{\Sigma_1}, \ldots, \str{\Sigma_n} \seq \D$, or
  (iii) $\M \models \Sigma_1, \ldots, \Sigma_n \seq$.  If (i) or (ii)
  we are done.  If (iii), then
  $\M \models \neg(\AND\Sigma_1 \land \ldots \land \AND\Sigma_n)$.
  And by the soundness of rule \RDnplus,
  $\M \models \neg(\Box\AND\Sigma_1 \land \ldots \land
  \Box\AND\Sigma_n) = \fint(\str{\Sigma_1}, \ldots, \str{\Sigma_n}
  \seq)$.  Then
  $\M\models \G, \str{\Sigma_1}, \ldots, \str{\Sigma_n} \seq \D$.

  (\ruled) Suppose $\M$ is a \cD-model and assume
  $\M \models \hG \hyp \G, \str\Sigma, \str\Pi \seq \D \hyp \Sigma,
  \Pi \seq$, and
  $\M \models \hG \hyp \G, \str\Sigma, \str\Pi \seq \D \hyp \ \seq A,
  B$ for all $A \in\Sigma, B \in\Pi$.  Then (i) $\M\models \hG$, or
  (ii) $\M\models \G, \str\Sigma, \str\Pi \seq \D$, or (iii)
  $\M \models \Sigma, \Pi \seq$ and $\M \models \ \seq A, B$ for all
  $A \in\Sigma, B \in\Pi$.  If (i) or (ii) we are done.  If (iii),
  then $\M \models \AND\Sigma \land \AND\Pi \to \bot$ and
  $\M \models A \lor B$ for all $A \in\Sigma, B \in\Pi$.  Thus
  $\M \models \AND\Sigma \leftrightarrow \neg \AND\Pi$.  By the
  soundness of axiom \axD,
  $\M \models \Box\AND\Sigma \land \Box\AND\Pi \to \bot =
  \fint(\str\Sigma, \str\Pi \seq)$.  Then
  $\M\models \G, \str\Sigma, \str\Pi \seq \D$.

  (\ruledaux) Assume
  $\M \models \hG \hyp \G, \str\Sigma \seq \D \hyp \Sigma \seq$, and
  $\M \models \hG \hyp \G, \str\Sigma \seq \D \hyp \ \seq A$ for all
  $A \in\Sigma$.  Then (i) $\M\models \hG$, or (ii)
  $\M\models \G, \str\Sigma \seq \D$, or (iii)
  $\M \models \Sigma \seq$ and $\M \models \ \seq A$ for all
  $A \in\Sigma$.  If (i) or (ii) we are done.  If (iii), then
  $\M \models \AND\Sigma \to \bot$ and $\M \models \AND\Sigma$, which
  is impossible.  Then (i) or (ii) holds.
\end{proof}

\subsection{Structural properties and syntactic completeness}
\label{sec:struct-prop}

We now investigate the structural properties of our calculi.  We first
show that weakening and contraction are height-preserving admissible,
both in their internal and in their external variants, and that all
rules are invertible.  We then prove that the cut rule is admissible,
which allows us to directly prove the completeness of the calculi with
respect to the corresponding axiomatisations.

\begin{definition}
  \label{def:weight hyp}
  The {\em weight} $\w$ of a formula is recursively defined as
  $\w(\bot) =\w(\top) = \w(p) = 0$; for
  $\circ\in\{\land, \lor, \to\}$, $\w(A \circ B) = \w(A) + \w(B) + 1$;
  $\w(\langle A_1, \ldots, A_n \rangle) = max\{\w(A_1), \ldots,
  \w(A_n)\} + 1$; $\w(\Box A) = \w(A) + 2$.

  The {\em height} of a derivation is the greatest number of
  successive applications of rules in it, where axioms have height
  $0$. A property is {\em height-preserving} if the height of
  derivations is an invariant.
\end{definition}

\begin{proposition}\label{prop:adm-struct-rules}
  The following structural rules are height-preserving admissible in
  $\HEstar$, where $\phi$ is any formula $A$ or block $\str\Sigma$.
  Moreover, all rules in $\HEstar$ are height-preserving invertible.

\begin{center}
\begin{small}
\begin{tabular}{lllll}
\vspace{0.3cm}
\ax{$\hG \hyp \G \seq \D $}
\llab{\hyplwk}
\uinf{$\hG \hyp \G, \phi \seq \D$}
\disp 
&
\ax{$\hG \hyp \G \seq \D $}
\llab{\hyprwk}
\uinf{$\hG \hyp \G \seq A, \D$}
\disp
&
\ax{$\hG$}
\llab{\exwk}
\uinf{$\hG \hyp \G \seq \D$}
\disp 
\\

\vspace{0.3cm}
\ax{$\hG \hyp  \G, \phi, \phi \seq \D $}
\llab{\hyplctr}
\uinf{$\hG \hyp \G, \phi \seq \D$}
\disp 
&
\ax{$\hG \hyp \G \seq A, A, \D$}
\llab{\hyprctr}
\uinf{$\hG \hyp \G \seq A, \D$}
\disp
&
\ax{$\hG \hyp \G \seq \D \hyp \G \seq \D$}
\llab{\exctr}
\uinf{$\hG \hyp \G \seq \D$}
\disp  
\\

\ax{$\hG \hyp  \G, \str{\Theta,A,A} \seq \D $}
\llab{\textsf{Sctr}}
\uinf{$\hG \hyp \G, \str{\Theta,A} \seq \D$}
\disp 
&
\ax{$\hG \hyp  \G, \str{\Theta,A} \seq \D $}
\llab{\textsf{Smgl}}
\uinf{$\hG \hyp \G, \str{\Theta,A,A} \seq \D$}
\disp 
\end{tabular}
\end{small}
\end{center}
\end{proposition}

\begin{proof}
  For each rule $\varR$, we prove that if the premise is derivable
  with height at most $n$, then the conclusion is also derivable with
  height at most $n$. The proof is by induction on the height of the
  derivation of the premise, highlighting that rules \hyplctr\ and
  \hyprctr\ are simultaneously proved admissible by mutual
  induction. Moreover, admissibility of \textsf{Sctr} and
  \textsf{Smgl} rely on height-preserving admissibility of contraction
  and weakening on formulae outside blocks, respectively.

  We will illustrate with the following cases
  \begin{itemize}
  \item Case \exctr\ + \rboxm. Suppose that
    \[
      \infer[\mbox{\rboxm}]{\hG \hyp \G, \str\Sigma \seq \Box B, \D \hyp \G, \str\Sigma \seq \Box B, \D}{
        \deduce{\hG \hyp \G, \str\Sigma \seq \Box B, \D \hyp \G, \str\Sigma \seq \Box B, \D \hyp \Sigma \seq B}{\pi}}
    \]
    where the height of $\pi$ is at most $n$.  By induction
    hypothesis, there is a proof $\pi'$, with height at most $n$, of
    $\hG \hyp \G, \str\Sigma \seq \Box B, \D \hyp \Sigma \seq
    B$. Hence
    \[
      \infer[\mbox{\rboxm}]{\hG \hyp \G, \str\Sigma \seq \Box B, \D}{\deduce{\hG \hyp  \G, \str\Sigma \seq \Box B, \D \hyp \Sigma \seq B}{\pi'}}
    \]
  \item Case \hyplctr\ + \ruled. Consider the derivation
    \[
      \infer[\mbox{\ruled}]{\hG \hyp \G, \str\Sigma, \str\Sigma \seq \D}
      {\deduce{\hG \hyp \G, \str\Sigma, \str\Sigma \seq \D \hyp \Sigma, \Sigma \seq}{\pi_1} & \deduce{\{ \hG \hyp \G,  \str\Sigma, \str\Sigma \seq \D \hyp \ \seq A, B\}_{A, B \in\Sigma}}{\pi_2^{A,B}}}
    \]
    where the heights of $\pi_1$ and $\pi_2^{A,B}$ are at most
    $n$. Observe that, in particular, there are proofs $\pi_2^{A,A}$
    of
    $\{\hG \hyp \G, \str\Sigma, \str\Sigma \seq \D \hyp \ \seq A,
    A\}_{A \in\Sigma}$.  By induction hypothesis, there are proofs
    $\pi_1',\pi_2^{A}$, with height at most $n$, of
    $\hG \hyp \G, \str\Sigma \seq \D \hyp \Sigma \seq$ and
    $\{ \hG \hyp \G, \str\Sigma \seq \D \hyp \ \seq A\}_{A
      \in\Sigma}$, respectively. Hence
    \[
      \infer[\mbox{\ruledaux}]{\hG \hyp \G,  \str\Sigma \seq \D}
      {\deduce{\hG \hyp \G, \str\Sigma \seq \D \hyp \Sigma \seq}{\pi_1} & \deduce{\{ \hG \hyp \G,  \str\Sigma \seq \D \hyp \ \seq A\}_{A \in\Sigma}}{\pi_2^{A}}}
    \]
  \end{itemize}

  Finally, note that since all rules are cumulative, height-preserving
  invertibility of all rules is an immediate consequence of
  height-preserving admissibility of weakening.  For instance,
  invertibility of rule \rboxm\ is proved as follows

  \begin{center}
    \ax{$\hG \hyp \G,  \langle\Sigma\rangle \seq \Box B, \D$}
    \rlab{\exwk}
    \uinf{$\hG \hyp \G, \langle\Sigma\rangle \seq \Box B, \D \hyp \Sigma \seq B$}
    \disp
  \end{center}
  \vspace{-0.3cm}
\end{proof}

We note that due to the fact that the \rbox\ rule isolates single
formulae from block in its right premiss, in the non-monotonic case
the full-blown weakening inside blocks is \emph{not}
admissible. However, the weaker rule of \emph{mingle} inside blocks
\textsf{Smgl} is.

The proof of admissibility of \hypcut\ is more intricate and deserves
more attention.  In the hypersequent framework, the \hypcut\ rule is
formulated as follows
\begin{center}
  \ax{$\hG \hyp \G \seq A, \D$}
  \ax{$\hG \hyp \G, A \seq \D$}
  \llab{\hypcut}
  \binf{$\hG \hyp \G \seq \D$}
  \disp
\end{center}
The admissibility of \hypcut\ is proved simultaneously with the
admissibility of the following rule \subrule{}, which states that a
formula $A$ inside one or more blocks can be replaced by any
equivalent multiset of formulas $\Sigma$
\begin{center}
  \ax{$\hG \hyp \Sigma \seq A$} 
  \ax{$\{  \hG \hyp  A \seq B  \}_{B\in\Sigma}$}
  \ax{$\hG \hyp \G, \langle A^{n_1}, \Pi_1 \rangle, \ldots, \langle A^{n_k}, \Pi_k \rangle \seq \D$}
  \llab{\subrule}
  \TrinaryInfC{$\hG \hyp \G, \langle \Sigma^{n_1}, \Pi_1 \rangle, \ldots, \langle \Sigma^{n_k}, \Pi_k \rangle \seq \D$}
  \disp
\end{center}
where $A^{n_i}$ (resp.~$\Sigma^{n_i}$) is a compact way to denote
$n_i$ occurrences of $A$ (resp.~$\Sigma$).  In the monotonic case we
need to consider, instead of \subrule{}, the rule

\begin{center}
  \ax{$\hG \hyp \Sigma \seq A$} 
  \ax{$\hG \hyp \G, \langle A^{n_1}, \Pi_1 \rangle, \ldots, \langle A^{n_k}, \Pi_k \rangle \seq \D$}
  \llab{\subrulem}
  \binf{$\hG \hyp \G, \langle \Sigma^{n_1}, \Pi_1 \rangle, \ldots, \langle \Sigma^{n_k}, \Pi_k \rangle \seq \D$}
  \disp
\end{center}

\begin{theorem}\label{thm:cut}
  If $\HEstar$ is non-monotonic, then the rules \hypcut{} and
  \subrule{} are admissible in $\HEstar$; otherwise \hypcut{} and
  \subrulem{} are admissible in $\HEstar$.
\end{theorem}

\begin{proof}
  We prove that \hypcut{} and \subrule{} are admissible in
  non-monotonic $\HEstar$; the proof in the monotonic case is
  analogous.  Recall that, for an application of \hypcut, the
  \emph{cut formula} is the formula which is deleted by that
  application, while the \emph{cut height} is the sum of the heights
  of the derivations of the premisses of \hypcut.  The theorem is a
  consequence of the following claims, where $Cut(c, h)$ means that
  all applications of \hypcut{} of height $h$ on a cut formula of
  weight $c$ are admissible, and $Sub(c)$ means that all applications
  of \subrule{} where the principal formula $A$ has weight $c$ are
  admissible (for any $\Sigma, \Pi_1, \ldots, \Pi_k$)
  \begin{itemize}
  \item \textbf{(A)}  $\forall c. Cut(c, 0)$. 
  \item \textbf{(B)}  $\forall h. Cut(0, h)$. 
  \item \textbf{(C)}  $\forall c. (\forall h. Cut(c, h) \to Sub(c))$. 
  \item \textbf{(D)}  $\forall c. \forall h. ((\forall c'<c. (Sub(c') \land \forall h'. Cut(c', h')) \land \forall h''<h. Cut (c, h'') ) \to Cut(c, h))$.
  \end{itemize}
  \textbf{(A)} If the cut height is 0, then \hypcut\ is applied to
  initial hypersequents $\hG \hyp \G \seq A, \D$ and
  $\hG \hyp \G, A \seq \D$.  We show that the conclusion of \hypcut\
  $\hG \hyp \G \seq \D$ is an inital hypersequent, whence it is
  derivable without \hypcut.  If $\hG$ is an inital hypersequent we
  are done.  Otherwise $\G \seq A, \D$ and $\G, A \seq \D$ are initial
  sequents.  For the first sequent there are three possibilities: (i)
  $\G \seq \D$ is an initial sequent, or (ii) $A = \top$, or (iii)
  $A = p$ and $\G = \G', p$.  If (ii), then the second sequent is
  $\G, \top \seq \D$, which implies that $\G \seq \D$ is an initial
  sequent.  If (iii), then the second sequent is $\G', p, p \seq \D$.
  Then $\G' \seq \D$ is an initial sequent, or $\D = p, \D'$, which
  implies that $\G', p \seq p, \D' = \G \seq \D$ is an initial
  sequent.

  \bigskip \textbf{(B)} If the cut formula has weight $0$, then it is
  $\bot$, $\top$, or a propositional variable $p$.  For all three
  possibilities the proof is by complete induction on $h$. The basic
  case $h=0$ is a particular case of \textbf{(A)}.  For the inductive
  step, we distinguish three cases.

  (i) The cut formula $\bot$, $\top$, or $p$ is not principal in the
  last rule applied in the derivation of the left premiss.  By
  examining all possible rule applications, we show that the
  application of \hypcut{} can be replaced by one or more applications
  of \hypcut{} at a smaller height.  For instance, assume that the
  last rule applied is \lbox.

  \begin{center}
    \ax{$\hG \hyp \angA, \Box A, \G \seq \D, \bot$}
    \llab{\lbox}
    \uinf{$\hG \hyp \Box A,  \G \seq \D, \bot$}
    \ax{$\hG \hyp \bot, \Box A, \G \seq \D$}
    \rlab{\hypcut}
    \binf{$\hG \hyp \Box A, \G \seq \D$}
    \disp
  \end{center}
  The derivation is transformed as follows, with a \hp\ application of
  \hyplwk{} and an application of \hypcut{} of smaller height.
  \begin{center}
    \ax{$\hG \hyp \angA, \Box A, \G \seq \D, \bot$}
    \ax{$\hG \hyp \bot, \Box A, \G \seq \D$}
    \rlab{\hyplwk}
    \uinf{$\hG \hyp \bot, \angA, \Box A, \G \seq \D$}
    \rlab{\hypcut}
    \binf{$\hG  \hyp \angA, \Box A, \G \seq \D$}
    \rlab{\lbox}
    \uinf{$\hG  \hyp \Box A, \G \seq \D$}
    \disp
  \end{center}
  The situation is similar if the last rule in the derivation of the
  left premiss is applied to some sequent in $\hG$.

  (ii) The cut formula $\bot$, $\top$, or $p$ is not principal in the
  last rule applied in the derivation of the right premiss.  The case
  is analogous to (i).  As an example, suppose that the last rule
  applied is \rboxm{}.
  
  \begin{center}
    \ax{$\hG \hyp \angSigma, \G \seq \D, \Box B, \bot$}
    \ax{$\hG \hyp \bot, \angSigma, \G \seq \D, \Box B \hyp \Sigma \seq B$}
    \rlab{\rboxm}
    \uinf{$\hG \hyp \bot, \angSigma, \G \seq \D, \Box B$}
    \rlab{\hypcut}
    \binf{$\hG \hyp \angSigma, \G \seq \D, \Box B$}
    \disp
  \end{center}

  The derivation is converted into

  \begin{center}
    \ax{$\hG \hyp \angSigma, \G \seq \D, \Box B, \bot$}
    \llab{\exwk}
    \uinf{$\hG \hyp \angSigma, \G \seq \D, \Box B, \bot \hyp \Sigma \seq B$}   
    \ax{$\hG \hyp \bot, \angSigma, \G \seq \D, \Box B \hyp \Sigma \seq B$}
    \rlab{\hypcut}
    \binf{$\hG \hyp \angSigma, \G \seq \D, \Box B \hyp \Sigma \seq B$}
    \rlab{\rboxm}
    \uinf{$\hG \hyp \angSigma, \G \seq \D, \Box B$}
    \disp
  \end{center}
  where \hypcut{} is applied at a smaller height.
   
  (iii) The cut formula $\bot$, $\top$, or $p$ is principal in the
  last rule applied in the derivation of both premisses.  Then the cut
  formula is $p$, as $\bot$ (resp.~$\top$) is never principal on the
  right-hand side (resp.~the left-hand side) of the conclusion of any
  rule application.  This means that both premisses are derived by
  \init, which implies that $h=0$.  Then we are back to case
  \textbf{(A)}.

  \bigskip \textbf{(C)} Assume $\forall h Cut(c,h)$.  We prove that
  all applications of \subrule\ where $A$ has weight $c$ are
  admissible.  The proof is by induction on the height $m$ of the
  derivation of
  $\hG \hyp \langle A^{n_1}, \Pi_1 \rangle, \ldots, \langle A^{n_k},
  \Pi_k \rangle, \G \seq \D$.  If $m = 0$ or no block among
  $\langle A, \Pi_1 \rangle, \ldots, \langle A, \Pi_k \rangle$ is
  principal in the last rule application, then the proof proceeds
  similarly to previous cases.  If $m>0$ and at least one block among
  $\langle A, \Pi_1 \rangle, \ldots, \langle A, \Pi_k \rangle$ is
  principal in the last rule application we have the following
  possibilities.

  $\bullet$ The last rule applied is \rbox:

\begin{center}
  \resizebox{\textwidth}{!}{ \ax{\textcircled{\footnotesize{1}}}
    \noLine
    \uinf{$\hG \hyp \langle A^{n_i}, \Pi_i \rangle, \G \seq \D, \Box D
      \hyp A^{n_i}, \Pi_i \seq D$}
    \ax{$\qquad \quad \ \{ \hG \hyp \langle A^{n_i}, \Pi_i \rangle, \G
      \seq \D, \Box D \hyp D \seq C \}_{C\in\Pi_i}$} \noLine
    \uinf{$\{ \hG \hyp \langle A^{n_i}, \Pi_i \rangle, \G \seq \D,
      \Box D \hyp D \seq A \}_1^{n_i}$ \quad $\vdots$} \rlab{\rbox}
    \BinaryInfC{$\hG \hyp \langle A^{n_i}, \Pi_i \rangle, \G \seq \D,
      \Box D$} \disp }
\end{center}

The derivation is converted as follows.  First we derive:

\begin{center}
  \vspace{0.2cm}
  \noindent
\resizebox{\textwidth}{!}{
\ax{$\hG \hyp \Sigma \seq A$}
\rlab{\exwk}
\uinf{$\hG \hyp \Sigma \seq A \hyp A^{n_i}, \Pi_i \seq D$}
\ax{$\hG \hyp A \seq B$}
\llab{$\Big\{$}
\rlab{\exwk $\Big\}_{B\in\Sigma}$}
\uinf{$\hG \hyp A \seq B \hyp A^{n_i}, \Pi_i \seq D$}
\ax{\textcircled{\footnotesize{1}}}
\rlab{\subrule}
\TrinaryInfC{$\hG \hyp \langle \Sigma^{n_i}, \Pi_i \rangle, \G \seq \D, \Box D \hyp A^{n_i}, \Pi_i \seq D$}
\disp}
\end{center}

\noindent
(where rule \subrule\ possibly applies to further blocks inside $\G$).
Then by applying \exwk{} to $\hG \hyp \Sigma \seq A$ we obtain
$\hG \hyp \langle \Sigma^{n_i}, \Pi_i \rangle, \G \seq \D, \Box D \hyp
\Sigma \seq A$.  By auxiliary applications of \wk{} we can cut $A$ and
get
$\hG \hyp \langle \Sigma^{n_i}, \Pi_i \rangle, \G \seq \D, \Box D \hyp
\Sigma, A^{n_i - 1}, \Pi_i \seq D$.  Then with further applications of
cut (each time with auxiliary applications of \wk{}) we obtain
$\hG \hyp \langle \Sigma^{n_i}, \Pi_i \rangle, \G \seq \D, \Box D \hyp
\Sigma^{n_i}, \Pi_i \seq D.$ By doing the same with the other
premisses of \rbox{} in the initial derivation we obtain also
$\{ \{ \hG \hyp \langle \Sigma^{n_i}, \Pi_i \rangle, \G \seq \D, \Box
D \hyp D \seq B \}_{B\in\Sigma} \}_1^{n_1}$ and
$\{ \hG \hyp \langle \Sigma^{n_i}, \Pi_i \rangle, \G \seq \D, \Box D
\hyp D \seq C \}_{C\in\Pi_i}.$ Finally by \rbox{} we derive the
conclusion of \subrule{}
$\hG \hyp \langle \Sigma^{n_i}, \Pi_i \rangle, \G \seq \D, \Box D.$

$\bullet$ The last rule applied is \rulec: 

\begin{center}
\ax{$\hG \hyp \langle A^{n_i}, \Pi_i \rangle, \langle A^{n_j}, \Pi_j  \rangle, \langle A^{n_i}, A^{n_j}, \Pi_i, \Pi_j  \rangle, \G \seq \D$}
\rlab{\rulec}
\uinf{$\hG \hyp \langle A^{n_i}, \Pi_i \rangle, \langle A^{n_j}, \Pi_j \rangle,  \G \seq \D$}
\disp
\end{center}

\noindent
By applying \subrule\
to the premiss we obtain\\
$\hG \hyp \langle \Sigma^{n_i}, \Pi_i \rangle, \langle \Sigma^{n_j},
\Pi_j \rangle, \langle \Sigma^{n_i}, \Sigma^{n_j}, \Pi_i, \Pi_j,
\rangle, \G \seq \D$,
then by \rulec\ we derive\\
$\hG \hyp \langle \Sigma^{n_i}, \Pi_i \rangle, \langle \Sigma^{n_j},
\Pi_j \rangle, \G \seq \D$.

$\bullet$ The last rule applied is \rulet:

\begin{center}
\ax{$\hG \hyp A^{n_i}, \Pi_i, \str{A^{n_i}, \Pi_i}, \G \seq \D$}
\rlab{\rulet}
\uinf{$\hG \hyp \str{A^{n_i}, \Pi_i}, \G \seq \D$}
\disp
\end{center}
By applying the inductive hypothesis to the premiss we obtain\\
$\hG \hyp A^{n_i}, \Pi_i, \str{\Sigma^{n_i}, \Pi_i}, \G \seq \D$.
Then, from this and $\hG \hyp \Sigma \seq A$,
by several applications of \hypcut{}
(each time with auxiliary applications of \wk{})
we obtain 
$\hG \hyp \Sigma^{n_i}, \Pi_i, \str{\Sigma^{n_i}, \Pi_i}, \G \seq \D$.
Finally, by \rulet{} we derive 
$\hG \hyp \str{\Sigma^{n_i}, \Pi_i}, \G \seq \D$.

$\bullet$ The last rule applied is \rulep:

\begin{center}
\ax{$\hG \hyp \str{A^{n_i}, \Pi_i}, \G \seq \D \hyp A^{n_i}, \Pi_i \seq$}
\rlab{\rulep}
\uinf{$\hG \hyp \str{A^{n_i}, \Pi_i}, \G \seq \D$}
\disp
\end{center}
By applying the inductive hypothesis to the premiss (after auxiliary
applications of \exwk{} to the other premisses of \subrule) we obtain
$\hG \hyp \str{\Sigma^{n_i}, \Pi_i}, \G \seq \D \hyp A^{n_i}, \Pi_i
\seq$.  Then, from this and $\hG \hyp \Sigma \seq A$, by several
applications of \hypcut{} (each time with auxiliary applications of
\wk{}) we obtain
$\hG \hyp \str{\Sigma^{n_i}, \Pi_i}, \G \seq \D \hyp \Sigma^{n_i},
\Pi_i \seq$.  Finally, by \rulep{} we derive
$\hG \hyp \str{\Sigma^{n_i}, \Pi_i}, \G \seq \D$.

$\bullet$ The last rule applied is \ruled: 
Then $\hG \hyp \langle A^{n_i}, \Pi_i \rangle, \langle A^{n_j}, \Pi_j \rangle,  \G \seq \D$ 
has been derived by the following premisses.
$\hG \hyp \langle A^{n_i}, \Pi_i \rangle, \langle A^{n_j}, \Pi_j \rangle,  \G \seq \D \hyp A^{n_i}, A^{n_j}, \Pi_i, \Pi_j \seq $;
$\{\hG \hyp \langle A^{n_i}, \Pi_i \rangle, \langle A^{n_j}, \Pi_j \rangle,  \G \seq \D \hyp \seq A,A\}_1^{n_i+n_j}$;
$\{\{\hG \hyp \langle A^{n_i}, \Pi_i \rangle, \langle A^{n_j}, \Pi_j \rangle,  \G \seq \D \hyp \seq A,C\}_{C\in\Pi_i}\}_1^{n_j}$;
$\{\{\hG \hyp \langle A^{n_i}, \Pi_i \rangle, \langle A^{n_j}, \Pi_j \rangle,  \G \seq \D \hyp \seq A,D\}_{D\in\Pi_j}\}_1^{n_i}$;
and
$\{\hG \hyp \langle A^{n_i}, \Pi_i \rangle, \langle A^{n_j}, \Pi_j \rangle,  \G \seq \D \hyp \seq A,A\}_{C\in\Pi_i, D\in\Pi_j}$.
We consider the other premisses of \subrule\ and apply \hypcut\ many times
(each time with auxiliary applications of \wk)
so to replace all occurrences of $A$ with formulas in $\Sigma$. 
As final step we can apply \ruled\ and obtain
$\hG \hyp \langle \Sigma^{n_i}, \Pi_i \rangle, \langle \Sigma^{n_j}, \Pi_j  \rangle, \G \seq \D$.

$\bullet$ The remaining cases \rulednplus\ and \ruledaux\ are similar
to the previous ones.

\bigskip \textbf{(D)} Assume
$\forall c'<c.\, (Sub(c') \land \forall h'.\, Cut(c', h'))$ and
$\forall h''<h.\, Cut (c, h'')$.  We show that all applications of
\hypcut{} of height $h$ on a cut formula of weight $c$ can be replaced
by different applications of \hypcut{}, either of smaller height or on
a cut formula of smaller weight.  We can assume $h, c > 0$ as the
cases $h=0$ and $c=0$ have been already considered in \textbf{(A)} and
\textbf{(B)}.  We distinguish two cases.

(i) The cut formula is not principal in the last rule application in
the derivation of at least one of the two premisses of \hypcut.  This
case is analogous to (i) or (ii) in \textbf{(B)}.

(ii) The cut formula is principal in the last rule application in the
derivation of both premisses.  Then the cut formula is either
$B\circ C$, with $\circ\in\{\land, \lor, \to\}$, or $\Box B$.

$\bullet$ The case of boolean connective is standard. We consider as
an example $B \to C$.  We have:

\begin{center}
\resizebox{\textwidth}{!}{
\ax{$\textcircled{\footnotesize{1}}$}
\noLine
\uinf{$\hG \hyp B, \G \seq \D, B \to C, C$}
\llab{\rto}
\uinf{$\hG \hyp \G \seq \D, B \to C$}
\ax{$\textcircled{\footnotesize{2}}$}
\noLine
\uinf{$\hG \hyp B\to C, \G \seq \D, B$}
\ax{$\textcircled{\footnotesize{3}}$}
\noLine
\uinf{$\hG \hyp C, B\to C, \G \seq \D$}
\rlab{\lto}
\binf{$\hG \hyp B\to C, \G \seq \D$}
\rlab{\hypcut}
\binf{$\hG \hyp \G \seq \D$}
\disp }
\end{center}

The derivation is converted into the following one:

\begin{center}
\resizebox{\textwidth}{!}{
\ax{$\hG \hyp \G \seq \D, B \to C$}
\llab{\rwk}
\uinf{$\hG \hyp \G \seq \D, B \to C, B$}
\ax{$\textcircled{\footnotesize{2}}$}
\llab{\hypcut}
\binf{$\hG \hyp \G \seq \D, B$}
\llab{\rwk}
\uinf{$\hG \hyp \G \seq \D, B, C$}
\ax{$\textcircled{\footnotesize{1}}$}
\ax{$\hG \hyp B\to C, \G \seq \D$}
\rlab{\lwk}
\uinf{$\hG \hyp B, B\to C, \G \seq \D$}
\rlab{\rwk}
\uinf{$\hG \hyp B, B\to C, \G \seq \D, C$}
\rlab{\hypcut}
\binf{$\hG \hyp B, \G \seq \D, C$}
\rlab{\hypcut}
\binf{$\hG \hyp  \G \seq \D, C$}
\ax{$\hG \hyp \G \seq \D, B \to C$}
\llab{\lwk}
\uinf{$\hG \hyp C, \G \seq \D, B \to C$}
\ax{$\textcircled{\footnotesize{3}}$}
\llab{\hypcut}
\binf{$\hG \hyp C, \G \seq \D$}
\llab{\hypcut}
\binf{$\hG \hyp \G \seq \D$}
\disp}
\end{center}

$\bullet$ If the cut formula is $\Box B$ we have  

\begin{center}
\resizebox{\textwidth}{!}{
\ax{$\hG \hyp \angSigma, \G \seq \D, \Box B \hyp \Sigma \seq B$
\qquad\qquad\qquad\qquad\qquad}
\noLine
\uinf{$\vdots$ \quad $\{\hG \hyp \angSigma, \G \seq \D, \Box B \hyp B \seq C \}_{C\in\Sigma}$}  
\llab{\rbox}
\uinf{$\hG \hyp \angSigma, \G \seq \D, \Box B$}
\ax{$\hG \hyp \angB, \Box B, \angSigma, \G \seq \D$}
\rlab{\lbox}
\uinf{$\hG \hyp \Box B, \angSigma, \G \seq \D$}
\rlab{\hypcut}
\binf{$\hG \hyp \angSigma, \G \seq \D$}
\disp }
\end{center}

The derivation is converted as follows,
with several applications of \hypcut{} of smaller height
and an admissible application of \subrule.

\begin{center}
\ax{$\hG \hyp \angSigma, \G \seq \D, \Box B \hyp \Sigma \seq B$}  
\ax{$\hG \hyp \Box B, \angSigma, \G \seq \D$}
\rlab{\exwk}
\uinf{$\hG \hyp \Box B, \angSigma, \G \seq \D \hyp \Sigma \seq B$}
\rlab{\hypcut}
\binf{$\textcircled{\footnotesize{4}} \ \hG \hyp \angSigma, \G \seq \D \hyp \Sigma \seq B$}
\disp

\vspace{0.5cm}
\ax{$\hG \hyp \angSigma, \G \seq \D, \Box B$}
\rlab{\hyplwk}
\uinf{$\hG \hyp \str B, \angSigma, \G \seq \D, \Box B$}
\ax{$\hG \hyp \angB, \Box B, \angSigma, \G \seq \D$}
\rlab{\hypcut}
\binf{$\hG \hyp \str B, \str\Sigma, \G \seq \D$}
\rlab{\exwk}
\uinf{$\textcircled{\footnotesize{5}} \ \hG \hyp \str\Sigma, \G \seq \D \hyp \str B, \str\Sigma, \G \seq \D$} 
\disp

\vspace{0.5cm}
\resizebox{\textwidth}{!}{
\ax{\textcircled{\footnotesize{4}}}
\ax{$\hG \hyp \angSigma, \G \seq \D, \Box B \hyp B \seq C$}  
\ax{$\hG \hyp \Box B, \angSigma, \G \seq \D$}
\rlab{\exwk}
\uinf{$\hG \hyp \Box B, \angSigma, \G \seq \D \hyp B \seq C$}
\llab{$\Big\{$}
\rlab{\hypcut $\Big\}_{C \in \Sigma}$}
\binf{$\hG \hyp \angSigma, \G \seq \D \hyp B \seq C$}
\ax{\textcircled{\footnotesize{5}}}
\rlab{\subrule}
\TrinaryInfC{$\hG \hyp \str\Sigma, \G \seq \D \hyp \str\Sigma, \str\Sigma, \G \seq \D$}
\rlab{\angctr}
\uinf{$\hG \hyp \str\Sigma, \G \seq \D \hyp \str\Sigma, \G \seq \D$}
\rlab{\exctr}
\uinf{$\hG \hyp \str\Sigma, \G \seq \D$}
\disp
}
\end{center}
\end{proof}

Given the admissibility of the structural rules and \hypcut\ 
we can prove that the calculi are syntactically complete
with respect to the corresponding axiom systems.

\begin{theorem}[Syntactic completeness]
  If $\vd_{\Estar} A$ then $\seq A$ is derivable in $\HEstar$.
\end{theorem}

\begin{proof}
  As usual, we have to show that all axioms of $\Estar$ are derivable
  in $\HEstar$, and that all rules of $\Estar$ are admissible in
  $\HEstar$.  The derivations of the modal axioms and rules are
  displayed in Figure~\ref{fig:hypersequent derivations}.  For the
  derivations of the axioms we implicitly consider
  Proposition~\ref{prop:extended initial sequents hyp}.  For the
  derivation of rule \RE\ we assume that $A \to B$ and $B \to A$ are
  derivable in $\Estar$, and for the derivation of rule \RDnplus\ we
  assume that $\neg(A_1, \ldots, A_n)$ is derivable in
  $\hEstar{D_n^+}$.  Finally, \MP{} is simulated by \hypcut{} in the
  usual way.
\end{proof}


%% file: complexity.tex

\section{Complexity of proof search}
\label{sec:complexity hyp}

One of the advantages of formal calculi is that they can often be used
to establish complexity-optimal decision procedures for the
corresponding logics via backwards proof search. In this section we
will use our hypersequent calculi to do so. Before considering the
results in detail, we note again that since all the considered logics
have standard sequent calculi, generic \textsf{PSPACE} complexity
results for all the logics follow standard backwards proof search
using these calculi. However, as established in~\cite{Vardi:1989}, in
many cases dropping the axiom \axC{} lowers the complexity of the
logic to \textsf{coNP}. Here we show how the hypersequent calculi give
rise to complexity optimal decision procedures for the logics without
\axC{}, before briefly commenting on the case where \axC{} is present.

\subsection*{Extensions without axiom \axC}

The decision procedures for the logics without the axiom \axC{}
implement backwards proof search on a polynomially bounded
nondeterministic Turing machine with universal choices to handle the
branching caused by rules with several premisses.  Since all the rules
are invertible, we can fix an order in which the rules are applied.
To prevent loops, we employ a \emph{local loop checking} strategy,
stating that a rule is not applied (bottom-up) to a hypersequent $G$,
if at least one of its premisses is trivial in the sense that each of
its components can be derived from a component of the conclusion using
only weakening and contraction. The formal definition is as follows.

\begin{definition}
  An application of a hypersequent rule with premisses $H_1,\dots,H_n$
  and conclusion $G$ satisfies the \emph{local loop checking
    condition} if for each premiss $H_i$ there exists a component
  $\Gamma \seq \Delta$ in $H_i$ such that for no component
  $\Sigma \seq \Pi$ of the conclusion $G$ we have: for all
  $A \in \Gamma$ also $A\in \Sigma$; and for all
  $\str{\Theta} \in \Gamma$ there is a $\str{\Xi} \in \Sigma$ with
  $\set{\Theta} = \set{\Xi}$; and $\set{\Delta} \subseteq \set{\Pi}$.
\end{definition}

Since the rules are cumulative, every application of a rule satisfying
the local loop checking condition adds in each of its premisses at
least one new block or formula to an existing component, or adds a new
component, which is not subsumed by a component of the conclusion. The
following proposition shows that local loop checking does not
jeopardise completeness.

\begin{proposition}
  If a hypersequent is derivable in $\HEstar$ with a derivation of
  height $n$, then it is derivable using a derivation of height $n$ in
  which every rule application satisfies the local loop checking
  condition.
\end{proposition}

\begin{proof}
  By induction on the height of the derivation. The zero-premiss rules
  trivially satisfy the local loop checking condition. If the height
  of the derivation is $n+1$, consider the bottom-most rule
  application. If it satisfies the local loop checking condition, we
  apply the induction hypothesis to each of its premisses and are
  done. Otherwise, there is a premiss such that for each of its
  components $\Gamma,\str{\Theta_1},\dots,\str{\Theta_m} \seq \Delta$
  (where $\Gamma$ does not contain any block) there is a component
  $\Sigma \seq \Pi$ of the conclusion $G$ of the derivation with
  $\set{\Gamma} \subseteq \set{\Sigma}$, and for every $i \leq m$
  there is a $\str{\Theta_i'}\in \Sigma$ with
  $\set{\Theta_i} =\set{\Theta_i'}$, and
  $\set{\Delta} \subseteq \set{\Pi}$. Using height-preserving
  admissibility of the structural rules
  (Proposition~\ref{prop:adm-struct-rules}) we thus obtain a
  derivation of $G$ of height $n$, and an appeal to the induction
  hypothesis yields a derivation of height $n$ where every rule
  application satisfies the local loop checking condition.
\end{proof}

Note that in the proof of this proposition, no new rule applications
are added to a derivation, and that the order of rule applications is
preserved in the proof of admissibility of the structural rules
(Proposition~\ref{prop:adm-struct-rules}). Hence given a derivation of
a hypersequent, we can first adjust the ordering of the rules using
invertibility, then remove all rule applications violating the local
loop checking condition. This yields completeness of proof search
under these constraints:

\begin{corollary}
  Proof search in $\HEstar$ with local loop checking and a fixed order
  on the applications of rules is complete.\qed
\end{corollary}

The proof search algorithm thus applies the rules backwards in an
arbitrary but fixed order, universally chooses one of their premisses
and then recursively checks that this premiss is derivable.  The
procedure is shown in Algorithm~\ref{alg:decision}. In order to
facilitate the countermodel construction for underivable hypersequents
in the next section, we show termination for all considered logics,
even those containing axiom \axC:
\begin{algorithm}[t]
  \KwIn{A hypersequent $\hG$ and the code of a logic $\logic$ 
  }
  \KwOut{``yes'' if $G$ is derivable in $\HL$, a
    hypersequent if it is not.}
  \BlankLine
  \uIf{there is a component $\Gamma \seq \Delta$ in $\hG$ with
    $\bot \in \Gamma$, or 
    $\top \in \Delta$, or
    $\Gamma \cap \Delta \neq \emptyset$}{return ``yes'' and halt
    \nllabel{line:initial}\;}
  \uElseIf{there is an applicable rule}{pick the first applicable
    rule\;
    universally choose a premiss $H$ of this rule application\;
  check recursively whether $H$ is derivable, output the answer and halt\;}
  \Else{return $G$ and halt\nllabel{line:reject}\;}
\caption{Decision procedure for the derivability problem in $\HEstar$
}
\label{alg:decision}
\end{algorithm}

\begin{theorem}
  Algorithm~\ref{alg:decision} terminates for all logics $\HEstar$.
\end{theorem}

\begin{proof}
  Due to the subformula property of the rules, every formula occurring
  in a hypersequent in a run of Algorithm~\ref{alg:decision} is a
  subformula of the input. Moreover, local loop checking prevents the
  duplication of formulas, blocks and components.  Thus, every
  component occurring in a run of the algorithm contains a subset of
  (occurrences of) subformulas of the input both on its antecedent and
  succedent, together with a set of blocks, each containing a subset
  of (occurrences of) subformulas of the input. Since there are only
  finitely many of these, the number of possible components is finite,
  and hence also the number of hypersequents occurring in a run of the
  algorithm. Since every rule application satisfying local loop
  checking strictly increases the size of the hypersequent, each run
  of the algorithm thus halts after finitely many steps.
\end{proof}

For the logics without axiom \axC{}, a closer analysis of the
run time yields the optimal complexity bound:

\begin{theorem}\label{thm:complexity}
  For the logics without \axC{}, Algorithm~\ref{alg:decision} runs in
  $\mathsf{coNP}$, whence for these logics the calculi provide a
  complexity-optimal decision procedure.
\end{theorem}

\begin{proof}
  Since the procedure is in the form of a non-deterministic Turing
  machine with universal choices, it suffices to show that every
  computation of this machine has polynomial length. Every application
  of a rule adds either a subformula of its conclusion or a new block
  to one of the components, or adds a new component. Due to local loop
  checking it never adds a formula, block or component which is
  already in the conclusion, so it suffices to calculate the maximal
  size of a hypersequent occurring in proof search for $G$. Suppose
  that the size of $G$ is $n$. Then both the number of components in
  $G$ and the number of subformulas of $G$ are bounded by $n$. Since
  the local loop check prevents the duplication of formulas, each
  component contains at most $n$ formulas in the antecedent and $n$
  formulas in the succedent. Moreover, since we only consider logics
  without the axiom \axC, every newly created block contains exactly
  one formula. Again, due to the local loop checking condition no
  block is duplicated, so every component contains at most $n$
  blocks. Thus every component has size at most $3n$. The procedure
  creates new components from a block and a formula of an already
  existing component using one of the rules \rbox\ and \rboxm, or from
  $\ell$ components using one of the rules \rulep, \ruled, \ruledaux,
  $\mathsf{D}_\ell^+$, with $\ell \leq k$ for a fixed $k$ depending on
  the logic. Hence there are at most $n^2 + k \cdot n^k$ many
  different components which can be created without violating the
  local loop checking condition. Thus every hypersequent occurring in
  the proof contains at most $n+n^2+k \cdot n^k$ many components, each
  of size at most $3n$, giving a total size and thus running time of
  $\mathcal{O}(n^3)$, resp. $\mathcal{O}(n^{k+1})$ for $k > 2$.
\end{proof}

As noted above, Algorithm~\ref{alg:decision} works properly also for
logics with the axiom \axC{}, ensuring in particular termination.
However, hypersequents occurring in a proof of $\hH$ can be
exponentially large with respect to the size of $\hH$.  This is due to
the presence of the rule \rulec\ that, given $n$ formulas
$\Box A_1,\,\ldots,\, \Box A_n$, allows one to build a block for every
subset of $\{A_1, \ldots, A_n\}$.  In this respect, this decision
procedure does not match the \textsf{PSPACE} complexity upper bound
established for these systems by Vardi~\cite{Vardi:1989}. However,
this is not really unexpected, since one of the main appeals of the
hypersequent calculi is that they can be used to directly construct
countermodels for unprovable hypersequents, and in some logics with
\axC{} it is possible to force exponentially large countermodels, in
particular in normal modal logic $\K$~\cite{Blackburn:2001fk}. Hence
for these logics the hypersequents will need to be of exponential
size, suggesting that we need to modify the hypersequent calculi to
obtain complexity-optimal decision procedures.

\subsection*{Logics with axiom \axC}

In order to obtain a $\mathsf{PSPACE}$ decision procedure for logics
with axiom \axC\ we must adopt a different strategy.  Since already
the standard sequent calculi could be used to obtain
complexity-optimal decision procedures in a standard way, we only
sketch the ideas.  Instead of the rules in
Figure~\ref{fig:hypersequent rules}, we consider their
\emph{unkleene'd} -- and non-invertible -- version, \ie\ the ones with
all principal formulas and structures deleted from the premisses.  For
instance \rboxm{}, \rbox{} and \rulec{} are replaced respectively with

\bigskip
\noindent
\resizebox{\textwidth}{!}{
\ax{$\hG \hyp \Gamma \seq \Delta \! \hyp \! \Sigma \seq B$}
\uinf{$\hG \hyp \G, \str{\Sigma} \seq \Box B, \D$}
\disp
  \ 
\ax{$\hG \hyp \Gamma \seq \Delta \hyp \Sigma \seq B$ \ \  
$\{ \hG \hyp \Gamma \seq \Delta \hyp B \seq A \}_{A\in\Sigma}$}
\uinf{$\hG \hyp \G, \str{\Sigma} \seq \Box B, \D$}
\disp
  \ 
\ax{$\hG \hyp \G, \str{\Sigma,\Pi} \seq \Delta$}
\uinf{$\hG  \hyp \G, \str{\Sigma}, \str{\Pi}  \seq \Delta$}
\disp}

\bigskip

\noindent
Call the resulting calculus $\HEstar_-$.  Backwards proof search is
then implemented on an alternating Turing machine by existentially
guessing the last applied rule except for \rulen{}, and universally
checking that all of its premisses are derivable.  To ensure that
\rulen{} is applied if it is present in the system, we stipulate that
it is applied once to every component of the input, and that if the
existentially guessed rule creates a new component, the rule \rulen{}
is applied immediately afterwards to each of its premisses.  Since no
rule application keeps the principal formulas in the premisses, and
since the rule \rulen{} if present is applied exactly once to every
component, there is no need for any loop checking condition.

The calculi $\HEstar_-$ are sound and complete.  Soundness is obvious,
since we can add the missing formulas and structures and recover
derivations in $\HEstar$.  Completeness is seen easiest by simulating
the standard sequent calculi, e.g.~\cite{Lavendhomme:2000}.  We can
show that the calculi $\HEstar_-$ give a $\mathsf{PSPACE}$ upper
bound.

\begin{theorem}\label{thm:compE}
  Backwards proof search in $\HEstar_-$ is in $\mathsf{PSPACE}$.
\end{theorem}

\begin{proof}
  We need to show that every run of the procedure terminates in
  polynomial time. Assume that the size of the input is $n$. Let the
  \emph{weight} of a component in a hypersequent be the sum of the
  weights of the formulas and blocks occurring in it according to
  Definition~\ref{def:weight hyp}, and suppose that the maximal weight
  of components in the input is $w$.  Then every rule apart from
  \rulen{} decreases the weight of the component active in its
  conclusion. Moreover, a new component is only introduced in place of
  at least one subformula of the input, hence any hypersequent
  occurring in the proof search has at most $n+n$ components. The
  weight of each of these components is at most the maximal weight of
  a component of the input (plus one in the cases with
  \rulen{}). Since the rule \rulen{} is applied at most once to each
  component, it is thus applied at most $n$ times in the total proof
  search. Thus the runtime in total is $\mathcal{O}(n^2\cdot w)$,
  hence polynomial in the size of the input. Thus the procedure runs
  in alternating polynomial time, and thus in $\mathsf{PSPACE}$.
\end{proof}

Thus, the situation of logics with axiom \axC\ can be summarised as
follows.  On the one hand, we have a fully invertible calculus
$\hypcalc\ECstar$ which is terminating but not optimal.  As we shall
see in the next section, this calculus allows for direct extraction of
countermodels from single failed proofs.  On the other hand, we have a
calculus $\hypcalc\ECstar_-$ which is optimal but not invertible,
whence direct extraction of countermodels from single failed proofs is
not possible.  As for many other logics, this illustrates the
existence of a necessary trade-off between the optimal complexity of
the calculus and the direct countermodel extraction.


%% file: countermodel.tex

\section{Countermodel extraction and semantic completeness}
\label{sec:countermodels hyp}

We now prove semantic completeness of the hypersequent calculi, \ie\
every valid hypersequent is derivable.  This amounts to showing that a
non-provable hypersequent has a countermodel.  Countermodels are found
in the \bin{} semantics, as it is better suited for direct
countermodel extraction from failed proofs than the standard
semantics.  The reason is that, in order to define a standard
neighbourhood model, we need to determine \emph{exactly} the truth
sets of formulas: If we want a world $w$ to force $\Box A$, then we
have to make sure that $\ltrset A \rtrset$ belongs to $\N(w)$, thus
$\ltrset A \rtrset$ must be computed.  But this need conflicts with
the fact that failed proofs only provide \emph{partial information}.
Intuitively, countermodel extraction from a saturated hypersequent in
a proof of $\hH$ is based on the natural semantic reading according to
which every component corresponds to a world in the model, and every
formula in the antecedent (respectively in the succedent) of a
component is true (respectively false) in the corresponding world.
But, unless one resorts to a form of the analytic cut rule as
implicitly done in~\cite{Lavendhomme:2000}, it is hardly ever the case
that every subformula of $\hH$ is either in the antecedent or in the
succedent of every component, thus the failed proof does not suffice
to exactly determine the truth set of every subformula.  On the
contrary, in order for a world $w$ to force $\Box A$ in a \bin\ model
it suffices to find a suited pair $(\alpha,\beta)$ such that
$\alpha\subseteq\ltrset A \rtrset\subseteq\W\setminus\beta$.  As we
shall see, such a pair can be extracted direclty from the failed proof
even without knowing exactly the extension of $\ltrset A \rtrset$.

In order to prove semantic completeness we make use of the backwards
proof search strategy based on local loop checking already considered
in Section \ref{sec:complexity hyp} (Algorithm~\ref{alg:decision}).
This strategy amounts to considering the following notion of
saturation, stating that no bottom-up rule application is allowed to
initial sequents, and that a bottom-up application of a rule $\varR$
is not allowed to a hypersequent $\hG$ if $\hG$ already fulfills the
corresponding saturation condition $(\varR)$.

\begin{definition}[Saturated hypersequent]
  \label{def:saturated hyp}
  Let $\hH = \G_1 \seq \D_1 \hyp \ldots \hyp \G_n \seq \D_n$ be 
  a hypersequent occurring in a proof 
  for $\hH'$.
  The saturation conditions associated to each application of a rule
  of $\HEstar$ are as follows:
  $($\init$)$  $\G_i\cap\D_i = \emptyset$.
  $($\lbot$)$  $\bot\notin\G_i$.
  $($\rtop$)$ $\top\notin\D_i$. 
  $($\lto$)$  If $A\to B \in \G_i$, then $A\in\D_i$ or $B\in\G_i$.
  $($\rto$)$  If $A\to B\in\D_i$, then $A\in \G_i$ and $B\in\D_i$.
  $($\lland$)$  If $A\land B \in \G_i$, then $A\in\G_i$ and $B\in\G_i$. 
  $($\rland$)$  If $A\land B\in\D_i$, then $A\in \D_i$ or $B\in\D_i$.
  $($\lbox$)$  If $\Box A\in \G_i$, then $\str A \in \G_i$.
  $($\rbox$)$  If $\str\Sigma, \G \seq \D, \Box B$ is in $\hH$, then there is $\G' \seq \D', B$ in $\hH$
  such that $\set\Sigma\subseteq \G'$, 
  or there is $B, \G' \seq \D', A$ in $\hH$ for some $A\in\Sigma$.
  $($\rboxm$)$  If $\str\Sigma, \G \seq \D, \Box B$ is in $\hH$, then there is $\G' \seq \D', B$ in $\hH$
  such that $\set\Sigma\subseteq \G'$.
  $($\rulen$)$   $\langle \top \rangle \in \G_i$.
  $($\rulec$)$   If $\langle \Sigma \rangle, \langle \Pi \rangle \in \G_i$, then there is $\str\Omega \in \G_i$  
  such that $\set{\Sigma, \Pi} = \set\Omega$. 
  $($\rulet$)$   If $\str\Sigma \in \G_n$, then $\set\Sigma\subseteq\G_n$.  
  $($\rulep$)$   If $\G, \str\Sigma \seq \D$ is in $\hH$, then there is $\G' \seq \D'$ in $\hH$
  such that $\set\Sigma\subseteq\G'$.
  $($\ruledaux$)$   If $\G, \str\Sigma \seq \D$ is in $\hH$, then there is $\G' \seq \D'$ in $\hH$
  such that $\set\Sigma\subseteq\G'$,
  or there is $\G' \seq \D', A$ in $\hH$ for some $A\in\Sigma$.
  $($\ruled$)$   If $\G, \str\Sigma, \str\Pi \seq \D$ is in $\hH$, then there is $\G' \seq \D'$ in $\hH$
  such that $\set{\Sigma, \Pi}\subseteq\G'$,
  or there is $\G' \seq \D', A, B$ in $\hH$ for some $A\in\Sigma, B \in \Pi$.
  $($\rulednplus$)$   If $\G, \str{\Sigma_1}, \ldots, \str{\Sigma_n} \seq \D$ is in $\hH$, then there is $\G' \seq \D'$ in $\hH$
  such that $\set{\Sigma_1, \ldots, \Sigma_n}\subseteq\G'$.
  
  We say that $\hH$ is \emph{saturated} with respect to an application
  of a rule $R$ if it satisfies the saturation condition $(R)$ for
  that particular rule application, and that it is saturated with
  respect to $\HEstar$ if it is saturated with respect to all possible
  applications of any rule of $\HEstar$.
\end{definition}

\begin{proposition}
  If Algorithm~\ref{alg:decision} with input $G$ returns a hypersequent
  $H$, then $H$ is saturated and, for every component $\Gamma \seq
  \Delta$ of $G$, there is a component $\Sigma \seq \Pi$ of $H$ with
  $\set{\Gamma} \subseteq \set{\Sigma}$ and $\set{\Delta} \subseteq
  \set {\Pi}$.
\end{proposition}

\begin{proof}
  Saturation of $H$ follows from verifying that if one of the
  saturation conditions is not met, the corresponding rule can be
  applied without violating the local loop checking condition. Since
  Algorithm~\ref{alg:decision} applies all possible rules satisfying the
  local loop checking condition before halting and returning a
  hypersequent, $H$ must be saturated. The second statement follows
  from the cumulative nature of the rules.
\end{proof}

Then, given a saturated hypersequent $\hH$
we can directly construct a countermodel for $\hH$
in the \bin{} semantics
in the following way.

\begin{definition}[Countermodel construction]
  \label{def:countermodel hyp}
  Let $\hH$ be a saturated hypersequent occurring in a proof for
  $\hH'$.  Moreover, let $e: \mathbb N \longrightarrow \hH$ be an
  enumeration of the components of $\hH$.  Given $e$, we can write
  $\hH$ as $\G_1 \seq \D_1 \hyp \ldots \hyp \G_k \seq \D_k$. We call
  $k$ the {\em length} of $\hH$.  The model
  $\M = \langle \W, \N, \V \rangle$ is defined as follows:

  \begin{itemize}
  \item $\W = \{n \mid \G_n \seq \D_n \in \hH \}$.
  \item $\V(p) = \{n \mid p \in \G_n\}$.
  \item For all blocks $\str\Sigma$ appearing in a component
    $\G_m \seq \D_m$ of $\hH$,
    $\Sigma^+ = \{n \mid \set\Sigma \subseteq \G_n\}$ and
    $\Sigma^- = \{n \mid \Sigma \cap \D_n \not=\emptyset\}$.
  \item The definition of $\N$ depends whether the calculus is or not
    monotonic:
    \begin{itemize}
    \item Non-monotonic case:
      $\N(n) = \{ (\Sigma^+, \Sigma^-) \mid \str\Sigma \in \G_n \}$.
    \item Monotonic case:
      $\N(n) = \{ (\Sigma^+, \emptyset) \mid \str\Sigma \in \G_n \}$.
    \end{itemize}
  \end{itemize}
\end{definition}

\begin{lemma}\label{lemma:countermodel hyp} 
  Let $\hH = \G_1 \seq \D_1 \hyp \ldots \hyp \G_k \seq \D_k$ be a
  saturated hypersequent, and $\M$ be the model defined on the basis
  of $\hH$ as in Definition~\ref{def:countermodel hyp}.  Then for
  every $A$, $\angSigma$ and every $n\in\W$, we have:
  \begin{itemize}
  \item if $A\in\G_n$, then $\M, n\Vd A$;
  \item if $\str\Sigma\in \G_n$, then $\M, n\Vd \Box\AND\Sigma$; and
  \item if $A\in\D_n$, then $\M, n\not\Vd A$.
  \end{itemize}
  \noindent Moreover, if the proof is in calculus $\HEXstar$, then
  $\M$ is a \cX-model.
\end{lemma}

\begin{proof}
  The first claim is proved by mutual induction on $A$ and
  $\str\Sigma$.

  ($p \in \G_n$) By definition, $n \in \V(p)$. Thus $n \Vd p$.

  ($p \in \D_n$) By saturation of \init, $p \notin\G_n$.  Then
  $n \notin \V(p)$, thus $n\not\Vd p$.

  ($B \land C \in \G_n$) By saturation of \lland, $B \in \G_n$ and
  $C\in \G_n$.  Then by \ih, $n \Vd B$ and $n \Vd C$, thus
  $n \Vd B \land C$.

  ($B \land C \in \D_n$) By saturation of \rland, $B \in \D_n$ or
  $C\in \D_n$.  Then by \ih, $n \not\Vd B$ or $n \not\Vd C$, thus
  $n \not\Vd B \land C$.

  For $A= B \lor C, B \to C$, the proof is similar to the previous
  cases.

  ($\str\Sigma \in\G_n$) In the non-monotonic case we have: By
  definition $(\Sigma^+, \Sigma^-) \in \N(n)$.  We show that
  $\Sigma^+ \subseteq \ltrset \AND\Sigma \rtrset$ and
  $\Sigma^-\subseteq \ltrset \neg\AND\Sigma \rtrset$, which implies
  $n\Vd \Box\AND\Sigma$.  If $m\in\Sigma^+$, then
  $\set\Sigma\subseteq\G_m$.  By i.h. $m\Vd A$ for all $A\in\Sigma$,
  then $m\Vd \AND\Sigma$.  If $m\in\Sigma^-$, then there is
  $B \in \Sigma\cap\D_m$.  By i.h. $m\not\Vd B$, then
  $m\not\Vd \AND\Sigma$.  In the monotonic case the proof is
  analogous.

  ($\Box B \in \G_n$) By saturation of \lbox, $\str B \in\G_n$. Then
  by i.h. $n\Vd \Box B$.

  ($\Box B\in\D_n$) In the non-monotonic case, assume
  $(\alpha, \beta)\in\N(n)$.  Then there is $\str\Sigma\in\G_n$ such
  that $\Sigma^+=\alpha$ and $\Sigma^-=\beta$.  By saturation of rule
  \rbox{}, there is $m\in \W$ such that $\Sigma\subseteq\G_m$ and
  $B\in\D_m$, or there is $m\in\W$ such that
  $\Sigma\cap\D_m\not=\emptyset$ and $B\in\G_m$.  In the first case,
  $m\in\Sigma^+=\alpha$ and by \ih\ $m\not\Vd B$, thus
  $\alpha\not\subseteq\ltrset B \rtrset$.  In the second case,
  $m\in\Sigma^-=\beta$ and by \ih\ $m\Vd B$, thus
  $\beta\cap\ltrset B \rtrset\not=\emptyset$, \ie\
  $\ltrset B \rtrset\not\subseteq\W\setminus\beta$.  Therefore
  $n\not\Vd\Box B$.  The monotonic case is analogous.

  Now we prove that if the failed proof is in $\HEXstar$, then $\M$
  satisfies condition~(\cX).

  (\cM) By definition, $\beta=\0$ for every $(\alpha, \beta)\in\N(n)$.

  (\cN) By saturation of rule \rulen, $\str\top\in\G_n$ for all
  $n\in\W$, thus $(\top^+, \top^-)\in\N(n)$.  Moreover, by saturation
  of \rtop{}, $\top^-=\emptyset$.

  (\cC) Assume that $(\alpha,\beta), (\gamma,\delta) \in \N(n)$.  Then
  there are $\str\Sigma,\str\Pi \in \G_n$ such that
  $\Sigma^+ = \alpha$, $\Sigma^- = \beta$, $\Pi^+ = \gamma$ and
  $\Pi^- = \delta$.  By saturation or rule \rulec, there is
  $\str\Omega\in\G_n$ such that $\set\Omega = \set{\Sigma,\Pi}$, thus
  $(\Omega^+, \Omega^-) \in \N(n)$.  We show that $(i)$
  $\Omega^+ = \alpha\cap\gamma$ and $(ii)$
  $\Omega^- = \beta\cup\delta$.  $(i)$ $m\in\Omega^+$ iff
  $\set\Omega=\set{\Sigma,\Pi}\subseteq\G_m$ iff
  $\set\Sigma \subseteq \G_m$ and $\set\Pi \subseteq \G_m$ iff
  $m\in\Sigma^+ = \alpha$ and $m\in\Pi^+ = \gamma$ iff
  $m\in\alpha\cap\gamma$.  $(ii)$ $m\in\Omega^-$ iff
  $\Omega\cap\D_m \not=\emptyset$ iff
  $\Sigma,\Pi\cap\D_m \not=\emptyset$ iff
  $\Sigma\cap\D_m \not=\emptyset$ or $\Pi\cap\D_m \not=\emptyset$ iff
  $m\in\Sigma^- = \beta$ or $m\in\Pi^- = \delta$ iff
  $m \in\beta\cup\delta$.

  (\cT) If $(\alpha,\beta) \in \N(n)$, then there is
  $\str\Sigma \in \G_n$ such that $\Sigma^+ =\alpha$ and
  $\Sigma^- = \beta$.  By saturation of rule \hypruleT,
  $\set\Sigma\subseteq\G_n$, then $n\in\Sigma^+=\alpha$.

  (\cP) If $(\alpha,\beta) \in \N(n)$, then there is
  $\str\Sigma \in \G_n$ such that $\Sigma^+ =\alpha$ and
  $\Sigma^- = \beta$.  By saturation of rule \hypruleP, there is
  $m\in\W$ such that $\set\Sigma\subseteq\G_m$, then
  $m\in\Sigma^+=\alpha$, that is $\alpha\not=\emptyset$.

  (\cD) Assume $(\alpha,\beta), (\gamma,\delta) \in \N(n)$.  If
  $(\alpha,\beta) \not= (\gamma,\delta)$, then there are
  $\str\Sigma, \str\Pi \in \G_n$ such that
  $\Sigma^+ = \alpha , \Sigma^- = \beta, \Pi^+ = \gamma$ and
  $\Pi^- = \delta$.  If the calculus is non-monotonic, then by
  saturation of rule \ruled{} there is $m\in\W$ such that
  $\set{\Sigma,\Pi}\subseteq\G_m$ or there is $m\in\W$ such that
  $A, B\in\D_m$ for $A\in\Sigma$ and $B\in\Pi$.  In the first case,
  $\set{\Sigma}\subseteq\G_m$ and $\set{\Pi}\subseteq\G_m$, thus
  $m\in\Sigma^+=\alpha$ and $m\in\Pi^+=\gamma$, that is
  $\alpha\cap\gamma\not=\emptyset$.  In the second case,
  $m\in\Sigma^- = \beta$ and $m\in\Pi^- = \delta$, that is
  $\beta\cap\delta\not=\emptyset$.  If in contrast the calculus is
  monotonic, by saturation of \ruledm{} there is $m\in\W$ such that
  $\set{\Sigma,\Pi}\subseteq\G_m$.  Then $\set{\Sigma}\subseteq\G_m$
  and $\set{\Pi}\subseteq\G_m$, thus $m\in\Sigma^+=\alpha$ and
  $m\in\Pi^+=\gamma$, that is $\alpha\cap\gamma\not=\emptyset$.  The
  other possibility is that $(\alpha,\beta) \not= (\gamma,\delta)$.
  Then there is $\str\Sigma \in \G_n$ such that $\Sigma^+ = \alpha$
  and $\Sigma^- = \beta$.  In the non-monotonic case, by saturation of
  \ruledaux{} there is $m\in\W$ such that $\set{\Sigma}\subseteq\G_m$
  or there is $m\in\W$ such that $A\in\D_m$ for some $A\in\Sigma$.
  Then $m\in \Sigma^+ = \alpha$, that is $\alpha\not=\emptyset$, or
  $m\in=\Sigma^- = \beta$, that is $\beta\not=\emptyset$.  In the
  monotonic case we can consider saturation of \rulep{} and conclude
  that $\Sigma^+=\alpha\not=\emptyset$.

  (\cRDnplus) Assume
  $(\alpha_1, \beta_1), \ldots, (\alpha_m, \beta_m)$ be any $m \leq n$
  different bi-neigh-bourhood pairs belonging to $\N(n)$.  Then there
  are $\str{\Sigma_1}, \ldots, \str{\Sigma_m} \in \G_n$ such that
  $\Sigma_i^+ = \alpha_i$ and $\Sigma_i^- = \beta_i$ for every
  $1 \leq i \leq m$.  By saturation of rule \ruledmplus\ (that by
  definition belongs to the calculus $\hypcalc\EDnplusstar$), there is
  $\ell\in\W$ such that
  $\set{\Sigma_1,\ldots,\Sigma_m}\subseteq\G_\ell$.  Then
  $\ell\in\Sigma_1^+=\alpha_1$, \ldots, $\ell\in\Sigma_m^+=\alpha_m$,
  that is $\alpha_1 \cap \ldots \cap \alpha_m \not=\emptyset$.
\end{proof}

Observe that, since all rules are cumulative, $\M$ is also a
countermodel of the root hypersequent $\hH'$.  Moreover, since every
proof built in accordance with the strategy either provides a
derivation of the root hypersequent or contains a saturated
hypersequent, this allows us to prove the following theorem.

\begin{theorem}[Semantic completeness]\label{th:semantic completeness}
  If $\hH$ is valid in all \bin\ models for $\Estar$, then it is
  derivable in $\HEstar$.
\end{theorem}

\begin{proof}
  Assume $\hH$ not derivable in $\hypcalc\Estar$.  Then there is a
  failed proof of $\hH$ in $\hypcalc\Estar$ containing some saturated
  hypersequent $\hH'$.  By Lemma~\ref{lemma:countermodel hyp}, we can
  construct a \bin\ countermodel of $\hH'$, whence a countermodel of
  $\hH$, that satisfies all properties of \bin\ models for $\Estar$.
  Therefore $\hH$ is not valid in every \bin\ model for $\Estar$.
\end{proof}

Since the countermodels constructed for underivable hypersequents are
based on the saturated hypersequents returned by
Algorithm~\ref{alg:decision}, and since the latter are finite, we
immediately obtain the finite model property for all the logics. For
the logics without \axC{} we can further bound the \emph{size} of the
models, defined in the following way.

\begin{definition}
  The \emph{size} of a bi-neighbourhood or standard model $\M= \langle \W, \N,
  \V\rangle$ is defined as $\size{\M} := |\W| + \sum_{w \in \W} |\N(w)|$.
\end{definition}

\begin{corollary}\label{cor:polysize-bin}
  The logics without \axC{} have the \emph{polysize model property}
  wrt.\ bi-neighbourhood models, i.e., there is a polynomial $p$ such
  that if a formula $A$ of size $n$ is satisfiable, then it is satisfiable in a
  bi-neighbourhood model of size at most $p(n)$.
\end{corollary}

\begin{proof}
  Given a underivable formula of size $n$, from the proof of
  Thm.\ref{thm:complexity} we obtain that the saturated hypersequent
  used for constructing the countermodel has $\mathcal{O}(n^k)$ many
  components, each containing $\mathcal{O}(n)$ many blocks for $k$
  depending only on the logic. Since the worlds of the countermodel
  correspond to the components, and the neighbourhoods for each world
  are constructed from the blocks occurring in that component, this
  model has at most $\mathcal{O}(n^k)$ many worlds, each with a
  neighbourhood of size at most $\mathcal{O}(n)$. Hence the size of
  the model is $\mathcal{O}(n^{k+1})$.
\end{proof}

As the above construction shows, we can directly extract a \bin{}
countermodel from any failed proof.  If we want to obtain a
countermodel in the standard semantics we then need to apply the
transformations presented in Section~\ref{section:standard semantics}.
In principle, the rough transformation (Proposition~\ref{prop:rough
  transformation}) can be embedded into the countermodel construction
in order to directly construct a neighbourhood model, we just need to
modify the definition of $\N(n)$ in Definition~\ref{def:countermodel
  hyp} as follows:

\begin{center}
  $\N(n) = \{\gamma \mid \textup{there is }  \str\Sigma \in \G_n
  \textup{ such that } \Sigma^+ \subseteq \gamma \subseteq \W\setminus \Sigma^- \}$.
\end{center}

\noindent
However, in this way we might obtain a model with a larger
neighbourhood function than needed.  In contrast, there is no obvious
way to integrate the finer transformation of
Proposition~\ref{prop:finer transformation} into the countermodel
construction, since it relies on the evaluation of formulas in an
already existing model. But it does lead to smaller models:

\begin{corollary}
  The logics without \axC{} have the \emph{polysize model property}
  wrt.\ standard models.
\end{corollary}

\begin{proof}
  Given a satisfiable formula of size $n$, from
  Corollary~\ref{cor:polysize-bin} we obtain a bi-neighbourhood model
  with $\mathcal{O}(n)$ worlds. Since
  the transformation of Proposition~\ref{prop:finer transformation}
  constructs neighbourhoods from sets of truth sets of subformulas of
  the input, the size of $\Nn(w)$ is at most $n$ for each world
  $w$. Hence the total size of the standard model is polynomial in the
  size $n$ of the formula.
\end{proof}

An alternative way of obtaining countermodels in the standard
neighbourhood semantics is proposed in \cite{Lavendhomme:2000}. It
basically consists in forcing the proof search procedure to determine
exactly the truth set of each formula.  To this aim, whenever a
sequent representing a new world is created, the sequent is saturated
with respect to all disjunctions $A\lor\neg A$ such that $A$ is a
subformula of the root sequent.  This solution is equivalent to using
{\em analytic cut} and makes the proof search procedure significantly
more complex than the one given here.

Below we show some examples of countermodel extraction from failed
proofs, both in the bi-neighbourhood and in the standard neighbourhood
semantics.  The latter are obtained by applying the transformation in
Proposition~\ref{prop:finer transformation}.

\begin{example}[Proof search for axiom \axM{} in $\HE$ and countermodels]
  The following is a failed proof of $\Box (p \land q) \seq \Box p$ in
  $\HE$.

\vspace{0.4cm}

\noindent
\resizebox{\textwidth}{!}{
\ax{derivable}
\noLine
\uinf{$\str{p \land q}, \Box (p \land q) \seq \Box p \hyp p \land q \seq p$}
\ax{derivable}
\noLine
\uinf{$\ldots \hyp p \seq p \land  q, p$}
\ax{\textbf{saturated}}
\noLine
\uinf{$\str{p \land q}, \Box (p \land q) \seq \Box p \hyp p \seq p \land  q, q$}
\rlab{\rland}
\binf{$\str{p \land q}, \Box (p \land q) \seq \Box p \hyp p \seq p \land  q$}
\rlab{\rbox}
\binf{$\str{p \land q}, \Box (p \land q) \seq \Box p$}
\rlab{\lbox}
\uinf{$\Box (p \land q) \seq \Box p$}
\disp}

\vspace{0.4cm}

\noindent
\underline{Bi-neighbourhood countermodel}.  Let us consider the
following enumeration of the compontents of the saturated hypersequent
$\hH$: $1 \mapsto \str{p \land q}, \Box (p \land q) \seq \Box p$; and
$2 \mapsto p \seq p \land q, q$.  According to the construction in
Definition~\ref{def:countermodel hyp}, from $\hH$ we obtain the
following countermodel $\Mb = \langle \W, \Nb, \V\rangle$:
$\W = \{1, 2\}$.  $\V(p) = \{2\}$ and $\V(q)= \emptyset$.
$\Nb(2)=\emptyset$ and $\Nb(1)= \{(\emptyset, \{2\})\}$, as
$\Nb(1)= \{(p\land q^+, p\land q^-)\}$ and $p\land q^+ = \emptyset$,
$p \land q^-=\{2\}$.  We have $1 \Vd \Box(p \land q)$ because
$\0 \subseteq \ltrset p \land q \rtrset = \0 \subseteq \W
\setminus\{2\}$, and $1 \not\Vd \Box p$ because
$\ltrset p \rtrset = \{2\} \not\subseteq\W\setminus\{2\}$.  Then
$1 \not\Vd \Box(p\land q) \to \Box p$.

\smallskip
\noindent
\underline{Neighbourhood countermodel}.  We consider the set
$\s = \{\Box(p\land q) \to \Box p, \Box(p\land q), \Box p, p\land q,
p, q \}$ of the subformulas of $\Box(p\land q) \to \Box p$.  By
applying the transformation in Proposition~\ref{prop:finer
  transformation} to the \bin\ model $\Mb$, we obtain the standard
model $\Mn = \langle \W, \Nn, \V\rangle$, where $\W$ and $\V$ are as
in $\Mb$, and $\Nn(1) = \{ \emptyset \}$, since
$\Nn(1) = \{\ltrset p \land q \rtrset_{\Mb}\}$ and
$\ltrset p \land q \rtrset_{\Mb} = \emptyset$.
\end{example}

\begin{example}[Proof search for axiom \axK{} in $\HEC$ and countermodels]
  If Figure~\ref{fig:failed proof K hyp} we find a failed proof of
  $\Box (p \to q) \to (\Box p \to \Box q)$ in $\HEC$.  The
  countermodels are as follows.

  \smallskip
  \noindent
  \underline{Bi-neighbourhood countermodel}.  We consider the
  following enumeration of the compontents of the saturated
  hypersequent $\hH$:

  \begin{center}
    \begin{tabular}{r c l}
      $1$ & $\mapsto$ & $\Box (p \to q), \Box p,\langle p \to q \rangle, \langle p \rangle,  \langle p \to q, p \rangle \seq \Box q$. \\

      $2$ & $\mapsto$ & $q \seq p$. \\

      $3$ & $\mapsto$ & $p\to q \seq q, p$. \\
    \end{tabular}
  \end{center}

  \noindent
  According to the construction in Definition~\ref{def:countermodel
    hyp}, from $\hH$ we obtain the following countermodel
  $\Mb = \langle \W, \Nb, \V\rangle$: $\W = \{1, 2, 3\}$.
  $\V(p) = \emptyset$ and $\V(q)= \{2\}$.  $\Nb(2)=\Nb(3)=\emptyset$,
  and $\Nb(1)= \{(\emptyset, \{2, 3\}), (\{3\}, \emptyset)\}$, as
  $\Nb(1)= \{(p^+, p^-), (p\to q^+, p\to q^-), (p, p\to q^+, p, p\to
  q^-)\}$ and $p^+ = \emptyset$, $p^-=\{2, 3\}$, $p\to q^+ = \{3\}$,
  $p \to q^-=\emptyset$, $p, p\to q^+ = \emptyset$,
  $p, p \to q^-=\{2, 3\}$.

  Then we have $1 \Vd \Box(p \to q)$ because
  $\{3\} \subseteq \ltrset p \to q \rtrset = \W \subseteq
  \W\setminus\0$; and $x \Vd \Box p$ because
  $\0 \subseteq \ltrset p \rtrset = \0 \subseteq \W\setminus\{2, 3\}$;
  but $x \not\Vd \Box q$ because
  $\{3\} \not\subseteq \ltrset q \rtrset = \{2\}$ and
  $\ltrset q \rtrset = \{2\} \not\subseteq \W \setminus \{2, 3\}$,
  whence $x \not\Vd \Box (p \to q) \to (\Box p \to \Box q)$.  Observe
  that $\Mb$ is a \cC-model since
  $(\0 \cap \{3\}, \{2, 3\} \cup \0) = (\emptyset, \{2, 3\})$.

  \smallskip
  \noindent
  \underline{Neighbourhood countermodel}.  By logical equivalence we
  can restrict the considered set of formulas $\s$ to
  $\{\Box(p \to q), \Box p, \Box q, p \to q, p, q, \Box((p \to q)
  \land q), \Box(p \land q)\}$.  By the transformation in
  Proposition~\ref{prop:finer transformation}, from $\Mb$ we obtain
  the standard model $\Mn = \langle \W, \Nn, \V\rangle$, where $\W$
  and $\V$ are as in $\Mb$, and
  $\Nn(1) = \{\ltrset p \to q \rtrset_{\Mb}, \ltrset p \rtrset_{\Mb},
  \ltrset p \land q \rtrset_{\Mb}\} = \{\W, \0\}$.
\end{example}

Finally, the next example shows the need of rule \ruledaux\ for the
calculus $\hypcalc\ED$ and its non-monotonic extensions from the point
of view of the countermodel extraction.

\begin{example}[Proof search for $\neg\Box\top$ in $\hypcalc\ED$ and countermodel]
  Let us consider the following failed proof of $\Box\top \seq$ in
  $\hypcalc\ED$.

\begin{center}
\begin{small}
\ax{\textbf{saturated}}
\noLine
\uinf{$\Box\top, \str\top \seq \ \hyp \top \seq$}
\ax{}
\rlab{\rtop}
\uinf{$\Box\top, \str\top \seq \ \hyp \seq \top$}
\rlab{\ruledaux}
\binf{$\Box\top, \str\top \seq$}
\rlab{\lbox}
\uinf{$\Box\top \seq$}
\disp
\end{small}
\end{center}

\noindent
Consider the saturated hypersequent and establish
$1 \mapsto \Box\top, \str\top \seq$, and $2 \mapsto \top \seq$.  We
obtain the \bin\ countermodel $\M = \langle \W, \N, \V\rangle$, where
$\W = \{1, 2\}$; $\N(1)= \{(\top^+, \top^-)\} = \{(\{2\}, \0)\}$; and
$\N(2)=\emptyset$.  This is a \cD-model and falsifies $\neg\Box\top$,
as $1 \Vd \Box \top$.

Now imagine that the rule \ruledaux\ does not belong to the calculus
$\hypcalc\ED$.  In this case the proof would end with
$\Box\top, \str\top \seq$, as no other rule is backwards applicable to
it.  From this we would get the model
$\M' = \langle \W', \N', \V'\rangle$, where $\W' = \{1\}$ and
$\N'(1)= \{(\0, \0)\}$, which falsifies $\neg\Box\top$ but is
\emph{not} a \cD-model.
\end{example}

\begin{figure}
\resizebox{\textwidth}{!}{
\begin{tikzpicture}
\node (d) [label=90: 
\ax{derivable}
\noLine
\uinf{$\ldots  \hyp p\to q, q \seq q$}
\disp] at (3.2, 5.2) {};

\node (c) [label=90: 
\ax{derivable}
\noLine
\uinf{$\ldots \hyp q \seq p \to q$}
\disp] at (5.7, 4.6) {};

\node (b) [label=90: 
\ax{\textbf{saturated}}
\noLine
\uinf{$\Box (p \to q), \Box p,\langle p \to q \rangle, \langle p \rangle,  \langle p \to q, p \rangle \seq \Box q \hyp q \seq p \hyp p\to q \seq q, p$}
\ax{\quad}
\rlab{\lto}
\binf{$\Box (p \to q), \Box p,\langle p \to q \rangle, \langle p \rangle,  \langle p \to q, p \rangle \seq \Box q \hyp q \seq p \hyp p\to q \seq q$}
\ax{\quad}
\rlab{\rbox}
\binf{$\Box (p \to q), \Box p,\langle p \to q \rangle, \langle p \rangle,  \langle p \to q, p \rangle \seq \Box q \hyp q \seq p$}
\disp] at (0, 3.3) {};

\node (a) [label=90: 
\ax{derivable}
\noLine
\uinf{$\ldots \hyp p \to q, p \seq q$}
\ax{\qquad \qquad \qquad \qquad  \qquad \qquad \qquad \qquad \qquad}
\ax{derivable}
\noLine
\uinf{$\ldots \hyp q \seq p \to q$}
\rlab{\rbox}
\TrinaryInfC{$\Box (p \to q), \Box p, \langle p \to q \rangle, \langle p \rangle,  \langle p \to q, p \rangle \seq \Box q$}
\rlab{\rulec} 
\uinf{$\Box (p \to q), \Box p,\langle p \to q \rangle, \langle p \rangle \seq \Box q$}
\rlab{\lbox}
\uinf{$\Box (p \to q), \Box p,\langle p \to q \rangle \seq \Box q$}
\rlab{\lbox}
\uinf{$\Box (p \to q), \Box p \seq \Box q$}
\disp] at (0, 0.3) {};

\node(a') at (-0.2, 2.5) {};
\node(b') at (-0.2, 3.5) {};
\path[-] (a') edge [dashed] node {} (b');
\node(b2') at (6.1, 4) {};
\node(c') at (6.1, 4.8) {};
\path[-] (b2') edge [dashed] node {} (c');
\node(b3') at (4.1, 4.5) {};
\node(d') at (4.1, 5.4) {};
\path[-] (b3') edge [dashed] node {} (d');
\end{tikzpicture}
}
\caption{\label{fig:failed proof K hyp} Failed proof of axiom \axK\ in $\HEC$.}
\end{figure}

\subsection*{Relational countermodels for regular logics}

We now show that from failed proofs in $\HMCstar$ it is also possible
to directly extract relational countermodels of the non-derivable
formulas (cf.~Definition~\ref{def:relational models classical nnmls}).
This possibility not only makes the extraction of the relational
models more efficient (as it prevents to go through the transformation
of a previously extracted \bin\ model), but also shows the
independency of the calculus from any specific semantic choice.
Relational models are extracted from failed proofs in $\HMCstar$ as
follows.

\begin{definition}[Relational countermodel]\label{def:relational countermodel hyp}
  Let $\hH = \G_1 \seq \D_1 \hyp \ldots \hyp \G_k \seq \D_k$ be a
  saturated hypersequent occurring in a proof for $\hH'$ in
  $\HMCstar$.  For every $1 \leq n \leq k$, we say that a block
  $\str\Sigma$ is \emph{maximal} for $n$ if $\str\Sigma\in\G_n$, and
  for every $\str\Pi\in\G_n$, $\set\Pi \subseteq \set\Sigma$.  It is
  easy to see that by saturation of rule \rulec\ every component
  either contains a maximal block or does not contain any block at
  all.  On the basis of $\hH$ we define the relational model
  $\M = \langle \W, \Wimp, \R, \V \rangle$, as follows.

  \begin{itemize}
  \item $\W$, $\V$, and for every block $\str\Sigma$, $\Sigma^+$, are
    defined as in Definition~\ref{def:countermodel hyp}.

  \item $\Wimp$ is the set of worlds $n$ such that $\G_n$ does not
    contain any block.

  \item For every $n \in \Wnorm$, $\R(n) = \Sigma^+$, where
    $\str\Sigma$ is a maximal block for $n$.
  \end{itemize}
\end{definition}

Observe that if $\str\Sigma$ and $\str\Pi$ are two maximal blocks for
$n$, then $\Sigma^+ = \Pi^+$, whence $\R(n)$ is unique for every $n$.

\begin{lemma}
  Let $\hH = \G_1 \seq \D_1 \hyp \ldots \hyp \G_k \seq \D_k$ be a
  saturated hypersequent occurring in a proof for $\hH'$ in
  $\HMCstar$, and $\M$ be the model defined on the basis of $\hH$ as
  in Definition~\ref{def:relational countermodel hyp}.  Then for every
  formula $A$ and block $\str\Sigma$ we have: if $A\in\G_n$, then
  $\M, n\Vd A$; if $\str\Sigma\in \G_n$, then
  $\M, n\Vd \Box\AND\Sigma$; and if $A\in\D_n$, then $\M, n\not\Vd A$.
  Moreover, $\M$ is a relational model for $\EMC$, and if $\HMCstar$
  contains rule \rulen, then $\M$ is a standard Kripke model for
  normal modal logic $\K$.
\end{lemma}

\begin{proof}
  The truth lemma is proved by mutual induction on $A$ and
  $\str\Sigma$.  As usual we only consider modal formulas and blocks.

  ($\str\Sigma \in \G_n$) Then $n \in \Wnorm$. Moreover, given a block
  $\str\Pi$ maximal for $n$, $\set\Sigma\subseteq\set\Pi$.  We show
  that $\R(n) = \Pi^+ \subseteq \ltrset \AND \Sigma \rtrset$, which
  implies $n \Vd \Box\AND \Sigma$.  If $m \in \Pi^+$, then
  $\set\Pi \subseteq \G_m$, then for all $A \in \Pi$, $A \in \G_m$,
  and by \ih, $m \Vd A$.  Thus for all $A \in \Sigma$, $m \Vd A$, that
  is $m \Vd \AND \Sigma$.

  ($A \in \G_n$) By saturation of \lbox, $\str A \in \G_n$. Then by
  \ih, $n \Vd \Box A$.

  ($A \in \D_n$) If there is no block in $\G_n$, then $n \in \Wimp$,
  and by definition $n \not\Vd \Box A$.  Otherwise, let $\str\Sigma$
  be a maximal block for $n$.  Then by saturation of rule \rboxm\
  there is $m \in \W$ such that $\set\Sigma\subseteq\G_m$ and
  $A \in \D_m$.  Thus $m \in \Sigma^+ = \R(n)$, and by \ih,
  $m \not\Vd A$, therefore $\R(n) \not\subseteq \ltrset A \rtrset$,
  which implies $n \not\Vd \Box A$.
\end{proof}

As examples, we show failed proofs of axiom \axfour\ in $\HMC$ and
$\HMCN$ and the extracted countermodels.

\begin{example}[Proof search for axiom \axfour\ in $\HMC$ and countermodels]\label{ex: failed proof 4 in HMC}
  A failed proof of \axfour\ in $\HMC$ is as follows.

\begin{center}
\ax{\textbf{saturated}}
\noLine
\uinf{$\Box p, \str p \seq \Box\Box p \hyp p \seq \Box p$}
\rlab{\rboxm}
\uinf{$\Box p, \str p \seq \Box\Box p$}
\rlab{\lbox}
\uinf{$\Box p \seq \Box\Box p$}
\disp
\end{center}

\noindent
Let: $1 \mapsto \Box p, \str p \seq \Box\Box p$; and
$2 \mapsto p \seq \Box p$.

\smallskip
\noindent
\underline{Bi-neighbourhood countermodel}.  From
Definition~\ref{def:countermodel hyp} we obtain the following model
$\Mb = \langle \W, \N, \V \rangle$: $\W = \{1, 2\}$. $\V(p) = \{2\}$.
$\N(1) = \{(p^+, p^-)\} = \{(\{2\}, \0)\}$, and $\N(2) = \0$.  We have
$\{2\} \subseteq \ltrset p \rtrset=\{2\} \subseteq \W\setminus\0$,
then $1 \Vd \Box p$, but
$\{2\} \not\subseteq \ltrset \Box p \rtrset=\{1\}$, then
$1 \not\Vd \Box\Box p$.

\smallskip
\noindent
\underline{Relational countermodel}.  From
Definition~\ref{def:relational countermodel hyp} we obtain the
following model $\Mr = \langle \W, \Wimp, \R, \V \rangle$:
$\W = \{1, 2\}$ and $\Wimp = \{2\}$. $\V(p) = \{2\}$.
$\R(1) = p^+ = \{2\}$.  Since $2 \Vd p$ we have $1 \Vd \Box p$.
Moreover, since $2 \in \Wimp$, by definition $2 \not \Vd \Box p$, then
$1 \not \Vd \Box\Box p$.
\end{example}

\begin{example}[Proof search for axiom \axfour\ in $\HMCNT$ and
  countermodels]\label{ex: failed proof 4 in HMCNT}
  A failed proof of \axfour\ in $\HMCNT$ is as follows.

\begin{center}
\begin{small}
\ax{\textbf{saturated}}
\noLine
\uinf{$\Box p, \str p, \str \top, \str{p, \top}, p, \top \seq \Box\Box p \hyp p, \top, \str \top \seq \Box p \hyp \top, \str \top \seq p$}
\rlab{\rulen}
\uinf{$\Box p, \str p, \str \top, \str{p, \top}, p, \top \seq \Box\Box p \hyp p, \top, \str \top \seq \Box p \hyp \top \seq p$}
\rlab{\rboxm}
\uinf{$\Box p, \str p, \str \top, \str{p, \top}, p, \top \seq \Box\Box p \hyp p, \top, \str \top \seq \Box p$}
\rlab{\rulen}
\uinf{$\Box p, \str p, \str \top, \str{p, \top}, p, \top \seq \Box\Box p \hyp p, \top \seq \Box p$}
\rlab{\rboxm}
\uinf{$\Box p, \str p, \str \top, \str{p, \top}, p, \top \seq \Box\Box p$}
\rlab{\rulet}
\uinf{$\Box p, \str p, \str \top, \str{p, \top} \seq \Box\Box p$}
\rlab{\rulec}
\uinf{$\Box p, \str p, \str \top \seq \Box\Box p$}
\rlab{\rulen}
\uinf{$\Box p, \str p \seq \Box\Box p$}
\rlab{\lbox}
\uinf{$\Box p \seq \Box\Box p$}
\disp
\end{small}
\end{center}

\noindent
Let:
$1 \mapsto \Box p, \str p, \str \top, \str{p, \top}, p, \top \seq \Box\Box p$;
$2 \mapsto p, \top, \str \top \seq \Box p$;
and
$3 \mapsto \top, \str \top \seq p$.

\smallskip
\noindent
\underline{Bi-neighbourhood countermodel}.  From
Definition~\ref{def:countermodel hyp} we obtain the following model
$\Mb = \langle \W, \N, \V \rangle$: $\W = \{1, 2, 3\}$.
$\V(p) = \{1,2\}$.\\
$\N(1) = \{(p^+, p^-), (\top^+, \top^-), (p,\top^+;p, \top^-)\}
= \{(\{1,2\}, \{3\}), (\{1, 2, 3\}, \0)\}$. \\
$\N(2) = \N(3) = \{(\top^+, \top^-)\} = \{(\{1, 2, 3\}, \0)\}$.  It is
easy to see that $\Mb$ is a \cM\cC\cN\cT-model.  Moreover,
$1 \Vd \Box p$ and $1 \not\Vd \Box\Box p$, then
$1 \not\Vd \Box p \to \Box\Box p$.

\smallskip
\noindent
\underline{Relational countermodel}.  From
Definition~\ref{def:relational countermodel hyp} we obtain the
following model $\Mr = \langle \W, \Wimp, \R, \V \rangle$:
$\W = \{1, 2, 3\}$ and $\Wimp = \0$.  $\V(p) = \{1,2\}$.
$\R(1) = (p,\top)^+ = \{1, 2\}$; and
$\R(2) = \R(3) = \top^+ = \{1, 2, 3\}$.  Since $1 \Vd p$ and
$2 \Vd p$, $1 \Vd \Box p$.  But $3 \not\Vd p$, then
$2 \not\Vd \Box p$, thus $1 \not\Vd \Box\Box p$.  Then we have
$1 \not\Vd \Box p \to \Box\Box p$.  Notice that $\R$ is reflexive but
is not transitive, as $1 \R 2$, $2 \R 3$, but not $1 \R 3$.
\end{example}


%% file: translations.tex

\section{On translations for the classical cube}
\label{sec:translations}

Different proof-theoretical frameworks can be used for specifying
axiomatic systems, and there are many possible reasons for preferring
one over the other.  The best known (and maybe simplest) formalism for
analytic proof systems is Gentzen's {\em sequent
  calculus}~\cite{gentzen35}. But simplicity often implies less
comprehensiveness, and here it is not different: although being an
ideal tool for proving meta-logical properties, sequent calculus is
not expressive enough for constructing analytic calculi for many
logics of interest. Moreover, sequent rules seldom reflect the
semantic characterisation of the logic.  As a result, many new
formalisms have been proposed over the last 30 years, including
hypersequent calculi, but also {\em labelled
  calculi}~\cite{DBLP:books/daglib/0003059}. Hypersequents and
labelled sequents are very different in nature since, in the latter,
the basic objects are usually formulas of a more expressive language
reflecting the logic's semantics.  Hypersequent systems, in contrast,
are generalisations of sequent systems, carrying a more syntactic
characteristic, and are sometimes considered an {\em antagonist}
formalism w.r.t. labelled
calculi~\cite{DBLP:conf/tableaux/CiabattoniMS13}.

In the present work, we showed how hypersequents adequately reflect
the semantics of bi-neighbourhood models.  In~\cite{Dalmonte:2018} the
labelled sequent calculi $\oneEs$, were developed for all the logics
of the classical cube. These calculi also reflect the bi-neighbourhood
semantics, and are fully modular.  In this section, we will show that,
in the case of NNML, hypersequents and labels are far from being
antagonists. In fact, we show that they are strongly related, by
presenting translations between $\HEstar$
(Figure~\ref{fig:hypersequent rules}) restricted to the classical cube
and the labelled calculi $\oneEs$, presented next.

The language $\lanLS$ of labelled calculi extends $\lan$ with a set
$WL=\{x, y, z, ...\}$ of \emph{world labels}, and a set
$NL=\{a, b, c, ...\}$ of \emph{neighbourhood labels}.  We define
\emph{positive neighbourhood terms}, written $a_1\ldots a_n$, as
finite multisets of neighbourhood labels. Moreover, if $t$ is a
positive term, then $\ov t$ is a negative term. Negative terms $\ov t$
cannot be proper subterms, in particular cannot be negated.  The term
$\tau$ and its negative counterpart $\ov{\tau}$ are neighbourhood
constants.  We will represent by $\vart$ either $t$ or $\ov{t}$.

Intuitively, positive (resp.~negative) terms represent the
intersection (resp.~the union) of their constituents. Moreover, $t$
and $\ov{t}$ are the two members of a pair of neighbourhoods in
bi-neighbourhood models. Observe that the operation of overlining a
term cannot be iterated: it can be applied only once for turning a
positive term into a negative one. The operations of composition and
substitution over positive terms are defined as usual
(see~\cite{Dalmonte:2018}).

The formulas of $\lanLS$ are of the following kinds and respective
intuitive interpretation
\[\begin{array}{lclcl}
    \phi &::=& x: A & & A \mbox{ is satisfied by } x\\
         & & \mid t\ufor A & & A \mbox{ is satisfied by every world in the neighbourhood } t\\
         & & \mid \ov{t}\efor A & & A \mbox{ is satisfied by some world in the neighbourhood }\ov{t}\\
         & &  \mid x\in \vart & &  x \mbox{ is a world in the neighbourhood } \vart\\
         & & \mid t \nbr x & &  \mbox{the pair } (t,\ov{t}) \mbox{ is a bi-neighbourhood of } x. 
  \end{array}\]
\emph{Sequents} are pairs $\G  \Rightarrow \D$ of multisets of formulas of $\lanLS$.
The fully modular calculi $\oneEs$ are defined by the rules in Figure~\ref{figure:rules labelled calculi}.

\begin{figure}
\centering
\fbox{\begin{minipage}{39.7em} 
\vspace{0.1cm}
\begin{small}
\noindent
\textbf{Propositional rules}

\vspace{0.3cm}

\noindent
\begin{tabular}{l c r}
\multicolumn{3}{l}{\vspace{0.3cm}
\ax{}
\llab{\init}
\uinf{$x:p, \Gamma \Rightarrow \Delta, x:p$}
\disp
\hfill
\ax{}
\llab{\lbot}
\uinf{$x:\bot, \Gamma \Rightarrow \Delta$}
\disp
\hfill
\ax{}
\llab{\rtop}
\uinf{$\Gamma \Rightarrow \Delta, x:\top$}
\disp} \\

\vspace{0.3cm}

\ax{$\Gamma \Rightarrow \Delta, x:A$}
\ax{$x: B, \Gamma \Rightarrow \Delta$}
\llab{\lto}
\binf{$x: A\to B, \Gamma \Rightarrow \Delta$}
\disp
&&
\ax{$x:A, \Gamma \Rightarrow \Delta, x:B$}
\llab{\rto}
\uinf{$\Gamma \Rightarrow \Delta, x: A\to B$}
\disp \\

\vspace{0.3cm}

\ax{$x: A, x: B, \Gamma \Rightarrow \Delta$}
\llab{\lland}
\uinf{$x: A\land B, \Gamma \Rightarrow \Delta$}
\disp
&&
\ax{$\Gamma \Rightarrow \Delta, x:A$}
\ax{$\Gamma \Rightarrow \Delta, x:B$}
\llab{\rland}
\binf{$\Gamma \Rightarrow \Delta, x: A\land B$}
\disp \\

\vspace{0.3cm}

\ax{$x: A, \Gamma \Rightarrow \Delta$}
\ax{$x: B, \Gamma \Rightarrow \Delta$}
\llab{\llor}
\binf{$x: A\lor B, \Gamma \Rightarrow \Delta$}
\disp
&&
\ax{$\Gamma \Rightarrow \Delta, x : A, x: B$}
\llab{\rlor}
\uinf{$\Gamma \Rightarrow \Delta, x: A\lor B$}
\disp
 \\
\end{tabular}

\vspace{0.4cm}

\noindent
\textbf{Rules for the classical cube}

\vspace{0.3cm}

\noindent
\begin{tabular}{l}
\vspace{0.3cm}
\AxiomC{$a \nbr x, a \ufor A, \Gamma \Rightarrow \Delta, \n a \efor A$}
\llab{\lbox{}}
\rlab{($a!$)}
\UnaryInfC{$x:\Box A, \Gamma \Rightarrow \Delta$}
\DisplayProof  \\

\AxiomC{$t \nbr x, \Gamma \Rightarrow \Delta, x:\Box A, t \ufor A$}
\AxiomC{$t \nbr x, \n{t} \efor A, \Gamma \Rightarrow \Delta, x:\Box A$}
\llab{\rbox{}}
\BinaryInfC{$t \nbr x, \Gamma \Rightarrow \Delta, x:\Box A$}
\DisplayProof  \\
\end{tabular}

\vspace{0.4cm}

\noindent
\begin{tabular}{l c c c c r}
\vspace{0.3cm}
\AxiomC{}
\llab{\ruleM}
\UnaryInfC{$t \nbr x, y \in \n{t}, \Gamma \Rightarrow \Delta$}
\DisplayProof
&&
\AxiomC{$\tau \nbr x, \Gamma \Rightarrow \Delta$}
\llab{\ruleNtau}
\rlab{($x$ in $\G\cup\D$)}
\UnaryInfC{$\Gamma \Rightarrow \Delta$}
\DisplayProof
&&
\AxiomC{$ts \nbr x, t \nbr x, s \nbr x, \Gamma \Rightarrow \Delta$}
\llab{\ruleC}
\UnaryInfC{$t \nbr x, s \nbr x, \Gamma \Rightarrow \Delta$}
	\DisplayProof \\
\end{tabular}

\vspace{0.4cm}

\noindent
\textbf{Rules for local forcing}

\vspace{0.3cm}

\noindent
\begin{tabular}{l c c r}
\vspace{0.3cm}
\AxiomC{$x \in t, x:A, t \ufor A, \Gamma \Rightarrow \Delta$}
\llab{\lufor}
\UnaryInfC{$x \in t, t \ufor A, \Gamma \Rightarrow \Delta$}
	\DisplayProof
&\quad \ \ &&
\AxiomC{$y\in t, \Gamma \Rightarrow \Delta, y:A$}
\llab{\rufor}
\rlab{($y!$)}
\UnaryInfC{$\Gamma \Rightarrow \Delta,  t\ufor A$}
	\DisplayProof   \\

\AxiomC{$y\in\ov{t}, y:A, \Gamma \Rightarrow \Delta$}
\llab{\lefor}   
\rlab{($y!$)}
\UnaryInfC{$\ov{t} \efor A, \Gamma \Rightarrow \Delta$}
	\DisplayProof 
&&&
\AxiomC{$x\in\ov{t}, \Gamma \Rightarrow \Delta, x:A, \ov{t} \efor A$}
\llab{\refor}
\UnaryInfC{$x\in\ov{t}, \Gamma \Rightarrow \Delta, \ov{t} \efor A$}
	\DisplayProof  \\
\end{tabular}

\vspace{0.4cm}

\noindent
\textbf{Rules for neighbourhood terms}

\vspace{0.3cm}

\noindent
\begin{tabular}{l l l}
\vspace{0.5cm}
\AxiomC{$x\in\p{t}, x\in\p{s}, x\in\p{ts}, \Gamma \Rightarrow \Delta$}
\llab{\ruledec}
\UnaryInfC{$x\in\p{ts}, \Gamma \Rightarrow \Delta$}
\DisplayProof 
&&
\AxiomC{$x\in\n{t}, x\in\n{ts}, \Gamma \Rightarrow \Delta$}
\AxiomC{$x\in\n{s}, x\in\n{ts}, \Gamma \Rightarrow \Delta$}
\llab{\ruleovdec}
\BinaryInfC{$x\in\n{ts}, \Gamma \Rightarrow \Delta$} 
\DisplayProof  \\

\AxiomC{}
\llab{\ruleNovtau}
\UnaryInfC{$x \in \overline\tau, \Gamma \Rightarrow \Delta$}
\DisplayProof
 \\
\end{tabular}

\vspace{0.4cm}

\noindent
Application conditions: \hfill \quad \\
$y$ is fresh in \rufor and \lefor, $a$ is fresh in \lbox, and $x$ occurs in the conclusion of \ruleNtau. \hfill \quad  \\

\end{small}
\end{minipage}}
\caption{\label{figure:rules labelled calculi} Rules of labelled sequent calculi $\oneEs$.} 
\end{figure}

We are interested in the translation of {\em derivations} between
$\HEstar$ and $\oneEs$.  We start by explaining some choices made
thorough this work.

\paragraph{Hypersequents.} As already shown, hypersequents present an
elegant and modular solution for addressing non-normal
modalities. This is mainly due to two facts: (1) negative occurrences
of $\Box$-ed formulas are organized into blocks; and (2) components
are independent once created.  Hence proof search avoids the
non-determinism often generated by component {\em communication}
rules~\cite{Avron}, establishing a straight-forward proof-search
procedure. This is reflected in the left and right rules for the
$\Box$

\vspace{.2cm}
\noindent
\resizebox{\textwidth}{!}{
\ax{$\hG \hyp \G, \Box A, \str A \seq \D$}
\llab{\lbox}
\uinf{$\hG \hyp \G, \Box A \seq \D$}
\disp
\;
\ax{$\hG \hyp \G, \langle\Sigma\rangle \seq \Box B,\D \hyp \Sigma \seq B$}
\ax{$\{  \hG \hyp \G, \langle\Sigma\rangle \seq \Box B,\D \hyp B \seq A  \}_{A\in\Sigma}$}
\llab{\rbox} 
\binf{$\hG \hyp \G,  \langle\Sigma\rangle \seq \Box B, \D$}
\disp 
}

\vspace{.2cm}
\noindent
Reading rules from the conclusion upwards, the \lbox\ rule substitutes
a $\Box$ with a block. Blocks can then gather more formulas only by
the applications of the \rulec\ rule. Applications of the rule \rbox\
closes this proof cycle, creating new components involving only
right-boxed and blocked formulas, and immediately closing the
communication between components.

This determines a proof search procedure, where propositional rules
can be eagerly applied until only blocks remain and a
non-deterministic choice is triggered, where blocs/boxed formulas
should be combined for producing new components. The invertibility of
rules attenuates such non-determinism: allowing the generation of all
possible combinations avoids the need for backtracking.

But not only that: our calculi are greatly inspired and supported by
the choice of the {\em semantics}.

\paragraph{Bi-neighbourhood.} As pointed out in the introduction, in
the bi-neighbour-hood semantics the elements of a pair provide
positive and negative support for a modal formula. This is fully
captured by the box rules: the \lbox\ rule places formulas into fresh
neighbourhoods, the rule \rulec\ joins such formulas into
intersections of neighbourhoods and the \rbox\ rule carries the
formulas of a chosen neighbourhood together with a right-boxed formula
into a fresh world belonging to this neighbourhood.

These ideas can be also interpreted using {\em labels}.

\paragraph{Labels.} The labelled counterparts for the box rules are

\vspace{0.2cm}
\noindent
\resizebox{\textwidth}{!}{
\ax{$a \nbr x, a \ufor A, \Gamma \Rightarrow \Delta, \n a \efor A$}
\llab{\lbox{}}
\uinf{$x:\Box A, \Gamma \Rightarrow \Delta$}
\disp
\;
\AxiomC{$t \nbr x, \Gamma \Rightarrow \Delta, x:\Box A, t \ufor A$}
\AxiomC{$t \nbr x, \n{t} \efor A, \Gamma \Rightarrow \Delta, x:\Box A$}
\llab{\rbox{}}
\BinaryInfC{$t \nbr x, \Gamma \Rightarrow \Delta, x:\Box A$}
\DisplayProof}

\vspace{0.2cm}
\noindent
Starting from a labelled sequent $S$ placed in a component labelled by
a world-variable $x_1$, the \lbox\ rule over $\Box A_{ij}^1$ creates a
fresh neighbourhood-variable $a_{ij}^1$ of $x_1$, placing $A_{ij}^1$
in it. The rule \rulec\ then joins formulas
$A_{ij}^1, j=1,\ldots, s_{i1}$ into blocks
$\str {\Sigma_i^1}, i=1,\ldots, l_1$, given by the intersection of the
neighbourhoods $a_{ij}^1$, represented by
$a_i^1=a_{i1}^1\ldots a_{is_{i1}}^1$. That is, the blocks
$\str {\Sigma_i^1}$ formed from $S$ carry the information of boxed
formulas, grouped into subsets $\iset{a^1_i}{i}$, determined by
neighbourhood intersections.

The \rbox\ rule then non-deterministically choses one of such blocks
$a_i^1$ for $i\in\{1,\ldots l_1\}$ and a right-boxed formula $B$,
creating a fresh world-variable $x_2$ in $a_i^1$ and placing $B$ and
$A_{ij}^1$, for all $ j=1,\ldots, s_{i1}$, under this world. Observe
that the left and right premises of rule \rbox\ reflect the positive
and negative support for the modal formula $B$.

This strongly highlights the similarities between the hyper and label
formalisms.  We define next a translation from hypersequents to
labelled sequents.

\begin{definition}
  Let $\Sigma_i^k=\{A_{ij}^k\}$,
  $i=1,\ldots, l_k; j=1,\ldots,s_{ik}; k=1,\ldots n$ and fix a
  hypersequent enumeration (see~Definition~\ref{def:countermodel
    hyp}). The translation $\hl{\cdot}{\mathsf{a}_n}{\mathsf{x}_n}$
  from the hypersequent to the labelled languages, parametric on the
  world and neighbourhood labels $\mathsf{a}_n$ and $\mathsf{x}_n$,
  respectively, is recursively defined as

  \vspace{0.2cm}
  \noindent
  \resizebox{\textwidth}{!}{$
    \begin{array}{lcl}
\hl{\G, \str {\Sigma_i^1} \seq  \D}{\mathsf{a}_1}{\mathsf{x}_1}&=&
\iset{a_{i}^1\nbr x_1}{i}, \iset{a_{ij}^1\ufor A_{ij}^1}{ij},  x_1:\G \seq  x_1:\D,\iset{\overline{a_{ij}^1}\efor A_{ij}^1}{ij}\\
\hl{\hG \hyp \G, \str {\Sigma_i^n} \seq  \D}{\mathsf{a}_n}{\mathsf{x}_n}&=&
\hl{\hG}{\mathsf{a}_{n-1}}{\mathsf{x}_{n-1}}\otimes (x_n\in b_n,\iset{a_{i}^n\nbr x_n}{i}, \iset{a_{ij}^n\ufor A_{ij}^n}{ij}, x_n:\G \seq  \\
& & x_n:\D,\iset{\overline{a_{ij}^n}\efor A_{ij}^n}{ij})
\end{array}
$}

\noindent
where
\begin{itemize}
\item $k=1,\ldots, n$ indexes the components;
\item $i=1,\ldots, l_k$ indexes the blocks in the component $k$;
\item $j=1,\ldots,s_{ik}$ indexes the formulas in the block $i$ of the
  component $k$;
\item $\mathsf{x}_n=\{x_k\}_{1\leq k\leq n} $, where $x_k$ is a world
  variable relative to the $k$-th component;
\item $\mathsf{a}_n=\bigcup_{k=1}^{n}\iset{a_i^k}{1\leq i\leq l_k}$,
  where $\iset{a_i^k}{1\leq i\leq l_k}$ is the set of neighbourhood
  variables representing blocks in the $k$-th component,
  $a_i^k=a_{i1}^k\ldots a_{is_{ik}}^k$;
\item $b_n\in \iset{a_i^k}{i}\cup \iset{\overline{a_i^k}}{i}$ is a
  neighbourhood term, with $1\leq k\leq n-1$;
\item the operator $\otimes$ represents the concatenation of sequents
\[
(\Theta_1\seq\Upsilon_1)\otimes(\Theta_2\seq\Upsilon_2)
:=(\Theta_1,\Theta_2\seq\Upsilon_1,\Upsilon_2)
\]
\end{itemize}
\end{definition}

For readability, we will ease the notation by assuming that: the
active component in the conclusion of rule applications has label $n$;
$\str{\Sigma}$ is a block in this component with
$\Sigma=\{A_{j}\}, 1\leq j\leq s$; and $a=a_{1}\ldots a_{s}$ is the
neighbourhood variable representing $\str{\Sigma}$, where $a=\tau$ if
$\Sigma=\{\top\}$; if $b$ is the neighbourhood variable representing
$\str{\Pi}$, then $ab$ is the neighbourhood variable representing
$\str{\Sigma,\Pi}$.  Finally, we will omit the non-active formulas on
the derivations, replaced by (possibly indexed) context variables
$X,Y$.

Observe that hypersequent and labelled proofs have two important
differences: rules in $\HEstar$ are {\em kleene'd}, in the sense that
principal formulas are explicitly copied bottom-up; and $\oneEs$
introduces terms and proof-steps that have no correspondence in the
hypersequent setting.  As a result, we need to introduce some
flexibility in how contexts are related between an hypersequent proof
and the labelled proof emulating it.

Let $\varR$ be either a propositional or a left box rule in $\HEstar$,
and $H$ one of its premises. We say that $U_H$ is an {\em unkleene'd}
version of $H$ if $U_H$ coincides with $H$ but for the replication of
the principal formula, in an application of $\varR$ (see also
Section~\ref{sec:complexity hyp}). For example, in the derivation
\[
  \infer[]{\hG \hyp \G, \Box A \seq \D}{H=\hG \hyp \G, \Box A, \str A \seq \D}
\]
we have that $U_H= \hG \hyp \G,  \str A \seq \D$.

Similarly, consider an application of the $\Rbox$ rule in $\HEstar$,
with conclusion $H$ and premises $H_1,H_2^j, 1\leq j\leq s$

\vspace{0.2cm}
\noindent
\resizebox{\textwidth}{!}{$
\infer[\Rbox]{H=\hG \hyp \G,  \langle\Sigma\rangle \seq \Box B, \D}{H_1=\hG \hyp \G, \langle\Sigma\rangle \seq \Box B,\D \hyp \Sigma \seq B & 
\iset{H_2^j=\hG \hyp \G, \langle\Sigma\rangle \seq \Box B,\D \hyp B \seq A_j}{A_j\in\Sigma}}
$}

\vspace{0.2cm}
\noindent
and an application of the $\Rbox$ rule in $\oneEs$, with conclusion
$S=\hl{H}{\mathsf{a}_n}{\mathsf{x}_{n}}$ and premises $S_1,S_2$
\[\infer[\Rbox]{S=\hl{H}{\mathsf{a}_n}{\mathsf{x}_{n}}}
{S_1=\hl{H}{\mathsf{a}_n}{\mathsf{x}_n}\otimes(\;\seq   a \ufor B)&
S_2=\hl{H}{\mathsf{a}_n}{\mathsf{x}_n}\otimes(\n{a} \efor B \seq  \;)}
\]
Let  
$x_{n+1}$ be a fresh world variable. 
We call 

\vspace{0.2cm}
\noindent
\resizebox{\textwidth}{!}{
$E_{H_1}=\hl{H_1}{\mathsf{a}_n}{\mathsf{x}_{n+1}}\otimes(x_{n+1}\in a,\iset{x_{n+1}\in a_j}{j}\seq \;)
\mbox{ and }
E_{H_2^j}=\hl{H_2^j}{\mathsf{a}_n}{\mathsf{x}_{n+1}}\otimes(x_{n+1}\in \n{a},x_{n+1}\in \n{a_j} \seq \;)$}

\noindent {\em extensions} of
$\hl{H_1}{\mathsf{a}_n}{\mathsf{x}_{n+1}}$ and
$\hl{H_2^j}{\mathsf{a}_n}{\mathsf{x}_{n+1}}$,
respectively.\footnote{Here we slightly abuse the notation since
  $\mathsf{a}_{n+1}=\mathsf{a}_n$.}

The following lemma shows that unkleening and extensions do not alter
provability.

\begin{lemma}\label{lemma:equiv}
  Let $H, U_H,S_1,S_2,E_{H_1}, E_{H_2^j}$ as described above. Then
  \begin{itemize}
  \item[a.] $H$ and $U_H$ are height-preserving equivalent in
    $\HEstar$, that is, $H$ is provable with height at most $n$ in
    $\HEstar$ iff so it is $U_H$;
  \item[b.] $S_1$ (resp. $S_2$) is provable iff $E_{H_1}$ is provable
    (resp. $E_{H_2^j}$ is provable, for all $1\leq j\leq s$) in
    $\oneEs$;
  \item[c.] $E_{H_1}$ and $\hl{H_1}{\mathsf{a}_n}{\mathsf{x}_{n+1}}$
    (resp. $E_{H_2^j}$ and
    $\hl{H_2^j}{\mathsf{a}_n}{\mathsf{x}_{n+1}}$) are
    height-preserving equivalent in $\oneEs$.
  \end{itemize}
\end{lemma}

\begin{proof}
  (a) is easily proven by the usual invertibility argument. Regarding
  (b), observe that all the rules in $\oneEs$ are invertible. Hence,
  in the derivation $\pi_1$:
\vspace{0.2cm}
\noindent
\resizebox{\textwidth}{!}{$
\infer[\Rufor]{S_1=X, \iset{a_{j}\ufor A_{j}}{j}\seq a\ufor B,Y}
{\infer=[\dec]{X, \iset{a_{j}\ufor A_{j}}{j}, x_{n+1}\in a\seq x_{n+1}:B,Y}
{\infer=[\Lufor]{X,\iset{a_{j}\ufor A_{j}}{j},\iset{x_{n+1}\in a_{j}}{j}, x_{n+1}\in a\seq x_{n+1}:B,Y}
{E_{H_1}=X,\iset{a_{j}\ufor A_{j}}{j},\iset{x_{n+1}\in a_{j}}{j}, x_{n+1}\in a,\iset{x_{n+1}: A_j}{j}\seq x_{n+1}:B,Y}}} 
$}

\vspace{0.2cm}
\noindent
the sequent $S_1$
is provable iff $E_{H_1}$ is provable. Analogously for the derivation $\pi_2$:
\[
\infer[\Lefor]{S_2=X,\overline{a}\efor B\seq \iset{\overline{a_{j}}\efor A_{j}}{j},Y}
{\infer=[\ovdec]{X,x_{n+1}\in \overline{a}, x_{n+1}:B\seq \iset{\overline{a_{j}}\efor A_{j}}{j},Y}
{\infer[\Refor]{\iset{X,x_{n+1}\in \overline{a},x_{n+1}\in \overline{a_j}, x_{n+1}:B\seq  \iset{\overline{a_{j}}\efor A_{j}}{j},Y}{j}}
{\iset{E_{H_2^j}=X,x_{n+1}\in \overline{a},x_{n+1}\in \overline{a_j}, x_{n+1}:B\seq x_{n+1}:A_j,\iset{\overline{a_{j}}\efor A_{j}}{j},Y}{j}}}}
\]
Finally, for (c), assume that there is a proof $\pi$ of $E_{H_1}$ with
height $n$. Observe that the only rules that can be applied over
$x_{n+1}\in a$ and $x_{n+1}\in a_j$ in $\pi$ are \ruledec\ and \lufor,
respectively. But applying such rules would only duplicate formulas
already in $E_{H_1}$, and thus could be eliminated. Hence $\pi$ can be
transformed into a proof $\pi'$ of $E_{H_1}$ with height at most $n$,
where no rules are applied over $x_{n+1}\in a$ or $x_{n+1}\in a_j$,
and the result follows. The case for $E_{H_2^j}$ is similar.
\end{proof}

The next result establishes the relationship between $\HEstar$ and $\oneEs$.

\begin{theorem}
  Let $H$ be an hypersequent in $\HEstar$ with length $n$. The
  following are equivalent.
  \begin{itemize}
  \item[1.] $H$ is provable in $\HEstar$;
  \item[2.] $\hl{H}{\mathsf{a}_n}{\mathsf{x}_n}$ is provable in
    $\oneEs$.
  \end{itemize}
\end{theorem}    

\begin{proof}
  Consider the following translation between hypersequent {\em rule
    applications} and {\em derivations} in the labelled calculi, where
  the translation for the propositional rules is the trivial one.
  \begin{itemize}
  \item Case \lbox.
    \[
      \vcenter{\infer[]{\hG \hyp \G, \Box A \seq \D}{H=\hG \hyp \G,
          \Box A, \str A \seq \D}} \quad\leadsto\quad
      \vcenter{\infer[\Lbox]{\hl{\hG \hyp \G, \Box A \seq
            \D}{\mathsf{a}_n}{\mathsf{x}_n}}{\hl{U_H=\hG \hyp \G, \str
            A \seq \D}{\mathsf{a}_n\cup\{a\}}{\mathsf{x}_n}}}
    \]
    where $a$ is a fresh neighbourhood variable to be added, in
    $\mathsf{a}_n$, to the set of neighbourhood variables representing
    blocks in the $n$-th component.

  \item Case \rbox. 
\[
\vcenter{\infer[]{H=\hG \hyp \G,  \langle\Sigma\rangle \seq \Box B, \D}
{H_1=\hG \hyp \G, \langle\Sigma\rangle \seq \Box B, \D \hyp \Sigma \seq B &
\iset{H_2^j=  \hG \hyp \G,  \langle\Sigma\rangle \seq \Box B, \D \hyp B \seq A_j }{j}}}
\]
$\leadsto$

\vspace{0.2cm}
\noindent
\resizebox{\textwidth}{!}{$
\infer[\Rbox]{\hl{H}{\mathsf{a}_n}{\mathsf{x}_{n}}}
{\deduce{S_1=\hl{H}{\mathsf{a}_n}{\mathsf{x}_n}\otimes(\;\seq   a \ufor B)}{
\deduce{\vdots\;\pi_1}
{E_{H_1}=\hl{H_1}{\mathsf{a}_n}{\mathsf{x}_{n+1}}\otimes(x_{n+1}\in a,\iset{x_{n+1}\in a_j}{j}\seq \;)}} &
\deduce{S_2=\hl{H}{\mathsf{a}_n}{\mathsf{x}_n}\otimes(\n{a} \efor B \seq  \;)} 
{\deduce{\vdots\;\pi_2}{\iset{E_{H_2^j}=\hl{H_2^j}{\mathsf{a}_n}{\mathsf{x}_{n+1}}\otimes(x_{n+1}\in \n{a},x_{n+1}\in \n{a_j} \seq \;)}{j}}}}
$}

\vspace{0.2cm}
\noindent
where $\pi_1,\pi_2$ are the derivations in the proof
Lemma~\ref{lemma:equiv} (b).
\item Case \ruleM. Similar and simpler to the case \rbox, since the right premise in the derivation above has the proof
\[ 
\infer[\Lefor]{X,a\nbr x,\overline{a}\efor B\seq Y}
{\infer[\RLM]{X,a\nbr x,x_{n+1}\in \overline{a}, x_{n+1}:B\seq Y}{}}
\]
Hence,

\vspace{0.2cm}
\noindent
\resizebox{\textwidth}{!}{$
\vcenter{\infer[]{H=\hG \hyp \G,  \langle\Sigma\rangle \seq \Box B, \D}
{H_1=\hG \hyp \G, \langle\Sigma\rangle \seq \Box B, \D \hyp \Sigma \seq B}}
\leadsto
\vcenter{
\infer[\Rbox]{\hl{H}{\mathsf{a}_n}{\mathsf{x}_{n+1}}}
{\deduce{S_1=\hl{H}{\mathsf{a}_n}{\mathsf{x}_n}\otimes(\;\seq   a \ufor B)}{
\deduce{\vdots\;\pi_1}
{E_{H_1}=\hl{H_1}{\mathsf{a}_n}{\mathsf{x}_{n+1}}\otimes(x_{n+1}\in a,\iset{x_{n+1}\in a_j}{j}\seq \;)}}}}
$}

\item Case \ruleC.
\[
\vcenter{\infer{\hG \hyp \G, \langle\Sigma\rangle, \langle\Pi\rangle \seq \D}
{\hG \hyp \G, \langle\Sigma\rangle, \langle\Pi\rangle,  \langle\Sigma,\Pi\rangle \seq \D}}
\quad\leadsto\quad
\vcenter{\infer[\RLC]{\hl{\hG \hyp \G, \langle\Sigma\rangle, \langle\Pi\rangle \seq \D}{\mathsf{a}_n}{\mathsf{x}_n}}
{\hl{\hG \hyp \G, \langle\Sigma\rangle, \langle\Pi\rangle,  \langle\Sigma,\Pi\rangle \seq \D}{\mathsf{a}_n}{\mathsf{x}_n}}}
\]
\item Case \rulen. 
\[
\vcenter{\infer{\hG \hyp \G \seq \D}{\hG \hyp \G, \langle \top\rangle \seq \D}}
\quad\leadsto\quad
\vcenter{\infer[\RLN]{\hl{\hG \hyp \G \seq \D}{\mathsf{a}_n}{\mathsf{x}_n}}{\hl{\hG \hyp \G \seq \D}{\mathsf{a}_n}{\mathsf{x}_n}\otimes(\tau\nbr x_n\seq\; )}}
\]
Observe that
$\hl{\hG \hyp \G, \langle \top\rangle \seq
  \D}{\mathsf{a}_n}{\mathsf{x}_n}=\hl{\hG \hyp \G \seq
  \D}{\mathsf{a}_n}{\mathsf{x}_n}\otimes(\tau\nbr x_n, \tau\ufor
\top\seq \ov{\tau}\efor \top)$. But this sequent is provable iff
$\hl{\hG \hyp \G \seq \D}{\mathsf{a}_n}{\mathsf{x}_n}\otimes(\tau\nbr
x_n\seq\; )$ is provable, since $\tau\ufor \top$ can only add $x:\top$
to the right context, while $\overline{\tau}\efor \top$ can only be
triggered if $x\in \overline{\tau}$ is already in the left context for
some $x$.
\end{itemize}
Given this transformation and in the view of Lemma~\ref{lemma:equiv},
$(1)\Rightarrow(2)$ is easily proved by induction on a proof of $H$ in
$\HEstar$.

For proving $(2)\Rightarrow(1)$ observe that {\em provability} is
maintained from the end-sequent to the open leaves in the translated
derivations.  This means that choosing a formula
$\hl{H}{\mathsf{a}_n}{\mathsf{x}_n}$ to work on is equivalent to
performing all the steps of the translation given above, ending with
translated hypersequents of smaller proofs.  This is, in fact, one of
the pillars of the {\em focusing} method ~\cite{liang09tcs}.

In order to illustrate this, let $H=\hG \hyp \G,  \langle\Sigma\rangle \seq \Box B, \D$ and consider
the following derivation in the monotonic case
\[
\infer[\Rbox]{\hl{H}{\mathsf{a}_n}{\mathsf{x}_{n}}}
{\deduce{\hl{H}{\mathsf{a}_n}{\mathsf{x}_n}\otimes(\;\seq   a \ufor B)}{\pi}}
\]
where one decides to work on
$\hl{H}{\mathsf{a}_n}{\mathsf{x}_{n}}$. If $a\ufor B$ is never
principal in $\pi$, then $\pi$ acts over
$\hl{H}{\mathsf{a}_n}{\mathsf{x}_{n}}$ only and this derivation can be
substituted by
\[\infer[\Rbox]{\hl{H}{\mathsf{a}_n}{\mathsf{x}_{n}}}
{\deduce{\hl{H}{\mathsf{a}_n}{\mathsf{x}_n}\otimes(\;\seq   a \ufor B)}{
\deduce{\vdots \;\pi_1}
{\deduce{\hl{H}{\mathsf{a}_n}{\mathsf{x}_{n}}\otimes(x_{n+1}\in a,\iset{x_{n+1}\in a_j}{j}, \iset{x_{n+1}:A_j}{j}\seq x_{n+1}:B)}{\pi}}}}
\]
where $\pi_1$ is the derivation presented in the proof of
Lemma~\ref{lemma:equiv} (b).  Observe that
$\hl{H}{\mathsf{a}_n}{\mathsf{x}_{n}}\otimes(x_{n+1}\in
a,\iset{x_{n+1}\in a_j}{j}, \iset{x_{n+1}:A_j}{j}\seq x_{n+1}:B)$ is,
in fact,
$\hl{H_1}{\mathsf{a}_n}{\mathsf{x}_{n+1}}\otimes(x_{n+1}\in
a,\iset{x_{n+1}\in a_j}{j}\seq \;)$.

Suppose that $a\ufor B$ is principal at some point in $\pi$. Since
$\Rufor$ is invertible, it can be eagerly applied and $\pi$ can be
re-written as
\[
\infer[\Rbox]{\hl{H}{\mathsf{a}_n}{\mathsf{x}_{n+1}}}
{\infer[\Rufor]{\hl{H}{\mathsf{a}_n}{\mathsf{x}_n}\otimes(\;\seq   a \ufor B)}
{\deduce{\hl{H}{\mathsf{a}_n}{\mathsf{x}_n}\otimes(x_{n+1}\in a\seq   x_{n+1}:B)}{\pi'}}}
\]
where the application of the rule $\Rufor$ over $a\ufor B$ is permuted
down (and thus it does not appear in $\pi'$). This same argument can
be applied to $\dec$ and $\Lufor$ over $x_{n+1}\in a$ and
$x_{n+1}\in a_j$, respectively, obtaining the proof
\[\infer[\Rbox]{\hl{H}{\mathsf{a}_n}{\mathsf{x}_{n+1}}}
{\deduce{\hl{H}{\mathsf{a}_n}{\mathsf{x}_n}\otimes(\;\seq   a \ufor B)}{
\deduce{\vdots \;\pi_1}
{\deduce{\hl{H_1}{\mathsf{a}_n}{\mathsf{x}_{n+1}}\otimes(x_{n+1}\in a,\iset{x_{n+1}\in a_j}{j}\seq \;)}{\pi''}}}}\]
According to Lemma~\ref{lemma:equiv} (c),
$\hl{H_1}{\mathsf{a}_n}{\mathsf{x}_{n+1}}\otimes(x_{n+1}\in
a,\iset{x_{n+1}\in a_j}{j}\seq \;)$ is provable iff
$\hl{H_1}{\mathsf{a}_n}{\mathsf{x}_{n+1}}$ is provable. By inductive
hypothesis, $H_1$ is provable in $\HEstar$ with proof $\delta$. Hence
$H$ is provable with proof
\[
\infer[\MRbox]{\hG \hyp \G,  \langle\Sigma\rangle \seq \Box B, \D}
{\deduce{\hG \hyp \G, \langle\Sigma\rangle \seq \Box B, \D \hyp \Sigma \seq B}{\delta}}
\]
\end{proof}

We finish this section by illustrating the translation in the monotonic case.

\begin{example}
Consider the following derivation of the axiom $M$ in $\hypcalc\EM$
\[
\infer[\Lbox]{\Box (A\land B) \seq \Box A}
{\infer[\MRbox]{\Box (A\land B), \str{A \land B} \seq \Box A}
{\infer[\Lwedge]{\Box (A\land B), \str{A \land B} \seq \Box A \hyp A\land B \seq A}
{\infer[\Init]{H=\Box (A\land B), \str{A \land B} \seq \Box A \hyp A\land B, A, B \seq A}{}}}}
\]
This is mimicked in $\oneEs$ by

\vspace{0.2cm}
\noindent
\resizebox{\textwidth}{!}{$
\infer[\Lbox]{x_1:\Box (A\land B) \seq x_1:\Box A}
{\infer[\Rbox]{a \nbr x_1, a \ufor (A\land B)\seq x_1:\Box A,\n a \efor (A\wedge B)}
{\infer[\Rufor]{a \nbr x_1, a \ufor (A\land B) \seq  x_1:\Box A,\n a \efor (A\wedge B), a \ufor A}
{\infer[\Lufor]{a \nbr x_1, a \ufor (A\land B),x_2\in a \seq  x_1:\Box A,\n a \efor (A\wedge B), x_2:A}
{\infer[\Lwedge]{a \nbr x_1, a \ufor (A\land B), x_2\in a, x_2:(A\wedge B) \seq  x_1:\Box A,\n a \efor (A\wedge B), x_2: A}
{\infer[\Init]{S=a \nbr x_1, a \ufor (A\land B), x_2\in a, x_2:A, x_2:B \seq  x_1:\Box A,\n a \efor (A\wedge B), x_2: A}{}}}}&
\pi}}
$}

\vspace{0.2cm}
\noindent
where $\pi$ is
\[ 
\infer[\Lefor]{a \nbr x_1, \n{a} \efor A \seq  x_1:\Box A,\n a \efor (A\wedge B)}
{\infer[\RLM]{a\nbr x_1,x_{2}\in \overline{a}, x_{2}:A\seq x_1:\Box A,\n a \efor (A\wedge B)}{}}
\]
Observe that
$S=\hl{U_H}{a}{x_1x_2}\otimes(x_2\in a \seq \ )$.

\end{example}


%% file: conclusion-hypersequent-calculi-JV.tex

\section{Discussion and Conclusions} 
We have presented hypersequent calculi for the cube of classical
non-normal modal logics extended with axioms \axT, \axP, \axD, and
rules \RDnplus.  Apart from the distinction between monotonic and
non-monotonic systems, the calculi are modular.  They also have a
natural, and ``almost internal'' interpretation, as each component of
a hypersequent can be read as a formula of the language.  We have
shown that the hypersequent calculi have good structural properties,
in particular they enjoy cut elimination.  The calculi provide a
decision procedure, which is of optimal (\textsf{coNP}) complexity for
logics without \axC.  Moreover, from a failed proof we can easily
extract a countermodel (of polynomial size for logics without axiom
\axC) in the bi-neighbourhood semantics, whence by an easy
transformation also in the standard one.  Finally, the hypersequent
calculi can be embedded in the labelled calculi of
\cite{Dalmonte:2018} for the classical cube, providing thereby a kind
of ``compact encoding'' of derivations in the latter.

As we have already observed, not many works in the literature present
proof systems both allowing countermodel construction and enjoying
optimal complexity.  In this respect, the nested sequent calculus for
$\EM$ proposed in \cite{DBLP:conf/tableaux/Lellmann19} achieves both
goals it allows for both direct countermodel construction and can be
adapted for optimal complexity, similarly to what we did in
Section~\ref{sec:complexity hyp}.  However, as we explained, the
nested-sequent structure is of no help for non-monotonic
logics. Additionally, since the logics there also contain normal modal
logic $\K$, they are of \textsf{PSPACE} complexity.

Furthermore, optimal decision procedures for all logics of the
classical cube are presented in \cite{Giunchiglia}. The procedures
reduce validity/satisfiability in each modal logic to a set of SAT
problems, to be handled by a SAT solver; despite their efficiency, the
procedures provide neither ``proofs'', nor countermodels, whence
having a different aim from the calculi of this work. Our hypersequent
calculi have nonetheless an interest for automated reasoning: for
systems within the classical cube, they have been implemented in the
theorem prover HYPNO~\cite{HYPNO}.

All in all, the structure of our calculi, namely hypersequents with
blocks, provides an adequate framework for extracting countermodels
from a single failed proof, ensuring, at the same time, good
computational and structural properties, as well as modularity.  In
particular, we believe that this structure is likely the simplest and
the most adequate having these features in the non-monotonic case.

Two issues remain open: how to extend the present framework to deal
with the axioms 4, 5, B of normal modal logic, in particular in the
non-monotonic case.  Since our calculi are based on the \bin \
semantics, this investigation presupposes an extension of the
semantics itself to cover these axioms.

Another issue is the one of interpolation: \cite{Pattinson:2013}
presents a general result on uniform interpolation for rank-1 logics,
which would cover all examples considered here. However, since there
seem to be some issues with this result~\cite{Seifan:2017CALCO}, and
since the construction of the interpolants is not fully explicit
there, it is worth continuing the exploration of proof theoretic ways
of showing interpolation results. In \cite{Orlandelli:2019} a
constructive proof of Craig interpolation is provided for a good part
of the logics considered in this work, but not for non-monotonic
logics with \axC. Could our calculi be used to cover these missing
cases, perhaps using methods like those of~\cite{Kuznets:2018apal}?
We intend to investigate this issue in future work.
